\pdfoutput=1
\RequirePackage{ifpdf}
\ifpdf 
\documentclass[pdftex]{sigma}
\else
\documentclass{sigma}
\fi

\usepackage{eucal,mathrsfs,esint}

\numberwithin{equation}{section}

\newcommand{\R}{\mathbb{R}}
\newcommand{\N}{\mathbb{N}}
\newcommand{\C}{\mathbb{C}} %

\newcommand{\Z}{\mathbb{Z}}

\newcommand{\s}[1]{\CMcal{#1}}
\newcommand{\bb}[1]{\mathscr{#1}}
\newcommand{\rr}[1]{\mathfrak{#1}}
\newcommand{\n}[1]{\mathbb{#1}}

\newcommand{\nnorm}[1]{{\left\vert\kern-0.25ex\left\vert\kern-0.25ex\left\vert #1
 \right\vert\kern-0.25ex\right\vert\kern-0.25ex\right\vert}}

\newcommand{\ketbra}[2]{|#1\rangle\langle#2|}

\newcommand{\expo}[1]{\,\mathrm{e}^{#1}\,} %
\newcommand{\dd}{\,\mathrm{d}}
\newcommand{\ii}{\,\mathrm{i}\,}

\newcommand{\virg}[1]{\lq\lq#1\rq\rq} %

\DeclareMathOperator{\Tr}{Tr}

\newtheorem{Theorem}{Theorem}[section]
\newtheorem{Corollary}[Theorem]{Corollary}
\newtheorem{Lemma}[Theorem]{Lemma}
\newtheorem{Proposition}[Theorem]{Proposition}
 { \theoremstyle{definition}
\newtheorem{Definition}[Theorem]{Definition}
\newtheorem{Remark}[Theorem]{Remark} }

\begin{document}

\newcommand{\arXivNumber}{2006.06785}

\renewcommand{\thefootnote}{}

\renewcommand{\PaperNumber}{146}

\FirstPageHeading

\ShortArticleName{The Noncommutative Geometry of the Landau Hamiltonian: Metric Aspects}

\ArticleName{The Noncommutative Geometry\\ of the Landau Hamiltonian: Metric Aspects\footnote{This paper is a~contribution to the Special Issue on Noncommutative Manifolds and their Symmetries in honour of~Giovanni Landi. The full collection is available at \href{https://www.emis.de/journals/SIGMA/Landi.html}{https://www.emis.de/journals/SIGMA/Landi.html}}}

\Author{Giuseppe DE NITTIS~$^\dag$ and Maximiliano SANDOVAL~$^\ddag$}

\AuthorNameForHeading{G.~De Nittis and M.~Sandoval}

\Address{$^\dag$~Facultad de Matem\'aticas \& Instituto de F\'{\i}sica, Pontificia Universidad Cat\'olica de Chile,\\
\hphantom{$^\dag$}~Santiago, Chile}
\EmailD{\href{gidenittis@mat.uc.cl}{gidenittis@mat.uc.cl}}

\Address{$^\ddag$~Facultad de Matem\'aticas, Pontificia Universidad Cat\'olica de Chile, Santiago, Chile}
\EmailD{\href{msandova@protonmail.com}{msandova@protonmail.com}}

\ArticleDates{Received June 12, 2020, in final form December 22, 2020; Published online December 28, 2020}

\Abstract{This work provides a first step towards the construction of a noncommutative geometry for the quantum Hall effect in the continuum. Taking inspiration from the ideas developed by Bellissard during the 80's we build a spectral triple for the $C^*$-algebra of continuous magnetic operators based on a Dirac operator with compact resolvent. The metric aspects of this spectral triple are studied, and an important piece of Bellissard's theory (the so-called first Connes' formula) is proved.}

\Keywords{Landau Hamiltonian; spectral triple; Dixmier trace; first Connes' formula}

\Classification{81R60; 58B34; 81R15; 81V70}

\begin{flushright}
\begin{minipage}{50mm}
\it Dedicated to Giovanni Landi\\ for his 60th birthday
\end{minipage}
\end{flushright}

\renewcommand{\thefootnote}{\arabic{footnote}}
\setcounter{footnote}{0}

\section{Introduction}\label{sect:intro}

The study of the \emph{topological phases of matter} is undoubtedly one of the
research topics that has attracted most of the attention of the scientific
community during the last decades and it can be traced back to the discovery
of the \emph{quantum Hall effect} (QHE) in 1980.

The QHE is characterized by two peculiar and interesting phenomena: the
quantization of the transverse conductivity, which is robust with respect to
small perturbations (like \emph{disorder}), and the appearance of the
\emph{plateaux}. The mathematical literature on the QHE is endless and cannot be
reported extensively here. For this reason, and at the cost of forgetting
important contributions, we will report below only the main contributions to the
subject which are related to the aim of the present work. The topological nature
of the quantization of the Hall conductivity is understood since 1981 with the
pioneering works~\cite{laughlin-81,thouless-kohmoto-nightingale-nijs-82}, but
its modern interpretation in terms of {an} \emph{index theory} was completed
only in 1994, with the seminal paper~\cite{avron-seiler-simon-94}. The existence
of the plateaux as the consequence of the presence of disorder was
mathematically proved in 1987 for the case of the \emph{Landau Hamiltonian}
(continuous model)~\cite{kunz-87}. An improved version of this result was
obtained only in more recent
years~\cite{germinet-klein-schenker-07,germinet-klein-schenker-09} thanks to the
development of new sophisticated mathematical tools like the \emph{multiscale
 analysis}. During the 80's, Bellissard had the brilliant intuition to combine
the \emph{noncommutative geometry} (NGC), developed by Connes a few years earlier~\cite{connes-80,connes-85,connes-94}, with the study of the QHE. This
program eventually led to the groundbreaking paper~\cite{bellissard-elst-schulz-baldes-94} of
1994 where the NGC of the QHE was
rigorously established. The main advantage of Bellissard's approach is
mainly related to the synthesis power of NGC. In fact, inside this new
framework, the complete phenomenology of the QHE (quantization and plateaux) can
be explained with the use of a unified language. The NGC program of the QHE has
been pushed forward by several other authors in last years
(e.g., \cite{carey-hannabuss-mathai-06,goffeng-10,mccann-carey-96, xia-88}) and has been
extended with surprising success to the study of \emph{topological insulators}
(e.g., \cite{bellissard-elst-schulz-baldes-94,bourne-rennie-18, prodan-schulz-baldes-book}).

The formalism developed in~\cite{bellissard-elst-schulz-baldes-94} allows, in
principle, to treat both the \emph{continuous} case of magnetic Schrödinger
operators defined on $L^2\big(\R^2\big)$ and the \emph{discrete} case of tight-binding
magnetic Hamiltonians defined on $\ell^2\big(\Z^2\big)$. However, many of the crucial
results contained in~\cite{bellissard-elst-schulz-baldes-94}, like the
\emph{first} and \emph{second Connes' {formulas}}~\cite[Theorems 9 and~10]{bellissard-elst-schulz-baldes-94} are proved only for the discrete
case. The extension to the continuous case of these results is not simply a
matter of a straightforward generalization but requires the change of some
structural elements for the construction of the appropriate NGC. We will try to
clarify these points below, but first, it is worth making a~methodological remark:
in order to focus on the crucial points of the discussion, we will consider in
this work only the purely magnetic aspects of the dynamics, avoiding the
technical complications associated with the presence of disorder usually present in the matter. In fact,
once the magnetic aspects of the problem have been understood, the effects due
to the (random) electrostatic interaction with matter can be restored in the theory in a standard,
but not necessarily trivial, way. The general random case will be the subject
of a~future work.

Let us start by {explaining} Bellissard's
approach {in a nutshell}. The magnetic (non-random) $C^*$-algebra considered
in~\cite{bellissard-elst-schulz-baldes-94},
 the so called \emph{noncommutative Brillouin zone}, is the subalgebra
$\widetilde{\bb{C}}_B\subset\bb{B}\big(\ell^2\big(\Z^2\big)\big)$ generated by the discrete
 \emph{magnetic translations}.\footnote{The name magnetic translations is common in
 the condensed matter community since the works of Zak~\cite{zak1,zak2}. Mathematically, they are known as \emph{Weyl systems}.} The
subscript $B$ in the symbol $\widetilde{\bb{C}}_B$ refers to the
{strength} of a~constant magnetic field perpendicular to the plane. The $C^*$-algebra
$\widetilde{\bb{C}}_B$ is \emph{unital} and is endowed with a {natural} \emph{finite}
trace $\s{T}_B$ known as the \emph{trace per unit volume} (see~\cite{veselic-08} for
a recent review). The NGC built in~\cite{bellissard-elst-schulz-baldes-94} is
based on the \emph{spectral triple}
$\big(\widetilde{\rho},\widetilde{\s{H}}_2,\widetilde{D}\big)$ where
$\widetilde{\s{H}}_2:=\ell^2\big(\Z^2\big)\otimes\C^2$ is a separable Hilbert space,
$\widetilde{\rho}$ is the diagonal representation of $\widetilde{\bb{C}}_B$ on
$\widetilde{\s{H}}_2$ defined by $\widetilde{\rho}(A):=A\otimes{\bf 1}_2$ and
$\widetilde{D}$ is the \emph{Dirac operator} defined by
\begin{equation}\label{eq:intro_01}
\widetilde{D} := x_1 \otimes \sigma_1 + x_2 \otimes \sigma_2 = \left(\begin{matrix}
0 & x_1-\ii x_2 \\
x_1+\ii x_2 & 0\end{matrix}\right),
\end{equation}
where $x_1$, $x_2$ are the \emph{position operators} on $\ell^2\big(\Z^2\big)$, and $\sigma_1$, $\sigma_2$ are the Pauli matrices.
There are two aspects in the choice of the operator $\widetilde{D}$ which are decisive for the construction of the framework of~\cite{bellissard-elst-schulz-baldes-94}. First of all $\widetilde{D}$ has \emph{compact} resolvent
as required by the \emph{original} definition of spectral triple, as given
 in~\cite{connes-94,gracia-varilly-figueroa-01}.
 Secondly, the operator $\widetilde{D}$ relates the trace per unit volume~$\s{T}_B$ with the \emph{Dixmier trace} ${\Tr}_{{\rm Dix}}$~\cite{alberti-matthes-02, connes-94,connes-moscovici-95,gracia-varilly-figueroa-01,
 lord-sukochev-zanin-12} {via} the formula
\begin{equation}\label{eq:intro_02}
\s{T}_B(A) = \frac{1}{\pi^2} {\Tr}_{{\rm Dix}}\big(\big|\widetilde{D}_{\varepsilon}\big|^{-2} \widetilde{\rho}(A)\big) ,\qquad A\in \widetilde{\bb{C}}_B\,,
\end{equation}
where
$\big|\widetilde{D}_{\varepsilon}\big|^{-2}:=\big(\widetilde{D}^2+\varepsilon\big)^{-1}=\big(x_1^2+x_2^2+\varepsilon\big)^{-1}\otimes{\bf 1}_2$, and $\varepsilon>0$. Formula~\eqref{eq:intro_02} is proved (and used) only
indirectly in~\cite{bellissard-elst-schulz-baldes-94} but is explicitly given in~\cite[p.~104]{Bellissard-03}. The spectral triple
$\big(\widetilde{\rho},\widetilde{\s{H}}_2,\widetilde{D}\big)$ generates a canonical
\emph{Fredholm module} with \emph{Dirac phase}
$\widetilde{F}_{\varepsilon}:=\widetilde{D}_{\varepsilon}\big|\widetilde{D}_{\varepsilon}\big|^{-1}$.
The operator $\widetilde{F}_{\varepsilon}$ is the key element for the
construction of the quantized calculus of the noncommutative Brillouin zone
$\widetilde{\bb{C}}_B$, and {the} related index theory which provides the topological
interpretation of the QHE. It is worth noting that the \virg{discrete nature} of
$\widetilde{F}_{\varepsilon}$, meaning {that there is a (discrete)} basis which diagonalizes the
Dirac phase, is a crucial ingredient in~\cite{bellissard-elst-schulz-baldes-94},
as well as in the algebraic theory of topological
insulators~\cite{bourne-carey-rennie-16,prodan-schulz-baldes-book} or in the
construction of the \emph{spectral localizer}~\cite{schulz-baldes-loring-17,schulz-baldes-loring-19,schulz-baldes-loring-20}.

The passage to the continuous in the theory developed
in~\cite{bellissard-elst-schulz-baldes-94} consists in the construction of a~suitable $C^*$-algebra $\bb{C}_B$ of magnetic operators acting on $L^2\big(\R^2\big)$
and an associated spectral triple which provides the right index theory for the
topological interpretation of the QHE. The construction of the \emph{magnetic
 algebra} $\bb{C}_B$ (still called {noncommutative Brillouin zone}
in~\cite{bellissard-elst-schulz-baldes-94}) is standard in the literature and is
related to the abstract theory of group $C^*$-algebras (see Section~\ref{sec:magn_alg} for more details). However, $\bb{C}_B$ can also be introduced
in a very concrete way. Let $\{\psi_{n,m}\}\subset L^2\big(\R^2\big)$, with
$n,m\in\N_0:=\N\cup\{0\}$ be the orthonormal basis provided be the
\emph{generalized Laguerre functions}~\eqref{eq:lag_pol}, and define the family of operators
\begin{equation}\label{eq:intro:basic_op}
\Upsilon_{j\mapsto k}\psi_{n,m} := \delta_{j,n} \psi_{k,m} ,\qquad k,j,n,m\in\N_0 .
\end{equation}
Then, $\bb{C}_B$ coincides with the algebra generated in $\bb{B}\big(L^2\big(\R^2\big)\big)$ by the family \mbox{$\big\{\Upsilon_{j\mapsto k} \,|\, (j,k)\in \N_0^2\big\}$} (Proposition~\ref{prop:struct}). The algebra $\bb{C}_B$
is \emph{non-unital} (Corollary~\ref{cor:not_unit}) and the trace per unit volu\-me~$\s{T}_B$ provides an unbounded (semi-finite) trace densely defined on $\bb{C}_B$ (see Section~\ref{sec:dix_tr} for full details). Its physical relevance for the theory of the QHE lies in the fact that the \emph{Landau Hamiltonian}~$H_B$ (defined in Section~\ref{sec:landau_ham}) is \emph{affiliated} to $\bb{C}_B$ (Proposition~\ref{prop:affil}). Therefore, the missing ingredient for the construction of a NGC of the QHE in the continuum is an appropriate Dirac operator.

In principle the operator $\widetilde{D}$ defined by~\eqref{eq:intro_01} makes
sense also in $L^2\big(\R^2\big)\otimes\C^2$, but in this situation, its resolvent is not
compact. This is not an insurmountable problem since it is possible to define spectral triples with the weakest condition that
the elements of the $C^*$-algebra are relatively compact with respect to this
Dirac operator. Even though the terminology used in the literature is not
uniform, the situation can be summarized as follows:\footnote{The distinction
 between \emph{compact} spectral triple vs.\ \emph{locally compact} spectral triple used in this paper is borrowed from~\cite[Definition 2.2.1 and Remark~2.2.2]{sukochev-zanin-18}.
However, this terminology is not at all standard in the literature.
In our opinion, a better distinction between the two types of spectral triples
should take into account the dichotomy between the \emph{discrete} and the \emph{continuous} nature of the spectrum of the Dirac operator.} a \emph{compact} spectral
triple is a spectral triple defined by a Dirac operator with compact resolvent
according to the original definition like~\cite[Definition~9.16]{gracia-varilly-figueroa-01}; a \emph{locally compact} spectral triple is a~spectral triple defined by a Dirac operator with respect to which the elements
of the $C^*$-algebra are relatively compact~\cite[Definition~3.1]{carey-gayral-rennie-sukochev-14} or~\cite[Definition~2.2.1 and Remark~2.2.2]{sukochev-zanin-18}. For unital $C^*$-algebras {every locally compact
 spectral triple is automatically compact, in} contrast to the non-unital case, {where both types of
spectral triples are in general possible}. This fact leads to a kind of ambiguity
for the construction of a suitable spectral triple for the magnetic algebra
$\bb{C}_B$. One possibility is to consider $\bb{C}_B$ as a noncommutative locally
compact space by endowing this algebra with the geometry defined by the Dirac
operator $\widetilde{D}$. This approach has already been explored in the
literature on topological phases of matter in~\cite{bourne-rennie-18, xia-88}, and
in a more embryonic form in~\cite{avron-seiler-simon-94,goffeng-10}. In particular, the Dirac
operator $\widetilde{D}$ allows
a complete description of the QHE and (other topological phases) in the continuum~\cite{bourne-rennie-18}.
A~closely
related approach, already considered in the literature in different contexts~\cite{gayral-gracia-Bondia-iochum-all-04,langmann-95,mcdonald-sukochev-xiong-20}, consists in considering
the Dirac operator $\widehat{D}=\sum_{j=1}^2(-\ii\partial_j)\otimes\sigma_j$
which is the Fourier transform of $\widetilde{D}$. The latter choice also allows
to prove a formula of the type~\eqref{eq:intro_02} by combining the result of~\cite{azamov-mcdonald-sukochev-zanin-19} with the observation that $\big|\widehat{D}\big|^2$ is proportional to the Laplacian $\Delta$ on $\R^2$. A~completely different approach consists in constructing a compact spectral triple for the magnetic algebra~$\bb{C}_B$. To the best of our knowledge, this attempt
seems to be unexplored in the literature and is the main aim of the
present work.

Before describing the main results of this paper, it is worth to point out
{why it is relevant to consider} a compact spectral triple for the magnetic algebra~$\bb{C}_B$. In our opinion,
there are at least to two good motivations: the first has a mathematical flavor
while the second {is more physical}. Mathematically the advantage of having a Dirac operator with a compact resolvent is
related to the existence of a basis of eigenvectors which diagonalizes both
the Dirac operator and the Dirac phase. This fact might have important
implications in the problem of the computation of topological quantities. In
fact, the more advanced and recent techniques for the numerical computation of
topological quantities, like the \emph{spectral
 localizer}~\cite{schulz-baldes-loring-17,schulz-baldes-loring-19,schulz-baldes-loring-20}
or the \emph{finite volume approximation}~\cite{prodan-book}, have been
developed only for tight-binding magnetic models, and ultimately are based on
the fact that the Dirac operators~\eqref{eq:intro_01}, and their related phase,
are diagonalized on the eigenbasis of the position operators. The construction
of a compact spectral triple for the magnetic algebra $\bb{C}_B$ would provide
the right tool to extend these computational techniques to continuous models.
The physical motivation goes in the same direction and relies on the observation
that the presence of a constant magnetic field \emph{localizes} the motion of
{charged} particles in the perpendicular plane. In this sense, the existence of a magnetic field $B$
can be interpreted as an \emph{effective discretization} of the space in
magnetic domains of characteristic \emph{magnetic length}~$\ell_B$ (see Section~\ref{sec:landau_ham}).
This fact suggests that a correct (noncommutative) geometry for the magnetic
problem must depend on~$\ell_B$ and has to be a of~\virg{discrete} nature. These
requests translate mathe\-matically into the need of finding a Dirac operator that
depends on an intrinsic way on~$\ell_B$ (in contrast to~$\widetilde{D}$ which is independent of the
magnetic field~$B$) and which has compact resolvent.

Our goal is to show that the magnetic algebra $\bb{C}_B$ can be endowed with an
appropriate \emph{compact} spectral triple. The natural candidate for such a geometry is the \emph{magnetic Dirac operator} (see Section~\ref{sec:spectral-triple})
\begin{equation}\label{eq:intro_D}
D_B := \frac{1}{\sqrt{2}}\big(K_1 \otimes \gamma_1 + K_2 \otimes \gamma_2 + G_1 \otimes \gamma_3 + G_2 \otimes \gamma_4\big) .
\end{equation}
In the definition~\eqref{eq:intro_D}, $\gamma_1,\ldots,\gamma_4$ are Hermitian $4\times 4$ matrices which satisfy the fundamental anti-commutation relations of the
 Clifford algebra $C\ell_4(\C)$. The dependence
of $D_B$ {upon} the magnetic field $B$ is through the \emph{magnetic momenta} $K_1$, $K_2$ and the \emph{dual magnetic momenta}~$G_1$,~$G_2$
defined by~\eqref{eq:momenta} and~\eqref{eq:dual-momenta}, respectively. It is interesting to notice that the limit $B\to0$
(or equivalently $\ell_B\to+\infty$) is singular, showing that $D_B$ is a purely magnetic object.
As required, $D_B$ is a densely defined self-adjoint operator with compact resolvent (Proposition~\ref{prop:comp_res_D}).
Another important observation is that
\[
 D_B^2 = Q_B \otimes {\bf 1}_4 + {\bf 1} \otimes \varpi,
\]
where $\varpi:={\rm diag}(0,0,+1,-1)$ is a constant diagonal matrix and $Q_B$,
defined by~\eqref{eq:harm_osc}, is a~two-dimensional isotropic \emph{harmonic
 oscillator}. This fact relates our problem with the construction of spectral
triples based on the harmonic oscillator
\cite{gayral-wulkenhaar-13,grosse-wulkenhaar-12, vandungen-paschke-rennie-13,wulkenhaar-09}.
To define our spectral triple
for the magnetic algebra $\bb{C}_B$ we need a few other ingredients. First of
all we need the dense $\ast$-subalgebra $\bb{S}_B\subset\bb{C}_B$ defined by linear
combinations of operators $\Upsilon_{j\mapsto k}$ defined by~\eqref{eq:intro:basic_op} with
rapidly decreasing coefficients. In particular $\bb{S}_B$ turns out to be a
\emph{pre-$C^*$-algebra} (Propositions~\ref{prop:discret_sch} and~\ref{prop:struct_C_infty})
and this ensures that from $\bb{S}_B$ one can recover the
all the topological information of the magnetic algebra $\bb{C}_B$, like the
 $K$-theory~\cite[Theorema~3.44]{gracia-varilly-figueroa-01}. The second ingredient is
the separable Hilbert space $\s{H}_4:=L^2\big(\R^2\big)\otimes\C^4$ on which the algebra
$\bb{C}_B$ is represented diagonally {by the map}
$\rho\colon \bb{C}_B\to\bb{B}(\s{H}_4)$ where $\rho(A):=A\otimes{\bf 1}_4$ for every $A\in \bb{C}_B$.
The last ingredient is the self-adjoint involution $\chi:={\bf 1}\otimes(\gamma_1\gamma_2\gamma_3\gamma_4)$
which commutes with the elements of $\rho(\bb{C}_B)$ and anticommutes with $D_B$.
We are now in a position to state our first main result:
\begin{Theorem}\label{theo:main_01}
The collection $(\bb{S}_B,\s{H}_4,D_B)$, endowed with the diagonal representation $\rho$ and the involution $\chi$
defines a regular even compact spectral triple for the magnetic algebra
$\bb{C}_B$.
\end{Theorem}

The proof of Theorem~\ref{theo:main_01} is a consequence of several results
proved in Section~\ref{sec:spectral-triple}. In particular the basic axioms for
a compact spectral triple (according to~\cite[Definition~9.16]{gracia-varilly-figueroa-01}) are proved in
Propositions~\ref{prop:comp_res_D} and~\ref{prop:commut_D}. The
graded structure provided by $\chi$ shows that the spectral triple
is~\emph{even}, {which is a manifestation of the dimension of the NGC manifold described by the Dirac operator}. The \emph{regularity} property
defined, as defined in~\cite[Definition~10.10]{gracia-varilly-figueroa-01}, is proved in
Proposition~\ref{prop:commut_D-II}. We will refer to $(\bb{S}_B,\s{H}_4,D_B)$ as
the \emph{magnetic spectral triple} or the \emph{magnetic $K$-cycle}. {We
 also note that our magnetic spectral triple satisfies all the axioms for a
 noncommutative spin geometry, with the exception of the existence of a real
 structure, in a~similar fashion to the spectral triple defined
 in~\cite{bellissard-elst-schulz-baldes-94}}.

The second main result concerns the generalization of the
formula~\eqref{eq:intro_02} to continuous magnetic operators. To state the
result in a precise form, we
need to introduce the fundamental \emph{magnetic disk} $\Lambda_B:=\pi\ell_B^2$ and the
regularized resolvent
\[
 {|D_{B,\varepsilon}|^{-2}} := \big(D_B^2+\varepsilon{\bf 1}\big)^{-1} ,\qquad \varepsilon>0 .
\]
\begin{Theorem}\label{theo:main_02}
 For every $A\in \bb{L}^{1}_B\cdot \bb{L}^{2}_B$ it holds true that
\[
 \s{T}_B(A) = \frac{1}{8\Lambda_B} {\Tr}_{\rm Dix}\left(|D_{B,\varepsilon}|^{-2}\rho(A)\right)
\]
independently of $\varepsilon>0$.
\end{Theorem}

Theorem~\ref{theo:main_02} is a direct consequence of Theorem~\ref{prop:trace_dix_main res_bis} and Remark~\ref{rk:trac_unit_volII} along with the structure of the square
$D_B^2$ (cf.\ Remark~\ref{rk:trac_unit_volIII}).
The domain $\bb{L}^{1}_B\cdot \bb{L}^{2}_B$ is defined in Remark~\ref{rk:trac_unit_volII}.
Even though
it is smaller than the natural domain where the trace per unit volume $\s{T}_B$ is defined (cf.\ Proposition~\ref{prop:FSN-trace})
it contains the
pre-$C^*$-algebra $\bb{S}_B$ and this implies that the magnetic spectral triple $(\bb{S}_B,\s{H}_4,D_B)$ has \emph{spectral dimension~$2$} (Theorem~\ref{theo:spec_dim}).

For the presentation of the last main result, we need to introduce the \emph{magnetic Dirac phase}
\begin{equation*}
F_{B,\varepsilon} := \frac{D_B}{|D_{B,\varepsilon}|} ,\qquad \varepsilon>0 .
\end{equation*}
and the \emph{$($quasi-$)$differential}
\begin{equation}\label{eq:def_diff_d_int}
\dd(\rho(A)) := [F_{B,\varepsilon}, \rho(A)] ,\qquad A\in \bb{S}_B .
\end{equation}
As discussed in Section~\ref{sec:magn_fre-mod}, the magnetic Dirac phase enters
in the definition of the $2^+$-summable even \emph{$($pre-$)$Fredholm module}
$(\rho,\s{H}_4,F_{B,\varepsilon})$ (Theorem~\ref{theo:Fred_module}). It is worth mentioning
that the latter fact allows to associate to $F_{B,\varepsilon}$ a unique $K\!K$-homology
class $[F_B]\in K\!K(\bb{C}_B,\C)$ (see \cite[Chapter~VII]{blackadar-98}),
although this aspect will not be investigated further in this work. Another
important remark is that with the help of the
differential~\eqref{eq:def_diff_d_int} one can build a \emph{quantized calculus}
for the magnetic algebra $\bb{C}_B$ in the spirit of~\cite[Chapter~4]{connes-94}. This aspect will be investigated in full detail in an upcoming work.
In the present work, we provide only a preliminary, but related result, known
in~\cite{bellissard-elst-schulz-baldes-94} with the name of first Connes'
formula. For that, we need to recall that the pre-$C^*$-algebra $\bb{S}_B$ admits
two spatial derivations defined by the commutators with the position operators,
i.e.,
\begin{equation*}
\nabla_j(A) = -\ii [x_j,A] ,\qquad A\in\bb{S}_B ,\qquad j=1,2 .
\end{equation*}
A more detailed discussion about these derivations is presented in Section~\ref{sec:diff_struct}. For the moment, we only need to introduce the \emph{noncommutative gradient} $\nabla(\cdot):=(\nabla_1(\cdot),\nabla_2(\cdot))$ and the inner product
\[
\nabla(A_1)\cdot\nabla(A_2) := \sum_{j=1}^2\nabla_j(A_1)\nabla_j(A_2)
\]
for every pair $A_1,A_2\in \bb{S}_B$.
\begin{Theorem}[first Connes' formula]\label{theo:1st_Connes_intro}
For every pair $A_1,A_2\in\bb{S}_B$ the following formula holds true:
\begin{equation}\label{eq.intro_toadd}
\s{T}_B\big(\nabla(A_1)^*\cdot\nabla(A_2)\big) = \frac{1}{2\pi} {\Tr}_{{\rm Dix}}\big(\dd\big(\rho(A_1)\big)^*\dd\big(\rho(A_2)\big)\big) .
\end{equation}
\end{Theorem}

Theorem~\ref{theo:1st_Connes_intro} is a direct corollary (in fact, a restatement) of Theorem~\ref{theo:1st_Connes}.

It is important to point out that formulas of the type \eqref{eq.intro_toadd} have already been investigated in the literature in contexts not necessarily related to the QHE. For instance, in~\cite{mcdonald-sukochev-xiong-19} a similar formula has been derived for the case of a ($d$-dimensional) noncommutative torus.
The latter result provides a direct generalization of the first Connes'
formula obtained in \cite{bellissard-elst-schulz-baldes-94} for the discrete
case. Even more related to Theorem~\ref{theo:1st_Connes_intro} is the work~\cite{mcdonald-sukochev-xiong-20}. In the latter paper a formula of the type~\eqref{eq.intro_toadd} is derived for the \emph{noncommutative Euclidean space} (or \emph{Moyal plane}) which is indeed similar to the algebra $\bb{C}_B$ introduced above. The main difference between the result in~\cite{mcdonald-sukochev-xiong-20} and Theorem~\ref{theo:1st_Connes_intro} lies in the choice of the Dirac operator. In fact, in~\cite{mcdonald-sukochev-xiong-20} the authors deal with the operator $\widehat{D}$ (described above) which has a non-compact resolvent, while formula~\eqref{eq.intro_toadd} is obtained from the Dirac operator $D_B$ which has compact resolvent.

{What is still missing?} In order to complete the program of~\cite{bellissard-elst-schulz-baldes-94} with our \emph{compact} spectral triple, we need to prove
the second Connes' formula (which provides the link between the Chern character
and the Hall conductivity) and to introduce the (random) perturbation by
electrostatic potentials. The first question is related to the study of the
{quantized calculus} of the magnetic algebra $\bb{C}_B$ generated by the
differential~\eqref{eq:def_diff_d_int}. The second question consists of
replacing the magnetic algebra $\bb{C}_B$, which is a twisted group
$C^*$-algebra of $\R^2$, with a twisted $C^*$-crossed product generated by the
action of $\R^2$ on the hull of the potentials (see, e.g., \cite[Sections E and~F]{bellissard-elst-schulz-baldes-94} and references therein). Both of these aspects will be investigated in an upcoming work.

\subsection*{Structure of the paper}
Section~\ref{sec:Framework} is devoted to the construction of the magnetic
algebra $\bb{C}_B$ and its associated von Neumann algebra $\bb{M}_B$. More
precisely, we start by recalling the fundamental facts about the Landau
Hamiltonian (Section~\ref{sec:landau_ham}). Then, we review the connection
between $\bb{C}_B$ and the theory of twisted group algebras (Section~\ref{sec:landau_ham}), and the role of the Laguerre basis in the structural properties of $\bb{C}_B$ (Section~\ref{sec:C-magn_alg}). After that, we focus on
the structure of the associated von Neumann algebra $\bb{M}_B$ (Sections~\ref{sec:Hil-alg} and~\ref{sec:W*magn_alg}) in order to present in a~rigorous form the (noncommutative) integration theory for the magnetic algebra
(Section~\ref{sec:dix_tr}) and its link with the Dixmier trace
(Section~\ref{sec:dix_tr_2}). Finally, the (noncommutative) differential
structure and the noncommutative Sobolev's theory for $\bb{C}_B$ are discussed
(Section~\ref{sec:diff_struct}). Section~\ref{sec:Framework} contains in large
part review material, but has the merit to collect various pieces of information
scattered in the literature in a logically structured presentation. An exception should
be made for Section~\ref{sec:dix_tr_2}, which contains new important results and in
particular Theorem~\ref{prop:trace_dix_main res_bis}.\smallskip

Section~\ref{sec:magn_spec_trip} is devoted to the metric aspects of the magnetic
spectral triple and, unlike Section~\ref{sec:Framework}, contains mainly new
material. We start with the construction of the spectral triple (Section~\ref{sec:spectral-triple}) and with the study of some important properties like
the measurability (Section~\ref{sec:mes_prop}) and the absence of real
structures (Section~\ref{sec:no real}). After that, we introduce the
magnetic Fredholm module (Section~\ref{sec:magn_fre-mod}) with its graded
structure (Section~\ref{sec:grad_struct_magn_fre-mod}). Finally, we end with the
proof of the first Connes' formula
(Section~\ref{sec:1stconn_form}).

Appendices~\ref{appA} and~\ref{appB} contain some technical results which have been separated
from the main body of the text in order to make the flow of the presentation more fluid.

\section{Background material}\label{sec:Framework}

In this section, we will describe the natural
operator algebras associated with the Landau Hamiltonian. On these algebras, we
will define a {differential structure} and an integration theory based on the Dixmier trace.
Much of the information that we will present in the next sections are borrowed from existing literature. In particular the papers~\cite{gayral-gracia-Bondia-iochum-all-04,gayral-wulkenhaar-13, gracia-varilly-88,kammerer-86,maillard-86,gracia-varilly-88-II}
deal with very related subjects.

\subsection{The Landau Hamiltonian}\label{sec:landau_ham}
 Much of the material presented in this section, including the choice of notations, is taken from~\cite[Section~3]{denittis-gomi-moscolari-19}.
Let $B \in \mathbb{R}$ be the strength of a constant magnetic field, perpendicular to the plane~$\R^2$. We are interested in the quantum dynamics of a particle of mass $m$ and negative charge $-q$. With these quantities one can define the \emph{magnetic length} $\ell_B$ and the \emph{magnetic energy} $\mathcal{E}_{B}$, as
\[
 \ell_{B} := \sqrt{\frac{c\hbar}{qB}} ,\qquad \mathcal{E}_{B} : = \frac{qB\hbar}{mc} ,
\]
where $c$ is the speed of light and $\hbar$ the Planck constant. The magnetic field $B$, through the length $\ell_{B}$, enters in the
definition of the group 2-cocycle\footnote{According to the definition used in~\cite{edwards-lewis-69}, a group 2-cocycle
with values in $\n{U}(1)$ is called a \emph{multiplier}. A~straightforward check shows that $\Phi_{B}$ meets the conditions for a (non-trivial) multiplier of $\R^2$, whenever $B\neq 0$. It is also worth noting that $\Phi_{B}$ is \emph{normalized}, meaning that $\Phi_{B}(x,-x)=1$ for all $x\in\R^2$.} on~$\R^2$
given by
\[
 \Phi_{B}(x,y) := \expo{\ii\frac{x\wedge y}{2\ell_{B}^{2}}} ,\qquad x,y \in \R^2 ,
\]
where $x\wedge y :=x_{1}y_{2} - x_{2}y_{1}$ for all $x=(x_1,x_2)$ and $y=(y_1,y_2)$.

The dynamics of a (spin-less) charged particle in the field $B$ is governed by
the \emph{Landau Hamiltonian} $H_{B}$,
defined on the Hilbert ${L}^2\big(\R^2\big)$ by
\[
 H_{B} : = \frac{\mathcal{E}_{B}}{2} \big(K_{1}^{2} + K_{2}^{2}\big) ,
\]
where the \emph{magnetic momenta} $K_1$ and $K_2$ are given by
\begin{equation}\label{eq:momenta}
 K_{1} := -\ii \ell_{B}\frac{\partial}{\partial x_{1}} - \frac{1}{2 \ell_{B}} x_{2} ,\qquad
 K_{2} := -\ii \ell_{B}\frac{\partial}{\partial x_{2}} + \frac{1}{2 \ell_{B}} x_{1} .
\end{equation}
The magnetic momenta
are essentially self-adjoint
operators with dense cores given by the
compactly supported smooth functions ${C}^\infty_{\rm c}\big(\R^2\big)$ or by the Schwartz functions
${S}\big(\R^2\big)$. As a~consequence, $H_{B}$ is essentially self-adjoint on
${C}^\infty_{\rm c}\big(\R^2\big)$ or ${S}\big(\R^2\big)$.
It is also useful to define the \emph{dual magnetic momenta}
\begin{equation}
 G_{1} := -\ii \ell_{B}\frac{\partial}{\partial x_{2}} - \frac{1}{2 \ell_{B}} x_{1} ,\qquad
 G_{2} : = -\ii \ell_{B}\frac{\partial}{\partial x_{1}} + \frac{1}{2 \ell_{B}} x_{2},\label{eq:dual-momenta}
\end{equation}
which are again essentially self-adjoint with cores ${C}^\infty_{\rm c}\big(\R^2\big)$ or ${S}\big(\R^2\big)$.
The magnetic momenta and their duals satisfy the following commutation relations
\[
 [K_{1}, K_{2}] = -\ii\mathbf{1} ,\qquad [G_{1}, G_{2}] = -\ii\mathbf{1} ,\qquad [K_{i}, G_{j}] = 0 .
\]

The magnetic momenta $K_{1}$, $K_{2}$ are the infinitesimal generators of the \emph{magnetic translations} $U_B(a):=\expo{-\frac{\ii}{\ell_B}(a_1K_1+a_2K_2)}$
for all $a:=(a_1,a_2)\in\R^2$. A direct computation provides
\[
(U_B(a)\psi)(x) = \Phi_{B}(a,x) \psi(x-a), \qquad a \in \R^2 ,\qquad \psi \in L^2\big(\R^2\big) .
\]
In the same way, $G_{1}$, $G_{2}$ are the infinitesimal generators of the \emph{dual
 magnetic translations} $V_B(a):=\expo{-\frac{\ii}{\ell_B}(a_1G_2+a_2G_1)}$, and in this case one has
\[
(V_B(a)\psi)(x) = \Phi_{B}(x,a) \psi(x-a), \qquad a \in \R^2 ,\quad \psi \in L^2\big(\R^2\big) .
 \]
The maps $U_B,V_B \colon \R^2 \to \bb{B}\big(L^2\big(\R^2\big)\big)$ define two mutually commuting projective unitary representations of~$\R^2$, with respective group 2-cocycles $\Phi_{B}$ and $\Phi_{B}^{-1}$. More explicitly,
 \begin{gather*}
 U_B(a)U_B(b) = \Phi_{B}(a,b) U_B(a+b),\\
 V_B(a)V_B(b) = \Phi_{B}(b,a) V_B(a+b),\qquad a,b \in \R^2 .
\end{gather*}

Let us introduce the \emph{creation} and \emph{annihilation} operators
\[
 \mathfrak{a}^\pm := \frac{1}{\sqrt{2}}(K_{1} \pm \ii K_{2}), \qquad \mathfrak{b}^\pm : = -\frac{1}{\sqrt{2}}(G_{1} \pm \ii G_{2}) .
\]
The operators $\rr{a}^{\pm}$ and $\rr{b}^{\pm}$, are closable operators initially defined on the dense domain ${C}^{\infty}_c\big(\R^2\big)\subset {S}\big(\R^2\big)$. Moreover, $\rr{a}^{-}$ and $\rr{b}^{-}$ are the adjoint to $\rr{a}^{+}$ and $\rr{b}^{+}$, respectively. These operators meet the following canonical commutation relations (CCR)
\begin{gather*}
\big[\mathfrak{a}^-,\mathfrak{a}^{+}\big] = \mathbf{1} = \big[\mathfrak{b}^-, \mathfrak{b}^{+}\big] ,\qquad
\big[\mathfrak{a}^\pm,\mathfrak{b}^\pm\big] = 0 = \big[\mathfrak{a}^\pm,\mathfrak{b}^{\mp}\big] .
\end{gather*}
Let $\psi_{0}\in {S}\big(\R^2\big)$ be the normalized function
\begin{equation*}
\psi_{0}(x) := \frac{1}{\ell_B\sqrt{2\pi}} \expo{-\frac{|x|^2}{4\ell_B^2}} .
\end{equation*}
 A direct computation shows that $\rr{a}^{-}\psi_{0}=0=\rr{b}^{-}\psi_{0}$.
Acting on $\psi_{0}$ with the creation operators one defines the system of
\emph{generalized Laguerre functions}. Let $\N_0:=\N\cup\{0\}$.
\begin{Definition}[generalized Laguerre functions]
For every pair $(n,m)\in\N_0^2$ the associated Laguerre function $\psi_{n,m}$ is defined by
 \begin{equation}\label{eq:herm2}
\psi_{n,m} := \frac{1}{\sqrt{n!m!}} (\rr{a}^{+})^{n}(\rr{b}^{+})^{m}\psi_{0} .
\end{equation}
\end{Definition}

From the definition it follows that $\psi_{0,0}=\psi_0$.
Evidently $\psi_{n,m}\in {S}\big(\R^2\big)$ for every $(n,m)\in\N^2_0$. Moreover,
it is well known that
$\big\{\psi_{n,m} \,|\, (n,m)\in\N^2_0\big\}$ is a~complete orthonormal system in~$L^2\big(\R^2\big)$ (see, e.g., \cite[Theorem~1.3.2]{thangavelu-93}). We will refer to this system as the (magnetic) \emph{Laguerre basis}.
The set of the finite linear combinations of elements of this basis defines a
dense subspace $\s{F}_B\subset {S}\big(\R^2\big)$ which is invariant under by the action of $\rr{a}^{\pm}$ and $\rr{b}^{\pm}$.
Moreover, $\rr{a}^{\pm}$ and $\rr{b}^{\pm}$ are closable operators on~$\s{F}_B$.
The normalized functions~\eqref{eq:herm2} can be expressed
as~\cite{johnson-lippmann-49,raikov-warzel-02}
\begin{equation}\label{eq:lag_pol}
\psi_{n,m}(x) := \psi_0(x)\ \sqrt{\frac{n!}{m!}}\left[\frac{x_1+\ii x_2}{\ell_B\sqrt{2}}\right]^{m-n}L_{n}^{(m-n)}\left(\frac{|x|^2}{2\ell_B^2}\right) ,
\end{equation}
where
\[
 L_n^{(\alpha)}\left(\zeta\right) := \sum_{j=0}^{n}\frac{(\alpha+n)(\alpha+n-1)\cdots(\alpha+j+1)}{j!(n-j)!}\left(-\zeta\right)^j ,\qquad\alpha,\zeta\in \R
\]
are the \emph{generalized Laguerre polynomial} of degree $m$ (with the usual convention $0!=1$)~\cite[Section~8.97]{gradshteyn-ryzhik-07}.

The Landau Hamiltonian can be expressed in function of the {creation}-{annihilation} opera\-tors~$\rr{a}^\pm$ as follows:
\begin{equation*}
H_B = \mathcal{E}_{B}\left(\rr{a}^{+}\rr{a}^{-}+\frac{1}{2}\mathbf{1}\right) = \mathcal{E}_{B}\left(\rr{a}^{-}\rr{a}^{+}-\frac{1}{2}\mathbf{1}\right) .
\end{equation*}
The first consequence is that any Laguerre vector $\psi_{n,m}$ is an eigenvector of $H_B$. This implies that the Laguerre basis provides an orthonormal system which diagonalizes~$H_B$ according to
\[
 H_B\psi_{n,m} = \mathcal{E}_{B}\left(n+\frac{1}{2}\right)\psi_{n,m},\qquad (n,m)\in\N_0^2 .
\]
Hence, the spectrum of $H_B$ is a sequence of eigenvalues given by
\begin{equation*}
\sigma(H_B)=\left.\left\{E_{n}:=\mathcal{E}_{B}\left(n+\frac{1}{2}\right)\, \right| \, n\in\N_0\right\} ,
\end{equation*}
and $H_B$ turns out to be essentially self-adjoint also on the core~$\s{F}_B$.
The eigenvalue $E_{n}$ is called the $n$-th \emph{Landau level}.

Each Landau level is infinitely degenerate. Let $\rr{H}_n\subset L^2\big(\R^2\big)$ be the eigenspace relative to the $n$-th eigenvalue of $H_B$. Clearly, $\rr{H}_n$ is spanned by $\psi_{n,m}$ with $m\in\N_0$ and the spectral projection $\Pi_{n}\colon L^2\big(\R^2\big)\to\rr{H}_n$ is described (in Dirac notation) by
\begin{equation*}
\Pi_{n} := \sum_{m=0}^{\infty}\ketbra{\psi_{n,m}}{\psi_{n,m}} .
\end{equation*}
We refer to $\Pi_{n}$ as the $n$-th \emph{Landau projections}.
One infers from~\eqref{eq:herm2} the recursive relations
\begin{equation*}
\Pi_{n} = \frac{1}{n} \rr{a}^{+} \Pi_{n-1} \rr{a}^{-} ,\qquad \Pi_{n} = \frac{1}{n+1} \rr{a}^{-}\Pi_{n+1}\rr{a}^{+},
\end{equation*}
and after iterating one gets
\[
 \Pi_{n} = \frac{1}{n!} (\rr{a}^{+})^n \Pi_{0} (\rr{a}^{-})^n .
\]

In the following we will make use of the anti-linear operator $J$ defined on $\psi\in L^2\big(\R^2\big)$ by
\[
(J\psi)(x) := \overline{\psi^-(x)} \qquad\text{with}\quad \psi^-(x) := \psi(-x) .
\]
This is an anti-linear self-adjoint involution, i.e.,
$J^{-1} = J = J^{*}$.
Moreover, a direct check shows
\begin{alignat*}{3}
& J\mathfrak{a}^-J = -\ii\mathfrak{b}^-,\qquad && J\mathfrak{a}^{+}J = \ii \mathfrak{b}^{+},&\\
& JK_{1}J = G_{2} ,\qquad && JK_{2}J= G_{1} .&
\end{alignat*}

The magnetic momenta~\eqref{eq:momenta} and their duals~\eqref{eq:dual-momenta} can be used to describe the two-dimensional isotropic \emph{harmonic oscillator},
\begin{equation}\label{eq:harm_osc}
 Q_{B} : = \frac{1}{2}\big(K_{1}^{2} + K_{2}^{2} + G_{1}^{2} + G_{2}^{2}\big) = \mathfrak{a}^{+}\mathfrak{a}^{-} + \mathfrak{b}^{+}\mathfrak{b}^{-} + {\bf 1} .
\end{equation}
The last equality
suggests that $Q_B$ is a self-adjoint operator diagonalized by the Laguerre
basis. Indeed one has that
\[
Q_B \psi_{n,m} = (n+m+1) \psi_{n,m},\qquad (n,m)\in\N^2_0 .
\]
Then, $Q_B$ has a pure point positive spectrum given by
\[
\sigma(Q_B) = \{\lambda_j:=j+1\, |\, j\in\N_0 \}
\]
and every eigenvalue $\lambda_j$ has a finite {multiplicity} $\text{Mult}[\lambda_j]=j+1$. The eigenspace associated to~$\lambda_j$ is spanned by $\{\psi_{n,m} \,|\, n+m=j\}$. Evidently, $Q_{B}$~commutes with~$H_{B}$.

\subsection{The magnetic group algebra}\label{sec:magn_alg}
Following~\cite{edwards-lewis-69, zeller-meier-68}, we will introduce the twisted group algebra associated to the group $\R^2$ and the 2-cocycle (multiplier)~$\Phi_B$.

Let $\s{L}^1_B:=\big(L^1\big(\R^2\big), \ast,{}^\ast,\nnorm{\ }_{B,1}\big)$ be the space $L^1\big(\R^2\big)$ endowed with the (scaled) norm
\[
 \nnorm{f}_{B,1} := \frac{1}{2\pi\ell_B^2}\lVert f \rVert _{L^1}\:= \frac{1}{2\pi\ell_B^2}\int_{\R^2}\dd x\, |f(x)| ,
\]
the \emph{magnetic convolution}\footnote{Also know as \emph{twisted convolution}
or \emph{convolution gauche}.} product
\begin{equation}\label{eq:magn_conv_01}
(f * g) (x) := \frac{1}{2\pi \ell_{B}^{2}}\int_{\mathbb{R}^{2}}\dd y\, f(y) g(x-y) \Phi_{B}(y, x)
\end{equation}
and the involution
\[
 f^{*}(x) = \overline{f(-x)}
\]
for all $f,g\in L^1\big(\R^2\big)$.
With these operations $\s{L}^1_B$ becomes a Banach $\ast$-algebra~\cite[Theorem~1]{edwards-lewis-69}. In particular, it holds true that
\[
\nnorm{f * g}_{B,1} \leqslant \nnorm{f}_{B,1} \nnorm{g}_{B,1} .
\]
We will refer to $\s{L}^1_B$ as the \emph{magnetic} group algebra of $\R^2$.

\begin{Remark}[singular limit]
 It is worth noting that the whole algebraic structure of $\s{L}^1_B$ defined
 above becomes singular in the limit $B\to0$, which corresponds to $\ell_B\to\infty$. For
 this reason $B\to0$ represents a singular limit for $\s{L}^1_B$, which therefore
 deserves the name of \emph{magnetic} algebra. Since $B\neq 0$ implies that the
 multiplier $\Phi_B$ is non trivial it follows that $\s{L}^1_B$ is always
 noncommutative as a consequence of~\cite[Theorem~7]{edwards-lewis-69}.
\end{Remark}

The Banach algebra $\s{L}^1_B$ has several interesting $\ast$-subalgebras. We will
be mainly interested in the \emph{Schwartz} $\ast$-subalgebra
$\s{S}_{B}:=\s{L}^1_B\cap{S}\big(\R^2\big)$ made of Schwartz functions.\footnote{The symbol
 of intersection in the definition $\s{S}_{B}:=\s{L}^1_B\cap{S}\big(\R^2\big)$ must be
 understood as \virg{the space ${S}\big(\R^2\big)$ endowed with the algebraic structure of~$\s{L}^1_B$}.}
Let us recall that the space of Schwartz functions is topologized by the
Fr\'echet topology induced by the norms
\[
p_k(f) := \sup_{x\in\R^2}\big\{\langle x\rangle^k\,|\,\partial^\alpha f(x)| \big| |\alpha|\leqslant k \big\} ,\qquad k\in\N_0,
\]
where
\begin{equation}\label{eq:jap_brack}
\langle x\rangle := \sqrt{1+|x|^2} ,\qquad x \in \R^2
\end{equation}
is the \emph{Japanese bracket},
 $\alpha=(\alpha_1,\alpha_2)\in\N_0^2$ is a multi-index with $|\alpha|=\alpha_1+\alpha_2$
 and $\partial^\alpha$ is a short notation for $\partial^{\alpha_1}_{x_1}\partial^{\alpha_2}_{x_2}$.
 the space $\s{S}_{B}$ is closed and the algebraic operations are continuous with respect to the
Fr\'echet topology of the Schwartz functions. A proof of these well-known facts can be found in~\cite[p.~870]{gracia-varilly-88} or in~\cite[Proposition~8.11]{folland-99}. The density of~$\s{S}_{B}$ in
 $\s{L}^1_B$ follows from the density of ${S}\big(\R^2\big)$
in $L^1\big(\R^2\big)$~\cite[Corollary~4.23]{brezis-87}. In summary one has that
\begin{Proposition}\label{prop:frech-sub}
$\s{S}_{B}$ is a dense Fr\'echet $\ast$-subalgebras of~$\s{L}^1_B$.
\end{Proposition}

\begin{Remark}
 Another interesting $\ast$-subalgebra of $\s{L}^1_B$, often considered in the
 literature (see, e.g., \cite{bellissard-elst-schulz-baldes-94,lenz-99}), is made
 up of compactly supported continuous functions
 $\s{L}^1_{B,{\rm c}}:=\s{L}^1_B\cap{C}_{c}\big(\R^2\big)$. The algebraic structure is
 preserved by the fact that the magnetic twisted convolution of two compactly supported
 continuous functions is again a compactly supported continuous function. This
 fact can be proved exactly as in~\cite[Propositions~4.18 and~4.19]{brezis-87}, where the presence of the 2-cocycle $\Phi_B$ is harmless. The density follows from the density of ${C}_{c}\big(\R^2\big)$ in
 $L^1\big(\R^2\big)$~\cite[Theorem~4.12]{brezis-87}. In the following we will not make
 use of this subalgebra.
\end{Remark}

The algebra $\s{L}^1_B$ has no identity but admits approximate identities~\cite[Theorem 4]{edwards-lewis-69}. For instance, an approximate identity is given by the sequence $\{\eta_n\}_{n\in\N}\in \s{L}^1_{B,{\rm c}}\cap \s{S}_{B}$ defined by
\begin{equation*}
\eta_n(x) := Cn^2\eta(nx) ,\qquad\eta(x) := \chi_{D_1}(x)\expo{\frac{1}{|x|^2-1}} ,
\end{equation*}
where $\chi_{D_1}$ is the characteristic function of the unit ball $D_1:=\big\{x\in\R^2 \,|\, |x|\leqslant1\big\}$ and $C:=\nnorm{\eta}^{-1}_{B,1}$ is a normalization constant.

Since the Schwartz space ${S}\big(\R^2\big)$ is invariant under the action of the operators $\rr{a}^\pm$ and $\rr{b}^\pm$ one can investigate the
 relation between the magnetic twisted convolution and these operators.
\begin{Proposition}\label{prop:acvt_lad_op}
For every pair $f,g \in\s{S}_{B}$, and for every $n\in\N_0$ the following relations hold true
\begin{gather}
\big(\mathfrak{b}^{\pm}\big)^n(f*g) = f*\big[\big(\mathfrak{b}^{\pm}\big)^ng\big] ,\nonumber\\
\big(\mathfrak{a}^{\pm}\big)^n(f*g) = \big[\big(\mathfrak{a}^{\pm}\big)^nf\big]*g . \label{eq:convolution-annihilation}
\end{gather}
\end{Proposition}
\begin{proof} Integrating under the integral sign, and a suitable change of variable
provides $G_j(f*g)=f*(G_j g)$, $j=1,2$, where~$G_j$ are defined
by~\eqref{eq:dual-momenta}. Since $\mathfrak{b}^{\pm}$ are linear combinations of~$G_1$ and~$G_2$, one obtains by linearity the first equation of~\eqref{eq:convolution-annihilation}. To prove the second equality let us
observe that the change of variable $x-y\mapsto y$ in the
definition~\eqref{eq:magn_conv_01} provides
\begin{equation}\label{eq:magn_conv_020}
(f * g) (x) = \frac{1}{2\pi \ell_{B}^{2}}\int_{\mathbb{R}^{2}}\dd y\, f(x-y) g(y) \Phi_{B}(x, y) .
\end{equation}
Then, integrating under the integral sign, and a suitable change of variables provides $K_j(f*g)=(K_jf)*g$, $j=1,2$, where $K_j$ are defined by~\eqref{eq:momenta}. The second equation of~\eqref{eq:convolution-annihilation} follows by observing that $\mathfrak{a}^{\pm}$ are linear combinations of~$K_1$ and~$K_2$.
\end{proof}

The latter result provides the generalization of~\cite[Proposition~4.20]{brezis-87} to the case of the convolution twisted by the multiplier $\Phi_B$. Moreover, similar relations can be deduced by the formulas in~\cite[p.~870]{gracia-varilly-88}.

Let $\psi_{n,m}$ be the Laguerre function defined by~\eqref{eq:herm2}. The next
result describes the magnetic twisted convolution of two Laguerre functions.
\begin{Lemma}\label{lemm:conv_laguer}
For every pair of Laguerre functions $\psi_{k,j}$, $\psi_{n,m}$ the relation
\begin{equation}\label{eq:conv_lag_scal}
\psi_{k,j}\ast\psi_{n,m} = \frac{1}{\sqrt{2\pi}\ell_B} \delta_{j,n} \psi_{k,m}
\end{equation}
holds true.
\end{Lemma}
\begin{proof}The proof, borrowed from~\cite[Theorem 3.1]{greiner-80},\footnote{Different proofs of formula~\eqref{eq:conv_lag_scal} are provided in~\cite[equation~(32)]{gracia-varilly-88} and~\cite[Proposition 1.3.2]{thangavelu-93}.} is based on a direct computation. Let us start by computing
\[
(\psi_{k,k}\ast\psi_{n,n})(x) = \frac{1}{2\pi^2\ell^2_B}I_{k,n}\left(\frac{|x|^2}{2\ell_B^2}\right),
\]
where, by using~\eqref{eq:lag_pol} and a suitable rescaling of variables, one gets
\[
I_{k,n}(v) := \int_{\R^2}\dd u \,\expo{-\frac{|u|^2}{2}}\expo{-\frac{|v-u|^2}{2}}\expo{\ii(u\wedge v)}L^{(0)}_k\big(|u|^2\big)L^{(0)}_n\big(|v-u|^2\big) .
\]
This integral can be evaluated with the method described in~\cite[Theorem~3.1]{greiner-80} and gives
\[
I_{k,n}(v) = \delta_{k,n} \pi \expo{-\frac{|v|^2}{2}}L^{(0)}_k\big(|v|^2\big)
\]
and in turn
\[
\psi_{k,k}\ast\psi_{n,n} =\frac{1}{\sqrt{2\pi}\ell_B} \delta_{k,n} \psi_{k,k},
\]
which proves~\eqref{eq:conv_lag_scal} in the special case $j=k$ and $m=n$. The general formula~\eqref{eq:conv_lag_scal} follows from this special case by acting with the operators $\rr{a}^\pm$ and $\rr{b}^\pm$ as described in Proposition~\ref{prop:acvt_lad_op}.
\end{proof}

Let $\s{F}_B$ be the set of the finite linear combinations of Laguerre functions endowed with the algebraic structure of $\s{L}^1_{B}$.
Clearly $\s{F}_B\subset \s{S}_{B}$ and Lemma~\ref{lemm:conv_laguer} ensures that $\s{F}_B$ is closed under product. In fact, one has that $\s{F}_B$ is a $\ast$-subalgebra of $\s{L}^1_B$ which is dense in view of the density of $\s{F}_B$ in $L^1\big(\R^2\big)$~\cite[Theorem~2.5.1]{thangavelu-93}. In summary:
\begin{Proposition}
$\s{F}_B$ is a dense $\ast$-subalgebra of $\s{L}^1_B$.
\end{Proposition}

The Laguerre basis $\{\psi_{n,m}\}$ can also be used to give a series representation of the Schwartz space ${S}\big(\R^2\big)$. Let ${S}\big(\N_0^2\big)$ be the space
of \emph{rapidly decreasing} sequences $\{c_{n,m}\}\subset\C$
such that
\begin{equation}\label{eq:fre_top}
r_k(\{c_{n,m}\})^2 := \sum_{(n,m)\in\mathbb{N}_{0}^2} (2n+1 )^k (2m+1 )^k |c_{n,m}|^2 < \infty ,
\end{equation}
for all $k\in\N_0$. It is well known that ${S}\big(\N_0^2\big)$ is a Fr\'echet algebra with respect to the system of semi-norms~\eqref{eq:fre_top}. The relation between ${S}\big(\N_0^2\big)$ and ${S}\big(\R^2\big)$ is described in~\cite[Theorem~6]{gracia-varilly-88}. For sake of completeness, and future utility, we state this important result:
\begin{Proposition}\label{prop:discret_sch}
The map
\[
{S}\big(\N_0^2\big) \ni \{c_{n,m}\} \stackrel{\jmath}{\longmapsto} \sum_{(n,m)\in\N_0^2}c_{n,m} \psi_{n,m} \in {S}\big(\R^2\big)
\]
is a topological isomorphism of Fr\'echet algebras.
\end{Proposition}

In view of Proposition~\ref{prop:discret_sch} we can look at the elements of $\s{S}_{B}$
as linear combinations with rapidly decreasing coefficients of the Laguerre functions $\psi_{n,m}$.
This representation will be very useful in the following.

\subsection[The C*-algebra of magnetic operators]{The $\boldsymbol{C^{*}}$-algebra of magnetic operators}\label{sec:C-magn_alg}
The norm $\nnorm{\ }_{B,1}$ of the Banach $\ast$-algebra $\s{L}^1_{B}$ does not verify the $C^*$-condition.
For various reasons it is convenient to complete $\s{L}^1_{B}$ to a $C^*$-algebra,
and in general there exist several inequivalent ways to achieve such a completion. Two of these $C^*$-extensions are of particular interest.

The first is obtained as the completion of $\s{L}^1_{B}$ with respect to the
\emph{universal enveloping norm}
\[
 \lVert f\rVert_{\rm u} := \sup\big\{\lVert \rho(f)\rVert \big| \rho\colon \s{L}^1_{B}\to \bb{B}(\s{H}) \text{is a} \ast\text{-representation}\big\},
\]
where ${\bb{B}(\s{H})}$ is the $C^*$-algebra of bounded operators on $\s{H}$. Moreover, in the definition of $\lVert \rVert_{\rm u}$ one can restrict to non-degenerated representations $\rho$.
Since every $\ast$-representation is automatically continuous~\cite[Proposition~2.3.1]{bratteli-robinson-87} the supremum is well defined and
$\lVert f \rVert_{\rm u}\leqslant\nnorm{f}_{B,1}$ for all $f\in \s{L}^1_{B}$
The resulting abstract $C^*$-algebra
\[
\bb{C}_B^{\rm u} := \overline{\s{L}^1_{B}}^{ \lVert \ \rVert_{\rm u}}
\]
 is called the \emph{enveloping algebra} of $\s{L}^1_{B}$.

The second $C^*$-norm is defined through the map $\pi\colon \s{L}^1_{B}\to \bb{B}\big(L^2\big(\R^2\big)\big)$ defined by
\begin{gather*}
 \pi(f)\psi(x) :=
 \frac{1}{2\pi \ell_{B}^{2}} \int_{\mathbb{R}^{2}}\dd y\, f(y-x)\Phi_{B}(x,y) \psi(y)\\
\hphantom{\pi(f)\psi(x)}{} =
 \left(\frac{1}{2\pi \ell_{B}^{2}} \int_{\mathbb{R}^{2}}\dd y\, f(-y)U_B(y)\right) \psi(x),\qquad \psi \in L^2\big(\R^2\big),
\end{gather*}
where the second equality uses the magnetic translations~$U_B(y)$.
The map $\pi$ is evidently linear and, as a consequence of Young's inequality
\[
 \rVert \pi(f)\psi\lVert_{L^2} \leqslant \nnorm{f}_{B,1} \lVert\psi\rVert_{L^2} ,
\]
it follows that $\pi(f)$ is a bounded operator of norm $\|\pi(f)\|\leqslant \nnorm{f}_{B,1}$. Moreover, the relations $\pi(f*g)=\pi(f)\pi(g)$ and
$\pi(f^*)=\pi(f)^*$ can be checked with a direct computation. This means that~$\pi$
is a $\ast$-representation of $\s{L}^1_{B}$, called the \emph{left-regular representation}.
We refer to
\[
\bb{C}_B := \overline{\pi(\s{L}^1_{B})}^{ \lVert \ \rVert}
\]
as the \emph{$C^*$-algebra of magnetic operators}. By construction
$\|\pi(f)\|\leqslant \lVert f \rVert_{\rm u}$. Therefore, there is a canonical $\ast$-morphism
$\lambda\colon \bb{C}_B^{\rm u}\to \bb{C}_B$, {given by the inclusion on
 $\pi(\mathcal{L}^{1}_{B})$}. Since the underlying group~$\R^2$ is \emph{amenable}, a~classical result proved in \cite[Section~5]{zeller-meier-68} implies that $\lambda$ is an
isomorphism,~i.e.,
\[
 \bb{C}_B \simeq \bb{C}_B^{\text{u}} .
\]
This allows us to identify the abstract {enveloping algebra} of $\s{L}^1_{B}$ with the concrete
$C^*$-algebra~$\bb{C}_B$ of magnetic operators on~$L^2\big(\R^2\big)$.
In particular $\pi$ turns out to be injective. It is worth noting that $\bb{C}_B$ inherits from $\s{L}^1_{B}$ the absence of the unit (cf.\ Corollary~\ref{cor:not_unit}).

Let $\bb{V}_B:=C^*\big(V_B(a),a\in\R^2\big)$ be the $C^*$-algebra generated by the dual magnetic translations. Let $\bb{V}_B'$ be the \emph{commutant} of $\bb{V}_B$, i.e., the set of {bounded} operators which commute with every element in~$\bb{V}_B$.
\begin{Proposition}\label{prop:C-commut}
It holds true that
\[
\bb{C}_B \subset\bb{V}_B' .
\]
In particular every magnetic operator in $\bb{C}_B$ is invariant under the
action of the dual magnetic translations.
\end{Proposition}
\begin{proof}
Since $\pi(f)$ can be written as an integral of the magnetic translations $U_B$, it follows that $
[V_B(a),\pi(f)]=0$
for all $a\in\R^2$ and $f\in\s{L}^1_{B}$ proving that $\pi\big(\s{L}^1_{B}\big)\subset \bb{V}_B'$. The result then follows by continuity.
The inclusion is proper since $\mathbf{1}\in \bb{V}_B'$ but $\bb{C}_B$ is not unital.
\end{proof}

By a classical \virg{$\varepsilon/2$-argument} one can prove that every dense subset
$\s{A}\subseteq \s{L}^1_{B}$ which is dense with respect to the $\nnorm{\ }_{B,1}$-norm provides under the left-regular representation a subset $\pi(\s{A})\subseteq\bb{C}_B$ of the magnetic operators which is dense with respect to the operator norm. As a~consequence the $\ast$-algebras
\[
\bb{F}_{B} := \pi\big(\s{F}_B\big) \qquad \text{and} \qquad
\bb{S}_{B} := \pi\big(\s{S}_{B}\big)
\]
are dense in $\bb{C}_{B}$.

Let us introduce the \emph{transition operators} $\Upsilon_{k\to j}$ defined on the Laguerre basis $\psi_{n,m}$ of $L^2\big(\R^2\big)$
by the relations
\begin{equation}\label{eq:def_Ups}
\Upsilon_{j\mapsto k}\psi_{n,m} := \delta_{j,n} \psi_{k,m} ,\qquad k,j,n,m\in\N_0 .
\end{equation}
In other words, the operator $\Upsilon_{k\mapsto j}$ implements the transition between the
subspaces~$\s{H}_j$ and~$\s{H}_k$, defined by the Landau projections~$\Pi_j$ and~$\Pi_k$, respectively, without changing the second quantum number of the Laguerre
basis.
\begin{Proposition}\label{prop:struct}
It holds true that $\Upsilon_{j\mapsto k}\in \bb{F}_{B}$ for every $k,j\in\N_0$ and
\begin{equation}\label{eq:equal_C-Ups}
\bb{C}_B = C^*(\Upsilon_{j\mapsto k}, k,j\in\N_0),
\end{equation}
where on the right-hand side one has the $C^*$-algebra generated by the transition operators. Moreover, for all $k,j,n,m\in\N_0$ it holds true that:
{\begin{itemize}\itemsep=0pt
\item[$(1)$] $(\Upsilon_{j\mapsto k})^*=\Upsilon_{k\mapsto j}$;
\item[$(2)$] $\Upsilon_{j\mapsto k}\Upsilon_{m\mapsto n}=\delta_{j,n}\Upsilon_{m\mapsto k}$;
\item[$(3)$] $(\Upsilon_{j\mapsto k})^*\Upsilon_{j\mapsto k}=\Upsilon_{j\mapsto j}=\Pi_j$;
\item[$(4)$] $\Pi_k\Upsilon_{j\mapsto k}=\Upsilon_{j\mapsto k}=\Upsilon_{j\mapsto k}\Pi_j$;
\item[$(5)$] $\rr{a}^\pm\Upsilon_{j\mapsto k}=\sqrt{k+\frac{1}{2}\pm\frac{1}{2}} \Upsilon_{j\mapsto k\pm1}$;
\item[$(6)$] $\rr{b}^\pm\Upsilon_{j\mapsto k}=\Upsilon_{j\mapsto k}\rr{b}^\pm$.
\end{itemize}}
\end{Proposition}
\begin{proof}A comparison between the definition of $\pi(f)$ and~\eqref{eq:magn_conv_020} shows that
$\pi(f)\psi=f^-\ast\psi$ whenever $f,\psi\in L^1\big(\R^2\big)$, with $f^-(x):=f(-x)$. From the definition of $\psi_{k,j}$ one gets $\psi_{k,j}^-=(-1)^{j-k}\psi_{k,j}$. Thus, one obtains
\[
 \pi(\psi_{k,j})\psi_{n,m} = (-1)^{j-k}\psi_{k,j}\ast\psi_{n,m} = \frac{(-1)^{j-k}}{\sqrt{2\pi}\ell_B}\delta_{j,n}\psi_{k,m},
\]
where in the last equality we used Lemma~\ref{lemm:conv_laguer}. This shows that
\[
\Upsilon_{j\mapsto k} := (-1)^{k-j}\sqrt{2\pi}\ell_B \pi(\psi_{k,j}) \in \bb{F}_{B} .
\]
Moreover, one obtains that $\bb{F}_B\subseteq C^*(\Upsilon_{j\mapsto k}, k,j\in\N_0)$ and the density of $\bb{F}_B$ in $\bb{C}_B$ implies the equality~\eqref{eq:equal_C-Ups}.
Properties (1)--(6) follow from direct computations and can be checked by the reader.
\end{proof}

By combining the various properties listed in Proposition~\ref{prop:struct} one can compute the following commutators:
\begin{gather}
\big[\rr{a}^+,{\Upsilon_{j\mapsto k}}\big] = \sqrt{k+1} \Upsilon_{j\mapsto k+1} - \sqrt{j} \Upsilon_{j-1\mapsto k} ,\nonumber\\
\big[\rr{a}^-,\Upsilon_{j\mapsto k}\big] = \sqrt{k} \Upsilon_{j\mapsto k-1} - \sqrt{j+1} \Upsilon_{j+1\mapsto k} .\label{eq:commut_ups}
\end{gather}

Let $\bb{K}\big(L^2\big(\R^2\big)\big)$ be the $C^*$-algebra of compact operators on~$L^2\big(\R^2\big)$. An immediate implication of Proposition~\ref{prop:C-commut} is that
$\bb{C}_B\cap \bb{K}\big(L^2\big(\R^2\big)\big)= \{0\}$.
Indeed, as a consequence of the invariance under translations, every eigenvalue of any element of $\bb{C}_B$ must have necessarily infinite multiplicity. However, despite the latter observation
there is a $C^{*}$-isomorphism of~$\bb{C}_B$ onto the elementary $C^{*}$-algebra, i.e.,
the algebra of compact operators on a separable Hilbert
space. To prove this result, and explore its consequences, we need first to introduce some notation. Let
\begin{gather}\label{eq:Lan_proj-dual}
P_{k} := \sum_{n=0}^{\infty}\ketbra{\psi_{n,k}}{\psi_{n,k}} ,\qquad k\in\N_0
\end{gather}
be the \emph{dual Landau projections}. The latter are the spectral projections of $\rr{b}^{-}\rr{b}^{+}$
 (proportional to the dual Landau operator).
Let $\s{V}_k:=P_{k}\big[L^2\big(\R^2\big)\big]$ be the range of
$P_{k}$ and $\bb{K}(\s{V}_k)$ the $C^*$-algebra of compact operators on $\s{V}_k$. Finally, let
\[
\s{U} \colon \ L^2\big(\R^2\big) \longrightarrow \bigoplus_{k\in\N_0}\s{V}_k
\]
be the unitary transformation defined by
$\s{U}\psi:=(P_0\psi,P_1\psi,\ldots)$ for all $\psi\in L^2\big(\R^2\big)$.
\begin{Proposition}\label{prop:isomorphi-to-compact-operators}
The unitary transform $\s{U}$ establishes the unitary equivalence of $C^*$-algebras
 \[
 \s{U} \bb{C}_B \s{U}^* = \bigoplus_{k\in\N_0} \bb{K}(\s{V}_k)
 \]
 defined by $ \s{U}A\s{U}^*=\bigoplus_{k\in\N_0}P_kAP_k$ for all $A\in \bb{C}_B$. Moreover, every projection $\rho_k\colon \bb{C}_B\to \bb{K}(\s{V}_k)$ defined by $\rho_k(A):=P_kAP_k$ is a $\ast$-isomorphism of
 $C^*$-algebras.
 \end{Proposition}
 \begin{proof} From its very definition it follows that every transition operator
 $\Upsilon_{n\mapsto m}$ commutes with the dual Landau projections $P_k$. Since in view of
 Proposition~\ref{prop:struct} the $\Upsilon_{n\mapsto m}$ are generators of $\bb{C}_B$, one
 gets that $[P_k,A]=0$ for all $k\in\N_0$ and for all $A\in \bb{C}_B$. This implies
 that $\rho_k\colon \bb{C}_B\to\bb{B}(\s{V}_k)$ is a well defined $\ast$-homomorphism and
 $\s{U}$ implements the decomposition $A=\sum_{k\in\N_0}P_kAP_k\simeq \bigoplus_{k\in\N_0}\rho_k(A)$.
 The $\ast$-homomorphisms $\rho_k$ are injective. In fact, $P_kAP_k=0$ implies
 $A\psi_{n,k}=0$ for all $n\in\N_0$. Lemma~\ref{lemma:comm_rel_bpm} (which
 generalizes item (5) of Proposition~\ref{prop:struct}) allows to act with the
 operators $\rr{b}^\pm$ proving that $A\psi_{n,m}=0$ for all $n,m\in\N_0$, namely
 $A=0$. Then, $\rho_k$ is a bijection between $\bb{C}_B$ and its image. Since
 $\rho_k(\Upsilon_{n\mapsto m})$ is a rank-one operator on~$\s{V}_k$, it follows from
 Proposition~\ref{prop:struct} that $\rho_k(\bb{C}_B)$ is the norm closure of
 finite rank operators, hence $\rho_k(\bb{C}_B)=\bb{K}(\s{V}_k)$. This concludes
 the proof.
\end{proof}

The $\ast$-isomorphism of $\bb{C}_B$ with the infinite direct sum of the algebra of
compact operators provides information about the spectral nature of the elements of $\bb{C}_B$. Let us denote with
$\sigma_{\rm ess}$ the essential spectrum and with $\sigma_{\rm p}$ the point spectrum.
\begin{Corollary}\label{cor:not_unit}
The following facts hold true:
\begin{itemize}\itemsep=0pt
\item[$(1)$] The $C^*$-algebra $\bb{C}_B$ is not unital;
\item[$(2)$] Let $A\in \bb{C}_B$, then
\[
\sigma(A) = \sigma_{\rm ess}(A) \qquad \text{and} \qquad \sigma(A)\setminus\{0\} = \sigma_{\rm p}(A) .
\]
\end{itemize}
\end{Corollary}

Let us describe some special operator that belongs to $\bb{C}_B$. Proposition~\ref{prop:struct} implies that $\Pi_n\in \bb{F}_B$ for every Landau level~$n$. More precisely, one has that $\Pi_n=\pi(p_n)$ with
\begin{equation}\label{eq:fun_proj}
p_n(x) := \expo{-\frac{|x|^2}{4\ell_B^2}} L_{n}^{(0)}\left(\frac{|x|^2}{2\ell_B^2}\right) \in \s{F}_B .
\end{equation}
For every $s>0$, let $\expo{- \frac{s}{\s{E}_B}H_B}
=\expo{- \frac{s}{2}\left(K_{1}^2+K_{2}^2\right)}$ be the heat semigroup. Since
\[
\expo{- \frac{s}{\s{E}_B}H_B} = \expo{- \frac{s}{2}}\sum_{j=0}^{+\infty}\expo{- sj}\Pi_j
\]
one can use the summation formula~\cite[equation~8.975(1)]{gradshteyn-ryzhik-07} to obtain that $\expo{- \frac{s}{\s{E}_B}H_B} = \pi(g_s)$, with
\begin{equation}\label{eq:fun_heat_ker}
g_s(x) := \frac{\expo{-\frac{|x|^2}{4\ell_B^2}\coth\left(\frac{s}{2}\right)}}{2\sinh\left(\frac{s}{2}\right)} \in {S}\big(\R^2\big) .
\end{equation}
The latter formula, known as \emph{Mehler's kernel} (see~\cite[equation~(3.5)]{avron-herbst-simon-78}), shows that
$\expo{- \frac{s}{\s{E}_B}H_B}\in \bb{S}_B$ for all $s>0$.
The resolvent
\[
R_{\lambda} := (H_B-\lambda\s{E}_B\mathbf{1})^{-1} ,\qquad \lambda-\frac{1}{2} \notin\N_0
\]
can be computed
 through the Laplace transform~\cite[equation~(1.28), p.~484]{kato-95}, i.e.,
\[
R_{\lambda} = \frac{1}{\s{E}_B}\int_0^{+\infty}\dd s \expo{\lambda s} \expo{- \frac{s}{\epsilon_B}H_B} .
\]
Then, one has that $R_{\lambda}=\pi(r_\lambda)$ where the function $r_\lambda$ can be obtained from the Laplace transform of~$g_s$. The use of the formula~\cite[equation~9.222(1)]{gradshteyn-ryzhik-07} and a suitable change of variables provide
\[
r_\lambda(x) := \frac{\Gamma\left(\frac{1}{2}-\lambda\right)}{\s{E}_B}\frac{\sqrt{2}\ell_B
}{|x|} W_{\lambda,0}\left(\frac{|x|^2}{2\ell_B^2}\right) \in L^1\big(\R^2\big)\cap L^2\big(\R^2\big),
\]
 where $\Gamma$ is the Gamma function and $W_{\lambda,0}$ is the Whittaker's function.
A more detailed derivation of this formula is described in~\cite[Appendix~C]{comtet-87}.
The integrability of $r_\lambda$ implies that $R_{\lambda}\in \bb{C}_B$ for all $\lambda-\frac{1}{2} \notin\N_0$.

Let us recall that a self-adjoint unbounded operator $H$ is affiliated to the $C^*$-algebra $\bb{C}_B$ if $(H-\lambda\mathbf{1})^{-1}\in \bb{C}_B$
for some (and therefore for every) $\lambda$ in the resolvent set of $H$ (see, e.g.,~\cite{damak-georgescu-04}). This definition
allows to state the following result.
\begin{Proposition}\label{prop:affil}
The Landau Hamiltonian $H_B$ is affiliated to the magnetic $C^*$-algebra $\bb{C}_B$.
\end{Proposition}

We are now in a position to prove a pair of crucial properties for the subalgebra $\bb{S}_{B}$. First of all, we need to recall the notion of
\emph{pre-$C^*$-algebra} according to~\cite[Definition~3.26]{gracia-varilly-figueroa-01} or~\cite[Definition~3.1.1]{blackadar-98} (where it is presented under the name of \emph{local-$C^*$-algebra}). In a nutshell,
a pre-$C^*$-algebra
 is a dense $\ast$-subalgebra of a $C^*$-algebra which is stable under holomorphic functional calculus. In the non-unital case (which is the relevant case for us) one has to restrict to holomorphic functions that satisfy $f(0)=0$.
\begin{Proposition}\label{prop:struct_C_infty}
The following properties hold true:
\begin{itemize}\itemsep=0pt
\item[$(1)$]
$\bb{S}_{B}$ is a non-unital pre-$C^*$-algebra of $\bb{C}_{B}$;
\item[$(2)$]
$\bb{S}_{B}$ has the $($non-unique$)$ factorization
property: for all $T\in \bb{S}_{B}$ there exist $S_1,S_2\in \bb{S}_{B}$ such that
$T=S_1S_2$.
\end{itemize}
\end{Proposition}
\begin{proof}
(1) Let $T\in \bb{S}_{B}$ with spectrum $\sigma(T)$. Since $\bb{C}_{B}$ is non-unital
it follows that $0\in \sigma(T)$.
{Let $f$ be an analytic function on an open neighborhood $\s{U}$ of the spectrum
 $\sigma(T)$ such that the boundary $\Gamma:=\partial\s{U}$ is a~Jordan
curve}
and $g(z):=f(z)-f(0)$. To prove that $\bb{S}_{B}$ is a~pre-$C^*$-algebra we need to shows that
\[
 g(T) := -\frac{1}{2\pi\ii}\oint_{\Gamma}\dd z \,\frac{g(z)}{T-z\mathbf{1}} = -\frac{1}{2\pi\ii}\oint_{\Gamma}\dd z \, \frac{f(z)}{T-z\mathbf{1}} - f(0){\bf 1}
\]
is in $\bb{S}_{B}$
for every analytic function $f$. By using Cauchy's formula $f(0)= \frac{1}{2\pi\ii}\oint_{\Gamma}\dd z\, f(z)z^{-1}$
we get that $g(T)\in \bb{S}_{B}$ follows if we can prove that $S_z:=(T-z\mathbf{1})^{-1}+z^{-1}{\bf 1}\in \bb{S}_{B}$
for every $z\in\C\setminus \sigma(T)$.
We know
a priori that
 $(T-z\mathbf{1})^{-1}\in \bb{M}_{B}$. Observing that
 \begin{equation}\label{eq:aux_oo1}
T(T-z\mathbf{1})^{-1}={\bf 1}+z(T-z\mathbf{1})^{-1} = z S_z
\end{equation}
one gets $TS_z=S_zT=z^{-1}T(T-z\mathbf{1})^{-1}T\in\bb{S}_{B}$ in view of
 Lemma~\ref{lemma:append_A}. On the other hand, by multiplying the second equality in~\eqref{eq:aux_oo1} on the left by $T$ one gets
 $zTS_z=T+z^2S_z$, and in turn $S_z=z^{-2}(zTS_z-T)\in \bb{S}_{B}$.

 (2) This follows from a similar factorization
property for $\s{S}_{B}$ proved in~\cite[Theorem 7 and Corollary]{gracia-varilly-88}.
\end{proof}

Item (2) of Proposition~\ref{prop:struct_C_infty} can be reformulated as follows. Let
$(\bb{S}_{B})^2$ be the linear space generated by the products $S_1S_2$ with $S_1,S_2\in \bb{S}_{B}$. Then $(\bb{S}_{B})^2=\bb{S}_{B}$.

\subsection{The Hilbert algebra structure}\label{sec:Hil-alg}
One interesting property of the magnetic twisted convolution is that in addition to the classic Young's inequalities, it also satisfies the following $L^2$-inequality~\cite[Proposition 1.33]{folland-89} or~\cite[Theorem 1.2.2]{thangavelu-93}
\begin{equation}\label{eq:continuity-magnetic-translations}
 \lVert f * g \rVert_{L^2} \leqslant \frac{1}{\sqrt{2\pi}\ell_B}\lVert f \rVert_{L^2} \lVert g \rVert_{L^2} ,\qquad f,g \in L^{2}\big(\R^2\big) .
\end{equation}
It is worth noting that the inequality becomes singular in the limit of vanishing magnetic field $\ell_B\to\infty$.
In view of~\eqref{eq:continuity-magnetic-translations}, every $f\in L^{2}\big(\R^2\big)$ define the integral type operator
\[
(L_f\psi)(x) := \frac{1}{2\pi \ell_{B}^{2}} \int_{\mathbb{R}^{2}}\dd y\, f(y-x) \Phi_{B}(x,y) \psi(y) ,
\qquad \psi \in L^2\big(\R^2\big)
\]
with norm bounded by $\|L_f\|\leqslant \frac{1}{\sqrt{2\pi}\ell_B}\lVert f \rVert_{L^2}$. It turns out that $L_f$ is a Carleman operator~\cite[Section~11]{halmos-sunder-78}.
Let
\[
\bb{L}^2_B := \big\{L_f\in\bb{B}\big(L^2\big(\R^2\big)\big) \,|\, f\in L^2\big(\R^2\big)\big\}
\]
be the set of convolution-type operators with $L^2$-kernel.
Evidently, $L_f=\pi(f)$ whenever $f\in L^1\big(\R^2\big)\cap L^2\big(\R^2\big)$.
The relation between $\bb{L}^2_B$ and the $C^*$-algebra $\bb{C}_B$ is described in the following result.
\begin{Proposition}\label{prop:inclus_after_first}
$\bb{L}^2_B$ is a dense two-sided ideal of $\bb{C}_B$. Moreover, the following $($proper$)$ inclusions hold true
\[
\bb{S}_B \subset \bb{L}^2_B \subset \bb{C}_B .
\]
\end{Proposition}
\begin{proof} Let $\s{L}^{1,2}_{B}:=\s{L}^1_{B}\cap L^2\big(\R^2\big)$. This is a $\ast$-subalgebra of $\s{L}^1_{B}$ in view of the inequality~\eqref{eq:continuity-magnetic-translations}
and $\s{S}_{B}\subset \s{L}^{1,2}_{B}$. This implies that $\s{L}^{1,2}_B$ is dense in $\s{L}^1_{B}$ and in turn $\pi(\s{L}^{1,2})$ is dense in $\bb{C}_B$. Moreover, the inclusions $\bb{S}_B\subset \pi\big(\s{L}^{1,2}_B\big)\subset \bb{L}^2_B$ holds. Let $f\in L^2\big(\R^2\big)$ and $\{f_n\}\subset \s{L}^{1,2}$ be a~sequence such that $f_n\to f$ in the $L^2$-norm. In view of the $L^2$-estimate for the norm of elements in $\bb{L}^2_B$ one has that $\|L_f-\pi(f_n)\|\to 0$, i.e., $L_f$ lies in the norm closure of $ \pi\big(\s{L}^{1,2}_B\big)$ which coincides with $\bb{C}_B$. This proves the inclusions
stated in the claim. The ideal property is a consequence of the more general result in Proposition~\ref{prop:von_propert}(2).
\end{proof}

The subspace $\s{L}^{1,2}_{B}\subset \s{L}^{1}_{B}$ endowed with the algebraic structure of $\s{L}^{1}_{B}$ and the $L^2$-scalar product
\[
\langle f,g\rangle_B := \frac{1}{2\pi \ell_{B}^{2}}\int_{\R^2}\dd x\, \overline{f(x)}g(x) ,
\]
is a \emph{Hilbert algebra} in the sense of~\cite[Part~I, Chapter~5, Section~1]{dixmier-81}
Indeed, it is only matter of straightforward computations to check that the conditions which define the structure of a~Hilbert algebra are satisfied.
Moreover, one can check that $L_f\psi=f^-\ast\psi$ for all $f,\psi\in \s{L}^{1,2}_{B}$, where $f^-(x):=f(-x)$. Then, after observing that the completion of $\s{L}^{1,2}_{B}$ with respect to the norm induced by the scalar product $\langle \,,\, \rangle_B$ is exactly $L^2\big(\R^2\big)$, one
can prove that every element in
 $\bb{L}^2_B$ is a \emph{left bounded} operator associated to $\s{L}^{1,2}_{B}$ according to~\cite[Part~I, Chapter~5, Definition~2]{dixmier-81}. Moreover, in view of the effect of the involution $J$, every element in
 $\bb{L}^2_B$ is also a \emph{right bounded} operator~\cite[Part~I, Chapter~5, Proposition~2]{dixmier-81}.
 Indeed, a direct computation shows
\begin{align*}
(JL_{f}J\psi)(x) &= \int_{\mathbb{R}^{2}}\dd y\, \overline{f(y-x)} \Phi_{B}(y,x) \psi(y)
 = \frac{1}{2\pi \ell_{B}^{2}} \int_{\mathbb{R}^{2}}\dd y\, \psi^{-}(y-x) \Phi_{B}(x,y) \overline{f(y)}\\
&= \big(\psi\ast (Jf)^{-}\big)(x) =: (R_{Jf}\psi)(x) ,
\end{align*}
where we used $\overline{f(x)}=\overline{f^-(-x)}=(Jf)^{-}(x)$.

Every $g\in L^2\big(\R^2\big)$ admits a unique expansion in terms of the Laguerre basis $\psi_{n,m}$ according to
$g=\sum_{(n,m)\in\N_0^2}g_{n,m} \psi_{n,m}$. As a consequence the associated operator $L_g\in \bb{L}^2_B$ is uniquely determined by the expansion
\begin{equation}\label{eq:exp_oper_Ups}
L_g = \frac{1}{\sqrt{2\pi}\ell_B}\sum_{(n,m)\in\N_0^2}(-1)^{m-n}g_{n,m} \Upsilon_{m\mapsto n} .
\end{equation}
In other words we are exploiting the identification
$L^2\big(\R^2\big)\simeq \ell^2\big(\N_0^2\big)$ induced by the
connection between the Laguerre basis
$\psi_{n,m}$ and the operators $\Upsilon_{m\mapsto n}$ to define the topological isomorphism
\[
\bb{L}^2_B \simeq \ell^2\big(\N_0^2\big) .
\]
In the same spirit we can define the space
\[
\bb{L}^1_B \simeq \ell^1\big(\N_0^2\big)
\]
formed by operators of the form~\eqref{eq:exp_oper_Ups} with associated sequence
$\{g_{n,m}\}\in \ell^1\big(\N_0^2\big)$. It turns out that $\bb{L}^1_B$ is a Banach space with respect to the topology induced by the norm of $\ell^1\big(\N_0^2\big)$.
Another space that will play an important role in Sections~\ref{sec:dix_tr} and~\ref{sec:dix_tr_2}
is the set $\big(\bb{L}^2_B\big)^2$ made up of products of operators of the form $L_{g_1}L_{g_2}$ with
$L_{g_1},L_{g_2}\in \bb{L}^2_B$.

\begin{Remark}\label{lemm:pre_inclus}
Every element in $\big(\bb{L}^2_B\big)^2$ is of the form $L_h$ with $h\in L^2\big(\R^2\big)\cap C_0\big(\R^2\big)$, where $C_0\big(\R^2\big)$ is the space of continuous function which vanish at infinity.
Indeed, by observing that $L_{g_1}L_{g_2}=L_{g_1\ast g_2}$, one can describe
$\big(\bb{L}^2_B\big)^2$ as the space of operators with the integral kernel given by the
magnetic twisted convolution of two $L^2$-functions. From~\eqref{eq:continuity-magnetic-translations} it follows that $g_1\ast g_2\in L^2\big(\R^2\big)$. Moreover, by adapting a classical argument~\cite[Lemma~2.20]{lieb-loss-01} one can prove that
$g_1\ast g_2\in C_0\big(\R^2\big)$.
\end{Remark}

The mutual relations between the various spaces introduced above are described in the next result.
\begin{Proposition}\label{prop:very_inclus}
The following chain of inclusions holds true
\begin{equation}\label{eq:L^p-inclus}
\bb{S}_B \subset \bb{L}^1_B \subset \big(\bb{L}^2_B\big)^2 \subset \bb{L}^2_B \subset \bb{C}_B .
\end{equation}
\end{Proposition}

\begin{proof} The first inclusion follows from Proposition~\ref{prop:discret_sch} and the inclusion
${S}\big(\N_0^2\big)\subset \ell^1\big(\N_0^2\big)$.
The third inclusion follows from Remark~\ref{lemm:pre_inclus}.
The last inclusion has been proved in Proposition~\ref{prop:inclus_after_first}.
Then, it only remains to prove the second inclusion.
Let $L_g\in \bb{L}^1_B$ with $g=\sum_{(n,m)\in\N_0^2}g_{n,m} \psi_{n,m}$, satisfying $\{g_{n,m}\}\in \ell^1\big(\N_0^2\big)$. Let
\[
d_r := \sup_{n\in\N_0} \Big\{\sqrt{|g_{n,r}|} \Big\}
\]
and define two sequences $a_{n,r}:=d_r^{-1}g_{n,r}$, (with $a_{n,r}=0$ if $d_r=0$) and $b_{r,m}:= \delta_{r,m} d_r$.
Let $a:=\sum_{(n,r)\in\N_0^2}a_{n,r} \psi_{n,r}$,
and $b:=\sum_{(r,m)\in\N_0^2}b_{r,m} \psi_{r,m}$. By construction
$a,b\in L^2\big(\R^2\big)$. Moreover, since $g_{n,m}=\sum_{r\in\N_0}a_{n,r}b_{r,m}$, one gets that $g=a\ast b$. This implies that
$L_g=L_aL_b$ with $L_a,L_b\in \bb{L}^2_B$ and this concludes the proof.
\end{proof}

\subsection{The von Neumann algebra of magnetic operators}\label{sec:W*magn_alg}
The von Neumann algebra of magnetic operators $\bb{M}_B$ is the bicommutant of the $C^*$-algebra $\bb{C}_B$, namely
\[
\bb{M}_B := \bb{C}_B'' .
\]
The celebrated density theorem
implies that $\bb{M}_B$ coincides with the weak-closure of $\bb{C}_B$, or equivalently with the strong-closure of $\bb{C}_B$.
The main properties of $\bb{M}_B$ are described below.
\begin{Proposition}\label{prop:von_propert}
Let $\bb{M}_B$ be the von Neumann algebra of magnetic operators. The following facts hold true:
\begin{itemize}\itemsep=0pt
\item[$(1)$] The $\ast$-subalgebras $\bb{F}_B$, $\bb{S}_B$, $\bb{L}^1_B$ and $\bb{L}^2_B$
are weakly and strongly dense in $\bb{M}_B$;
\item[$(2)$] $\bb{L}^2_B$ is a two-sided ideal of $\bb{M}_B$;
\item[$(3)$] $\bb{M}_B= \bb{V}_B'$ where $\bb{V}_B$ is the $C^*$-algebra generated by the dual magnetic translations.
\end{itemize}
\end{Proposition}
\begin{proof} (1) since the convergence in norm dominates the weak (strong)
convergence, every norm-dense $\ast$-subalgebra $\bb{A}\subset \bb{C}_B$ is automatically
weakly (strongly) dense in~$\bb{C}_B$. As a~consequence~$\bb{A}$ is also
weakly (strongly) dense in $\bb{C}_B''$.

(2) Since $\bb{M}_B={\bb{L}^2_B}''$ it follows that $\bb{M}_B$ is the \emph{left} von Neumann algebra generated by the Hilbert algebra $\s{L}^{1,2}_{B}$ (see the discussion and the references at the end of Section~\ref{sec:C-magn_alg}). Then $\bb{L}^2_B$ is a two-sided ideal in $\bb{M}_B$ in view of~\cite[Part~I, Chapter~5, Proposition~3]{dixmier-81}.

(3) Since $\bb{V}_B'$ is a commutant, and therefore a von Neumann algebra, one
gets the inclusion $\bb{M}_B\subseteq \bb{V}_B'$ directly from
Proposition~\ref{prop:C-commut}. Now, let us consider the \emph{right} von
Neumann algebra~$\bb{R}_B$ generated by the Hilbert algebra
$\s{L}^{1,2}_{B}$. For a given $f\in\s{L}^{1,2}_{B}$, the associated operator
$R_f\in\bb{R}_B$ acts as $R_f\psi:=\psi\ast f^{-}$. Then, an explicit check shows that
$R_f$ is exactly the integrated operator
$\frac{1}{2\pi\ell^2_B}\int_{\R^2}\dd y\,f(-y)V_B(y)$ where the $V_B(y)$'s are the dual
magnetic translations. As a consequence $\bb{R}_B\subseteq \bb{V}_B$ and
$\bb{V}_B'\subseteq \bb{R}_B'$. The proof is completed by the equality
$\bb{R}_B'=\bb{M}_B$ proved in~\cite[Part~I, Chapter~5,Theorem~1]{dixmier-81}.
\end{proof}

 The next result is a consequence of the invariance of the elements of $\bb{M}_B$ under the dual magnetic translations.

 \begin{Lemma}\label{lemma:comm_rel_bpm}
 Let $G_j$, with $j=1,2$, be the self-adjoint operators defined
 by~\eqref{eq:dual-momenta}, and let $\s{D}(G_j)$ their related domains. Then,
 for every $A\in \bb{M}_B$, it holds true that $A[\s{D}(G_j)]\subseteq \s{D}(G_j)$ and
 \[
 [A,G_j] = 0 ,\qquad j=1,2 .
 \]
 Furthermore, the commutation relation extends to the ladder operators $\rr{b}^\pm$, i.e.,
 \[
 \big[A,\rr{b}^\pm\big] = 0 ,\qquad \forall\, A\in\bb{M}_B .
 \]
\end{Lemma}
\begin{proof}
The resolvents $R_j:=(G_j-\ii\mathbf{1})^{-1}$ can be obtained from the dual magnetic translations via the Laplace transform,
\[
R_j = \ii\int_0^{+\infty}\dd s\, \expo{-s} V_B(s\ell_B e_j) ,
\]
where $e_1=(1,0)$ and $e_2=(0,1)$ are the canonical basis of $\R^2$. Since the integral is defined in the strong sense one gets $R_j\in \bb{V}_B''=\bb{M}_B'$ where the second equality is a consequence of Proposition~\ref{prop:von_propert}(3). Since $R_j\big[L^2\big(\R^2\big)\big]=\s{D}(G_j)$, it follows that
$A[\s{D}(G_j)]=R_jA\big[L^2\big(\R^2\big)\big]\subseteq \s{D}(G_j)$. As a consequence the
difference $AG_j-G_jA$ is well defined and closable on the dense domain $\s{D}(G_j)$
and we can denote with $[A,G_j]$ the related closure.
 From the equation $R_j G_j=\mathbf{1}+\ii R_j$, one gets $R_j(AG_j-G_jA)=[A,R_jG_j]=\ii[A,R_j]=0$
where the sequence of equalities is well defined on the dense domain $\s{D}(G_j)$. Since $R_j$ is invertible one obtains $(AG_j-G_jA)=0$ on a dense set. This implies $[A,G_j]=0$. The commutation relations for the
ladder operators are consequence of the fact that $\rr{b}^\pm$ are linear combinations of
$G_1$ and $G_2$.
\end{proof}

A further characterization of $\bb{M}_B$ can be obtained from an application of Proposition~\ref{prop:isomorphi-to-compact-operators}.
\begin{Proposition}\label{prop:isomorphi-to-compact-operators_von_neu}
The unitary transform $\s{U}$ establishes a unitary equivalence of von Neumann algebras
 \[
 \s{U} \s{M}_B \s{U}^* = \bigoplus_{k\in\N_0} \bb{B}(\s{V}_k)
 \]
 defined by $\s{U}A\s{U}^*=\bigoplus_{k\in\N_0}P_kAP_k$ for all $A\in \s{M}_B$. Moreover, every projection $\rho_k\colon \s{M}_B\to \bb{B}(\s{V}_k)$ defines a $\ast$-isomorphism.
 \end{Proposition}
\begin{proof}
The claim is a direct consequence of Proposition~\ref{prop:isomorphi-to-compact-operators} along with the fact that the weak (strong) closure of the algebra of compact operators is the algebra of all bounded operators.
\end{proof}

\subsection{Integration theory and the trace per unit volume}\label{sec:dix_tr}
The von Neumann algebra $\bb{M}_{B}$ admits a privileged trace which allows to build a noncommutative integration theory.
We will introduce such a trace following the construction of~\cite{lenz-99}.

The trace on $\bb{M}_{B}$ can be induced from the underlying Hilbert algebra structure presented in Section~\ref{sec:C-magn_alg}.
\begin{Proposition}\label{prop:FSN-trace}
There exists a unique normal trace $\fint_B$ on $\bb{M}_{B}$ defined by
\[
\fint_B(L_f^*L_f) := \langle f,f\rangle_B ,\qquad \forall\, L_f\in\bb{L}^2_B .
\]
Moreover, the trace $\fint_B$ is semi-finite and faithful, and its ideal of definition is given by $\big(\bb{L}^2_B\big)^2\subset \bb{L}^2_B$, where $\big(\bb{L}^2_B\big)^2$ is the span of the products $S_1S_2$
for all $S_1,S_2\in \bb{L}^2_B$.
Finally, it holds true that
\[
\fint_B(L_f^*L_g) := \langle f,g\rangle_B ,\qquad \forall\, L_f,L_g\in\bb{L}^2_B .
\]
\end{Proposition}
\begin{proof}
Since $\bb{M}_{B}$ coincides with the \emph{left} von Neumann algebra generated by the Hilbert algebra~$\s{L}^{1,2}_{B}$, then~\cite[Part~I, Chapter~5, Theorem~1]{dixmier-81} applies verbatim. Therefore, $\fint_B$ coincides with the \emph{natural} trace associated with the Hilbert algebra structure~\cite[Part~I, Chapter~5, Definition~2]{dixmier-81}. As a consequence
$\fint_B$
is semi-finite, faithful and normal with ideal of definition $\bb{L}^2_B\subset \bb{M}_{B}$. Unicity follows from~\cite[Lemma~2.2.1]{lenz-99}.
\end{proof}

Another way of describing the domain of $\fint_B$ is the following:
$T\in \bb{M}_{B}^+$ meets the condition $\fint_B(T)\leqslant +\infty$, if and only if $S:=\sqrt{T}\in \bb{L}^2_B$.
The next result provides a useful criterion for computation.
\begin{Corollary}\label{prop:comput_trace}
Let $L_h\in \big(\bb{L}^2_B\big)^2$. Then
\[
\fint_B(L_h) = h(0) .
\]
\end{Corollary}
\begin{proof}
From Remark~\ref{lemm:pre_inclus} we know that $h\in L^2\big(\R^2\big)\cap C_0\big(\R^2\big)$, then the expression $h(0)$ makes sense. Since $L_h=L_{f}^*L_g=L_{f^**g}$ for $f,g\in L^2\big(\R^2\big)$ by definition, Then
$\fint_B(L_h)= \langle f,g\rangle_B$ in view of Proposition~\ref{prop:FSN-trace}. The equality
$h(0)=(f^**g)(0)= \langle f,g\rangle_B$ concludes the proof.
\end{proof}

Corollary~\ref{prop:comput_trace} is particularly useful to compute the trace of elements in
$\bb{S}_B$ (or more in general in $\bb{L}^1_B$). For instance, the trace of the Landau projection $\Pi_n$ is given by
\begin{equation}\label{eq:trac_proj_olb}
\fint_B(\Pi_n) = p_n(0) = 1 ,\qquad n\in\N_0,
\end{equation}
where the function $p_n$ is defined in~\eqref{eq:fun_proj}. Similarly, the trace of the heat semigroup $\expo{- \frac{s}{\s{E}_B}H_B}$ is given by
\begin{equation*}
\fint_B\big(\expo{- \frac{s}{\s{E}_B}H_B}\big) = g_s(0) = \frac{1}{2\sinh\big(\frac{s}{2}\big)} ,\qquad s>0,
\end{equation*}
where the function $g_s$ is defined in~\eqref{eq:fun_heat_ker}. Finally the trace of the operators
$\Upsilon_{j\mapsto k}$ defined by~\eqref{eq:def_Ups} is given by
\begin{equation*}
\fint_B(\Upsilon_{j\mapsto k}) := (-1)^{j-k}{\sqrt{2\pi}\ell_B} \psi_{k,j}(0) = \delta_{k,j} ,\qquad j,k\in\N_0 .
\end{equation*}

The next result links the trace $\fint_B$ with the standard trace in $L^2\big(\R^2\big)$.
\begin{Lemma}\label{lemma:pre-trac_uni_vol}
Let $\Lambda\subseteq \R^2$ be a compact subset with finite volume $|\Lambda|$, and $\chi_\Lambda$
the related characteristic function acting as a projection on $L^2\big(\R^2\big)$.
The operator $\chi_\Lambda L_f^*L_g \chi_\Lambda$ is trace class
and the following equality
\[
 \fint_B(L_f^*L_g) = \frac{2\pi\ell^2_B}{|\Lambda|}{\Tr}_{L^2(\R^2)}( \chi_\Lambda L_f^*L_g \chi_\Lambda )
\]
holds true for every $L_f,L_g\in\bb{L}^2_B$.
\end{Lemma}
\begin{proof}
The operator $L_g \chi_\Lambda$ has integral kernel
\[
\kappa_g(x,y) := \frac{1}{2\pi \ell_{B}^{2}} g(y-x) \Phi_{B}(x,y) \chi_\Lambda(y)
\]
with $g\in L^2\big(\R^2\big)$. Since $\chi_\Lambda$ is supported on a compact set it follows that
 $\kappa_g \in L^2\big(\R^2\times \R^2\big)$, thus $L_g \chi_\Lambda$ is a Hilbert--Schmidt operator.
 A similar argument shows that also $\chi_\Lambda L_f^*$ is Hilbert--Schmidt. Therefore the product $\chi_\Lambda L_f^*L_g \chi_\Lambda$ is trace class and has kernel
 \[
 \kappa_{f,g}(x,y) := \frac{\chi_\Lambda(x)\chi_\Lambda(y)}{\big(2\pi \ell_{B}^{2}\big)^2}\int_{\R^2}\dd s \overline{f(x-s)} g(y-s) \Phi_{B}(s,y-x),
\]
 which along the diagonal reads
 \[
 \kappa_{f,g}(x,x) := \frac{\chi_\Lambda(x)}{2\pi \ell_{B}^{2}}\langle f,g\rangle_B .
 \]
 Since the trace of $\chi_\Lambda L_f^*L_g \chi_\Lambda$ is given by the integral of $ \kappa_{f,g}$ along the diagonal one obtains
\[
{\Tr}_{L^2(\R^2)}( \chi_\Lambda L_f^*L_g \chi_\Lambda) = \frac{|\Lambda|}{2\pi \ell_{B}^{2}}\langle f,g\rangle_B .
\]
The proof is completed by using the definition of $\fint_B$.
\end{proof}

\begin{Remark}[relation with the trace per unit of volume I]\label{rk:trac_unit_volI}
Let $\Lambda_{n} \subseteq \R^2$
an
 increasing sequence of compact subsets such that $\Lambda_n\nearrow\R^2$ and which satisfies the \emph{F{\o}lner condition} (see, e.g., \cite{greenleaf-69} for more details). A bounded operator $T$ admits a thermodynamic limit (related to the F{\o}lner sequence $\Lambda_n$) if the limit
 \[
\s{T}_B(T) := \lim_{n\to+\infty} \frac{1}{|\Lambda_n|}{\Tr}_{L^2(\R^2)}( \chi_{\Lambda_n}T \chi_{\Lambda_n} )
 \]
 exists. The linear functional $\s{T}_B$ is known as the \emph{trace per unit volume}.
 Lemma~\eqref{lemma:pre-trac_uni_vol} immediately implies that
every element $T$ in the domain of definition of $\fint_B$ admits trace per unit of volume independently of the election of a~particular F{\o}lner sequence. In particular, it holds true that
 \[
 \s{T}_B(T) = \frac{1}{2\Lambda_B} \fint_B(T) ,
 \]
 where $\Lambda_B:=\pi\ell_B^2$ is the area of the \emph{magnetic disk} of radius $\ell_B$.
\end{Remark}

In view of Proposition~\ref{prop:FSN-trace}, the von Neumann algebra
$\bb{M}_{B}$ turns out to be endowed with the faithful, semi-finite and normal
(FSN) trace $\fint_B$. The pair $(\bb{M}_{B},\fint_B)$ is the basic element for
the development of the noncommutative integration theory in the sense
of~\cite{nelson-13,segal-53,terp-93} (see also~\cite[Section~3.2]{denittis-lein-book} and references therein). Let $T\in\big(\bb{L}^2_B\big)^2$ be an
element of the domain of $\fint_B$. The $L^p$-norm of $T$ is given by
 \[
 \nnorm{T}_{B,p} := \left[\fint_B\big(|T|^p\big)\right]^{\frac{1}{p}} ,\qquad 0<p<+\infty .
 \]
 The noncommutative $L^p$-spaces are defined by
 \[
 \rr{L}_B^p := \overline{\big(\bb{L}^2_B\big)^2}^{ \nnorm{\ }_{B,p}} .
 \]
 The spaces $\rr{L}_B^p$ are Banach spaces of possibly unbounded operators and
 $\rr{L}_B^p\subset {\rm Aff}(\bb{M}_{B})$ where ${\rm Aff}(\bb{M}_{B})$ is the set of closed and densely defined operators affiliated with $\bb{M}_{B}$. The identification $\bb{M}_{B}=\rr{L}_B^\infty$ is often used.

\subsection{Integration theory and Dixmier trace}\label{sec:dix_tr_2}
The trace $\fint_B$ on $\bb{M}_{B}$ can be described making use of the \emph{Dixmier trace} along the line anticipated in~\cite{denittis-gomi-moscolari-19}.
 There are several standard references
 for the theory of the Dixmier trace, see, e.g.,~\cite[Chapter~4, Section~2]{connes-94}, \cite[Appendix~A]{connes-moscovici-95}, \cite[Section~7.5 and Appendix~7.C]{gracia-varilly-figueroa-01}, \cite[Chapter~6]{lord-sukochev-zanin-12}, \cite{alberti-matthes-02}.
 Here, we will recall only the basic definition of the Dixmier trace.
Let us start with the \emph{singular values} $\mu_n(T)$ of a compact operator~$T$ which are, by definition, the eigenvalues of $|T|:=\sqrt{T^*T}$. By convention the singular values will be listed in decreasing order, repeated according to their multiplicity, i.e.,
\[
\mu_0(T) \geqslant \mu_1(T) \geqslant \cdots \geqslant \mu_n(T) \geqslant \mu_{n+1}(T) \geqslant \cdots \geqslant 0 .
\]
In particular $\mu_0(T) =\|\,|T|\,\|=\|T\|$.
Let
\begin{equation*}
\sigma_N^p(T) := \sum_{n=0}^{N-1}\mu_n(T)^p ,\qquad p\in[1,+\infty) .
\end{equation*}
A compact operator $T$ is in the $p$-th
\emph{Schatten ideal} $\rr{S}^p$, if and only if,
 $\|T\|_p^p:=\lim\limits_{N\to\infty} \sigma_N^p(T)<+\infty$. Accordingly, $\rr{S}^1$ is the ideal
 of the trace-class operators. Let
\begin{equation*}
\gamma_N(T) := \frac{\sigma_N^1(T)}{\log(N)} = \frac{1}{\log(N)}\sum_{n=0}^{N-1}\mu_n(T) ,\qquad N>1 .
\end{equation*}
A compact operator $T$ is in the \emph{Dixmier ideal} $\rr{S}^{1^+}$ if its \emph{$($Calder\'on$)$} norm
\begin{equation}\label{eq:clad_norm}
\lVert T\rVert_{1^+} := \sup_{N>1}\ \gamma_N(T) < +\infty
\end{equation}
is finite. $\rr{S}^{1^+}$ is a two-sided self-adjoint ideal which is closed with respect to the norm~\eqref{eq:clad_norm} (but not
with respect to the operator norm). The set of operators such that
$\lim\limits_{N\to\infty} \gamma_N(T)=0$ forms {an ideal} inside $\rr{S}^{1^+}$ denoted with $\rr{S}^{1^+}_0$.
The chain of
 (proper) inclusions $\rr{S}^{1}\subset\rr{S}^{1^+}_0\subset\rr{S}^{1^+}\subset\rr{S}^{1+\epsilon}$ holds true for every $\epsilon>0$.
To define a trace functional with domain the Dixmier ideal $\rr{S}^{1^+}$ we need to fix a \emph{generalized scale-invariant limit}\footnote{A generalized scale-invariant limit is a continuous positive linear functional $\omega\colon \ell^{\infty}(\mathbb{N}) \to \mathbb{C}$ which coincides with the ordinary limit on the subspace of convergent
sequences and is invariant under \virg{dilations} of the sequences of the type $\{a_1,a_2,a_3,\ldots\}\mapsto \{a_1,a_1,a_2,a_2,a_3,a_3,\ldots\}$.}
$\omega\colon \ell^{\infty}(\mathbb{N}) \to \mathbb{C}$.
The $\omega$-{Dixmier trace} of a~positive element of the Dixmier ideal is defined as
\[
 {\Tr}_{{\rm Dix},\omega}(T) : = \omega( \{ \gamma_{N}(T)\}_{N} )
 ,\qquad T \in \rr{S}^{1^+} ,\qquad T\geqslant0 .
\]
The definition of ${\Tr}_{{\rm Dix},\omega}$
extends to non-positive elements of $\rr{S}^{1^+}$ by linearity. The
$\omega$-Dixmier trace
provides an example of a singular (hence non-normal) trace and it is continuous with respect to the Calder\'on norm~\eqref{eq:clad_norm}, i.e.,
$|{\Tr}_{{\rm Dix},\omega}(T)|\leqslant \|T\|_{1^+}$.
 A element $T\in \rr{S}^{1^+}$ is called \emph{measurable} if
the value of $\omega (\{\gamma_N(T)\}_N)$ is independent of the
 choice of the generalized scale-invariant limit~$\omega$. For a positive element $T\geqslant 0$ this is equivalent to the convergence of a certain Ces\`{a}ro mean of $\gamma_N(T)$ \cite[Chapter~4, Section~2, Proposition~6]{connes-94}. In particular, for a $T\geqslant 0$ such that $\gamma_N(T)$ is convergent, one has that $T$ is measurable and \[
 {\Tr}_{{\rm Dix}}(T) : =
 \lim_{N\to\infty}\left(\frac{1}{\log(N)}\sum_{n=0}^{N-1}\mu_n(T)\right) ,\]
independently of the election of the generalized scale-invariant limit $\omega$.
The set of measurable operators $\rr{S}^{1^+}_{\rm m}$ is
 a closed subspace of
$\rr{S}^{1^+}$ (but not an ideal) which is invariant under conjugation by bounded invertible operators.
Evidently, $\rr{S}^{1^+}_0\subset\rr{S}^{1^+}_{\rm m}$.

Let $\varepsilon>-1$ and define the operator
\[
Q_{B,\varepsilon} := Q_B + \varepsilon{\bf 1}.
\]
Since $Q_{B,\varepsilon}$ is strictly positive, the inverse powers $Q_{B,\varepsilon}^{-s}$ are well
defined for all $s>0$. The following facts have been proved in~\cite[Lemmas~B.4 and~B.5]{denittis-gomi-moscolari-19}.\footnote{In order to adapt the proofs of~\cite[Lemmas~B.4 and~B.5]{denittis-gomi-moscolari-19} to Proposition~\ref{props:old_resul} it is enough to set $\varepsilon=2\xi+1$ and to observe that the arguments of~\cite[Lemmas~B.4 and~B.5]{denittis-gomi-moscolari-19} are still valid for $\xi>-1$.}
\begin{Proposition}\label{props:old_resul}
Let $Q_{B,\varepsilon}^{-s}$ be defined as above. Then:
\begin{itemize}\itemsep=0pt\label{eq:comp_Trac_Q}
\item[$(1)$] $Q_{B,\varepsilon}^{-s}\in\rr{S}^{1}$ for every $s>2$ and $\varepsilon>-1$;
\item[$(2)$] $Q_{B,\varepsilon}^{-2}\in \rr{S}^{1^+}_{\rm m}$ for every $\varepsilon>-1$, and
\begin{equation*}
{\Tr}_{\rm Dix}\big(Q_{B,\varepsilon}^{-2}\big) = \frac{1}{2} ,
\end{equation*}
independently of $\varepsilon$;
\item[$(3)$] let $\Pi_j$ be the $j$-th Landau projection, then $Q_{B,\varepsilon}^{-s}\Pi_j\in\rr{S}^{1}$ for every $s>1$ and $\varepsilon>-1$;
\item[$(4)$] $Q_{B,\varepsilon}^{-1}\Pi_j\in \rr{S}^{1^+}_{\rm m}$ for every $\varepsilon>-1$
and
\begin{equation}\label{eq:traXXX_II}
{\Tr}_{\rm Dix}\big(Q_{B,\varepsilon}^{-1}\Pi_j\big) = 1 ,
\end{equation}
independently of $\varepsilon$.
\end{itemize}
\end{Proposition}

Since $Q_B$ commutes with $\Pi_j$, it follows that \[
Q_{B,\varepsilon}^{-1}\Pi_j = Q_{B,\varepsilon}^{-\frac{1}{2}}\Pi_jQ_{B,\varepsilon}^{-\frac{1}{2}} = \Pi_jQ_{B,\varepsilon}^{-1} .
\]
Then, the order of the operators $Q_{B,\varepsilon}^{-\frac{1}{2}}$ and $\Pi_j$ in~\eqref{eq:traXXX_II} is irrelevant.
\begin{Corollary}\label{corol:dix_trac1}
Let $\Upsilon_{j\mapsto k}$, with $k,j\in\N_0$, be the operators defined by~\eqref{eq:def_Ups}. Then $Q_{B,\varepsilon}^{-1}\Upsilon_{j\mapsto k}$, $Q_{B,\varepsilon}^{-\frac{1}{2}}\Upsilon_{j\mapsto k}Q_{B,\varepsilon'}^{-\frac{1}{2}}$ and $\Upsilon_{j\mapsto k}Q_{B,\varepsilon}^{-1}$ are elements of $\rr{S}^{1^+}_{\rm m}$, and
\begin{equation*}
{\Tr}_{\rm Dix}\big(Q_{B,\varepsilon}^{-1}\Upsilon_{j\mapsto k}\big) = {\Tr}_{\rm Dix}\big(\Upsilon_{j\mapsto k}Q_{B,\varepsilon}^{-1}\big)
= {\Tr}_{\rm Dix}\big(Q_{B,\varepsilon}^{-\frac{1}{2}}\Upsilon_{j\mapsto k}Q_{B,\varepsilon'}^{-\frac{1}{2}}\big) = \delta_{j,k}
\end{equation*}
independently of $\varepsilon,\varepsilon'>-1$.
\end{Corollary}
\begin{proof}
 Proposition~\ref{prop:struct}(4) implies that
 \[ Q_{B,\varepsilon}^{-1}\Upsilon_{j\mapsto k}=\big(Q_{B,\varepsilon}^{-1}\Pi_k\big)\Upsilon_{j\mapsto k} \qquad \text{and} \qquad \Upsilon_{j\mapsto k}Q_{B,\varepsilon}^{-1}=\Upsilon_{j\mapsto k}\big(Q_{B,\varepsilon}^{-1}\Pi_j\big).
 \] Proposition~\ref{props:old_resul}(4) and the ideal property of the Dixmier ideal imply that both $Q_{B,\varepsilon}^{-1}\Upsilon_{j\mapsto k}$ and $\Upsilon_{j\mapsto k}Q_{B,\varepsilon}^{-1}$ are elements of $\rr{S}^{1^+}$.
The trace property of any $\omega$-Dixmier trace implies
\begin{gather*}
{\Tr}_{{\rm Dix},\omega}\big(Q_{B,\varepsilon}^{-1}\Upsilon_{j\mapsto k}\big) = {\Tr}_{{\rm Dix},\omega}\big(Q_{B,\varepsilon}^{-1}\Pi_k\Upsilon_{j\mapsto k}\Pi_j\big)
 = {\Tr}_{{\rm Dix},\omega}\big(Q_{B,\varepsilon}^{-1}\Pi_j\Pi_k\Upsilon_{j\mapsto k}\big) = \delta_{j,k},
\end{gather*}
where the last equality follows from $\Pi_j\Pi_k=0$ if $j\neq k$ and $\Upsilon_{j\mapsto j}=\Pi_j$. The independence of the result from the choice of the
{generalized scale-invariant limit} $\omega$ implies that
$Q_{B,\varepsilon}^{-1}\Upsilon_{j\mapsto k}\in \rr{S}^{1^+}_{\rm m}$. The proof for
$\Upsilon_{j\mapsto k}Q_{B,\varepsilon}^{-1}$ uses the same argument.
For the last case we can use the identity
\[
Q_{B,\varepsilon}^{-\frac{1}{2}}\Upsilon_{j\mapsto k}Q_{B,\varepsilon'}^{-\frac{1}{2}} = Q_{B,\varepsilon}^{-1}\Pi_k\big(Q_{B,\varepsilon}^{\frac{1}{2}}\Upsilon_{j\mapsto k}Q_{B,\varepsilon'}^{-\frac{1}{2}}\big) .
\]
Then, to prove that the left-hand side is in $\rr{S}^{1^+}$ it is sufficient to prove
that the product in the round brackets on the right-hand side defines a bounded operator.
This follows by observing that
\[
Q_{B,\varepsilon}^{\frac{1}{2}}\Upsilon_{j\mapsto k}Q_{B,\varepsilon'}^{-\frac{1}{2}} = \left(\sum_{m\in\N_0}\sqrt{\frac{m+(k+1+\varepsilon)}{{m+(j+1+\varepsilon')}}}P_m\right) \Upsilon_{j\mapsto k},
\]
where $P_m$ are the dual Landau projections~\eqref{eq:Lan_proj-dual}. Finally, the trace property of any $\omega$-Dixmier trace implies
\begin{align*}
{\Tr}_{\text{Dix},\omega}\big(Q_{B,\varepsilon}^{-\frac{1}{2}}\Upsilon_{j\mapsto k}Q_{B,\varepsilon'}^{-\frac{1}{2}}\big) &= {\Tr}_{{\rm Dix},\omega}\big(\Pi_kQ_{B,\varepsilon}^{-\frac{1}{2}}\Upsilon_{j\mapsto k}Q_{B,\varepsilon'}^{-\frac{1}{2}}\Pi_j\big)\\
&= {\Tr}_{{\rm Dix},\omega}\big(\Pi_j\Pi_kQ_{B,\varepsilon}^{-\frac{1}{2}}\Upsilon_{j\mapsto k}Q_{B,\varepsilon'}^{-\frac{1}{2}}\big) = \delta_{j,k} ,
\end{align*}
independently of $\omega$. This concludes the proof.
\end{proof}

For the next crucial result, which strongly relies on Corollary~\ref{corol:dix_trac1}, we will need to recall the definition of the space $\bb{L}^1_B$ given at the end of Section~\ref{sec:C-magn_alg}.
\begin{Proposition}\label{prop:trace_dix_main res}
Let $L_g\in\bb{L}^1_B$ with $g=\sum_{(n,m)\in\N_0^2}g_{n,m} \psi_{n,m}$ and $\{g_{n,m}\}\in \ell^1\big(\N_0^2\big)$.
 Then $Q_{B,\varepsilon}^{-1}L_g$, $Q_{B,\varepsilon}^{-\frac{1}{2}}L_gQ_{B,\varepsilon'}^{-\frac{1}{2}}$ and $L_gQ_{B,\varepsilon}^{-1}$ are elements of $\rr{S}^{1^+}_{\rm m}$, and
\begin{gather*}
 {\Tr}_{\rm Dix}\big(Q_{B,\varepsilon}^{-1}L_g\big) = {\Tr}_{\rm Dix}\big(L_gQ_{B,\varepsilon}^{-1}\big) = {\Tr}_{\rm Dix}\big(Q_{B,\varepsilon}^{-\frac{1}{2}}L_gQ_{B,\varepsilon'}^{-\frac{1}{2}}\big)
 = g(0) = \fint_B(L_g)
\end{gather*}
independently of $\varepsilon,\varepsilon'>-1$.
\end{Proposition}

\begin{proof}
Before proceeding with the proof, let us observe that $g\in C_0\big(\R^2\big)$ in view of Remark~\ref{lemm:pre_inclus} and the inclusion~\eqref{eq:L^p-inclus}. Then the quantity $g(0)$ is well defined. Moreover, observing that $\psi_{n,m}(0)=\big(\sqrt{2\pi}\ell_B\big)^{-1}\delta_{n,m}$ one gets
\[
g(0) = \frac{1}{\sqrt{2\pi}\ell_B}\sum_{n\in\N_0}g_{n,n} .
\]
Let us start by proving that
\[
\big\|Q_{B,\varepsilon}^{-1}\Upsilon_{j\mapsto k}\big\|_{1^+} \leqslant 1 ,\qquad\big\|\Upsilon_{j\mapsto k}Q_{B,\varepsilon}^{-1}\big\|_{1^+} \leqslant 1 ,\qquad\big\|Q_{B,\varepsilon}^{-\frac{1}{2}}\Upsilon_{j\mapsto k}Q_{B,\varepsilon'}^{-\frac{1}{2}}\big\|_{1^+} \leqslant 1
\]
for all $(j,k)\in\N_0^2$. From property~(3) of Proposition~\ref{prop:struct} and using the spectral decomposition of $Q_{B,\varepsilon}$ in terms of the Landau projections $\Pi_j$ and the dual Landau projections $P_m$, one gets
\[
\big|\Upsilon_{j\mapsto k}Q_{B,\varepsilon}^{-1}\big|^2 = Q_{B,\varepsilon}^{-1}\Pi_jQ_{B,\varepsilon}^{-1} = \left(\sum_{m\in\N_0}\frac{P_m}{(m+j+1+\varepsilon)^2}\right)\Pi_j .
\]
Therefore, the singular values of $\Upsilon_{j\mapsto k}Q_{B,\varepsilon}^{-1}$ are $(m+j+1+\varepsilon)^{-1}$ and
\[
\big\|Q_{B,\varepsilon}^{-1}\Upsilon_{j\mapsto k}\big\|_{1^+} = \sup_{N>1}\frac{1}{\log(N)}\sum_{m=0}^{N-1}\frac{1}{m+(j+1+\varepsilon)} \leqslant 1 .
\]
The proof of
\[
\big\|\Upsilon_{j\mapsto k}Q_{B,\varepsilon}^{-1}\big\|_{1^+} \leqslant 1
\]
can be deduced from the equality $\Upsilon_{j\mapsto k}Q_{B,\varepsilon}^{-1}=\big(Q_{B,\varepsilon}^{-1}\Upsilon_{k\mapsto j}\big)^*$, along with the fact that the singular values of an operator and its adjoint coincide. For the remaining case we can look again at the spectral decomposition obtaining
\begin{equation}\label{eq:aux_1ref}
\big|Q_{B,\varepsilon}^{-\frac{1}{2}}\Upsilon_{j\mapsto k}Q_{B,\varepsilon'}^{-\frac{1}{2}}\big|^2 = \left(\sum_{m\in\N_0}\frac{P_m}{(m+k+1+\varepsilon)(m+j+1+\varepsilon')}\right)\Pi_j .
\end{equation}
Therefore, one obtains
\[
\big\|Q_{B,\varepsilon}^{-\frac{1}{2}}\Upsilon_{j\mapsto k}Q_{B,\varepsilon'}^{-\frac{1}{2}}\big\|_{1^+} = \sup_{N>1}\frac{1}{\log(N)}\sum_{m=0}^{N-1}\frac{1}{m\sqrt{1+\frac{1}{m}a+\frac{1}{m^2}b}} \leqslant 1 ,
\]
where $a:=k+j+2+\varepsilon+\varepsilon'$ and $b:=(k+1+\varepsilon)(j+1+\varepsilon')$ are positive constants. Now, let
$g_T=\sum_{j=0}^{N}\sum_{k=0}^{M}t_{k,j}\psi_{k,j}$ be a finite linear combination of
generalized Laguerre functions $\psi_{k,j}$ and
\[
T = \pi(g_T) = \frac{1}{\sqrt{2\pi}\ell_B}\sum_{j=0}^{N}\sum_{k=0}^{M}(-1)^{j-k}t_{k,j}\Upsilon_{j\mapsto k} \in \bb{F}_B
\]
the related operator. Then, one obtains that
\[
\big\|Q_{B,\varepsilon}^{-1}T\big\|_{1^+} \leqslant \frac{1}{\sqrt{2\pi}\ell_B}\sum_{j=0}^{N}\sum_{k=0}^{M}|t_{k,j}| = \frac{1}{\sqrt{2\pi}\ell_B}\|\{t_{k,j}\}\|_{\ell^1} .
\]
The last inequality, along with the density of $\bb{F}_B$ in $\bb{L}^1_B$ with respect to the norm induced by $\ell^1\big(\N_0^2\big)$,
implies the continuity of the linear map
\[
\bb{L}^1_B \ni T \longmapsto Q_{B,\varepsilon}^{-1}T \in \rr{S}^{1^+} .
\]
Exactly in the same way one can prove the continuity of the maps
\[
\bb{L}^1_B \ni T \longmapsto TQ_{B,\varepsilon}^{-1} \in \rr{S}^{1^+} ,
\]
and
\[
\bb{L}^1_B \ni T \longmapsto Q_{B,\varepsilon}^{-\frac{1}{2}}TQ_{B,\varepsilon'}^{-\frac{1}{2}} \in \rr{S}^{1^+} .
\]
Let us consider now a generic element $L_g\in \bb{L}^1_B$ identified with the expansion~\eqref{eq:exp_oper_Ups}. The linearity and the continuity of the Dixmier trace with respect to the Calder\'on norm provides
\begin{align*}
{\Tr}_{{\rm Dix},\omega}\big(Q_{B,\varepsilon}^{-1}L_g\big) &=
\frac{1}{\sqrt{2\pi}\ell_B}\sum_{(n,m)\in\N_0^2}(-1)^{m-n}g_{n,m}{\Tr}_{{\rm Dix}}\big(Q_{B,\varepsilon}^{-1}\Upsilon_{m\mapsto n}\big)\\
&= \frac{1}{\sqrt{2\pi}\ell_B}\sum_{n\in\N_0}g_{n,n} = g(0)
\end{align*}
in view of Corollary~\ref{corol:dix_trac1}.
In particular the result is independent of $\omega$ (and of $\varepsilon$) proving that $Q_{B,\varepsilon}^{-1}L_g$ is a measurable operator.
Finally, the equality with the trace $\fint_B$ follows from Corollary~\ref{prop:comput_trace} and the inclusion $\bb{L}^1_B\subset \big(\bb{L}^2_B\big)^2$ proved in Proposition~\ref{prop:very_inclus}.
 The trace of $L_gQ_{B,\varepsilon}^{-1}$ and
$Q_{B,\varepsilon}^{-\frac{1}{2}}L_gQ_{B,\varepsilon'}^{-\frac{1}{2}}$ can be computed
following the same argument.
\end{proof}

A first consequence of Proposition~\ref{prop:trace_dix_main res} can be deduced from the identity
\[
Q_{B,\varepsilon}^{-\frac{1}{2}}TQ_{B,\varepsilon'}^{-\frac{1}{2}} - Q_{B,\varepsilon}^{-1}T = Q_{B,\varepsilon}^{-\frac{1}{2}} \big[T,Q_{B,\varepsilon}^{-\frac{1}{2}}\big] .
\]
In the left-hand side there is the difference of two elements of
$\rr{S}^{1^+}_{\rm m}$ with same trace. Therefore, the element in the right-hand side is in $\rr{S}^{1^+}_{\rm m}$ and has vanishing Dixmier trace (by linearity). Thus we proved that
\begin{equation*}
Q_{B,\varepsilon}^{-\frac{1}{2}} \big[T,Q_{B,\varepsilon}^{-\frac{1}{2}}\big] \in \rr{S}^{1^+}_0 ,\qquad \forall\, T\in\bb{L}^1_B .
\end{equation*}
A similar argument also shows that
\begin{equation}\label{eq:null_dix_comm_II}
Q_{B,\varepsilon}^{-1} T - T Q_{B,\varepsilon'}^{-1} \in \rr{S}^{1^+}_0 ,\qquad \forall\, T\in\bb{L}^1_B
\end{equation}
for every $\varepsilon,\varepsilon'>-1$.

Although the content of Proposition~\ref{prop:trace_dix_main res} is sufficient for all the applications in Section~\ref{sec:magn_spec_trip}, we find it is important to extend its validity to a domain larger than $\bb{L}^1_B$.
The following result is proved in
Appendix~\ref{appB}.
\begin{Theorem}\label{prop:trace_dix_main res_bis}
Let $T=L_{f}^*L_g$ with $L_f\in \bb{L}^1_B$ and $L_g\in\bb{L}^2_B$. Then
$Q_{B,\varepsilon}^{-1}T\in\rr{S}^{1^+}_{\rm m}$ and the equality
\[
{\Tr}_{\rm Dix}\big(Q_{B,\varepsilon}^{-1}T\big) = \langle f,g \rangle_B =
\fint_B(T)
\]
holds true, independently of $\varepsilon>-1$.
\end{Theorem}

\begin{Remark}[relation with the trace per unit of volume II]\label{rk:trac_unit_volII}
By combining Remark~\ref{rk:trac_unit_volI} with the content of Theorem~\ref{prop:trace_dix_main res_bis}
 we get that
\begin{equation}\label{eq:impossib_eq}
\begin{aligned}
 \s{T}_B(T) &= \frac{1}{2\Lambda_B} {\Tr}_{\rm Dix}\big(Q_{B,\varepsilon}^{-1}T\big)
\end{aligned}
\end{equation}
 for every $\varepsilon>-1$ and for every $T$ in the set
\[
 \bb{L}^{1}_B\cdot \bb{L}^{2}_B := \big\{T=L_{f}L_g \,|\, L_f\in \bb{L}^1_B , L_g\in\bb{L}^2_B \big\} .
\]
 This result generalizes the validity of~\cite[Theorem~B.2]{denittis-gomi-moscolari-19}. However, we are firmly convinced
 that equation~\eqref{eq:impossib_eq} must be valid on the entire domain of definition $\big(\bb{L}^2_B\big)^2$ of the trace~$\s{T}_B$. Unfortunately we were not able to obtain such a generalization, which remains as an open problem for the time being (see a further comment at the end of Appendix~\ref{appB}).
\end{Remark}

\subsection{Differential structure}\label{sec:diff_struct}
The algebra $\bb{C}_B$ supports the action of two natural $\ast$-derivations,
which endow the magnetic operators
with a noncommutative differential calculus.

The Banach $\ast$-algebra $\s{L}_B^1$ can be endowed with an $\R^2$-action of automorphisms $\R^2\ni k\mapsto \widehat{\alpha}_k\in{\rm Aut}\big(\s{L}_B^1\big)$ defined by
\[
 \widehat{\alpha}_k(f)(x) := \expo{\ii k\cdot x}f(x) ,\qquad f\in \s{L}_B^1 ,
\]
where $k\cdot x:=k_1x_1+k_2x_2$. As a consequence of the Lebesgue's dominated convergence theorem, this action is strongly continuous, i.e.,
\[
 \lim_{k\to 0} \nnorm{\widehat{\alpha}_k(f)-f}_{B,1} = 0 , \qquad \forall\, f\in \s{L}_B^1 .
\]
Moreover, it can be differentiated with respect to the norm-topology of $\s{L}_B^1$ and produces
 two (unbounded) $\ast$-derivations
\[
 ( \partial_j(f))(x) := \left.\frac{\partial f}{\partial k_j}\right|_{k=0}(x) = \ii x_jf(x) ,\qquad j=1,2,
\]
which are densely defined~\cite[Proposition 3.1.6]{bratteli-robinson-87}. For
$f$ and $g$ in the domain of definition of the derivations, a direct computation shows that:
{\begin{itemize}\itemsep=0pt
\item[(1)] $\partial_{j} (f * g) = \partial_{j} (f) * g + f * \partial_{j}( g)$, $j=1,2$;
\item[(2)] $\partial_{j} (f^{*}) = {(\partial_{j}( f))}^{*}$, $j=1,2$;
\item[(3)] $\partial_{1} \partial_{2} (f) = \partial_{2} \partial_{1} (f)$.
\end{itemize}}
Property (1) is the Leibniz rule for derivations. Property~(2) says that the derivations are compatible with the involution of the algebra. Finally, property~(3) says that both derivations commute.
Let $\s{D}^{(n,m)}_B\subset \s{L}_B^1$ be the domain of the iterated derivation $\partial_{1}^n \partial_{2}^m$. It is straightforward to check that
\[
\s{S}_{B} \subset \bigcap_{(n,m)\in\N_{0}^2}\s{D}^{(n,m)}_B ,
\]
namely $\s{S}_{B}$ is comprised of elements that are \emph{smooth} with respect to the $\R^2$-action $\widehat{\alpha}_k$.

The $\R^2$-action can be defined directly in the representation on the Hilbert space $L^2\big(\R^2\big)$.
For every $k\in\R^2$, let ${\alpha}_k$ be the $\ast$-automorphism of $\bb{B}\big(L^2\big(\R^2\big)\big)$ defined by
\begin{equation}\label{eq:R2act}
{\alpha}_k(A) := \expo{-\ii k\cdot x} A \expo{\ii k\cdot x} ,\qquad A\in \bb{B}\big(L^2\big(\R^2\big)\big),
\end{equation}
where $k\cdot x=k_1x_1+k_2x_2$ is meant as the linear combination of the position operators on
$L^2\big(\R^2\big)$.
\begin{Proposition}\label{prop:R-flow_C^*}
The following facts hold true:
\begin{itemize}\itemsep=0pt
\item[$(1)$] $\pi(\widehat{\alpha}_k(f))={\alpha}_k(\pi(f))$ for all $f\in\s{L}_B^1$;
\item[$(2)$] the map $\R^2\ni k\mapsto {\alpha}_k\in{\rm Aut}(\bb{C}_B)$ extends to a strongly continuous
$\R^2$-action by automorphisms of the $C^*$-algebra $\bb{C}_B$.
\end{itemize}
\end{Proposition}
\begin{proof}
(1) follows from a direct computation. (2)
The norm-density of $\pi\big(\s{L}_B^1\big)$ in $\bb{C}_B$ assures that ${\alpha}_k(T)\in\bb{C}_B$ whenever $T\in \bb{C}_B$. The inequality
\[
 \lVert {\alpha}_k(\pi(f))-\pi(f)\rVert \leqslant \nnorm{\widehat{\alpha}_k(f)-f}_{B,1} ,\qquad f\in\s{L}_B^1 ,
\]
and the density of $\pi(\s{L}_B^1)$ imply that
\[
\lim_{k\to0}\lVert{\alpha}_k(A)-A\rVert = 0\qquad \forall\, A\in \bb{C}_B ,
\]
i.e., the strongly continuity of the $\R^2$-action
 $k\mapsto {\alpha}_k$ on $\bb{C}_B$.
\end{proof}

The $\R^2$-action $k\mapsto\alpha_k$
can be differentiated with respect to the topology of $\bb{C}_B$ and produces
 two (unbounded) commuting $\ast$-derivations
 $\nabla_1$ and $\nabla_2$
 which are densely defined~\cite[Proposition~3.1.6]{bratteli-robinson-87}. Let $\bb{D}^{(n,m)}_B$ be the domain of $\nabla_1^n\nabla_2^m$. As a consequence of
 Proposition~\ref{prop:R-flow_C^*}, one can prove that $\pi\big(\s{D}^{(n,m)}_B\big)\subset \bb{D}^{(n,m)}_B$
 and
 \[
 \pi\big(\partial_{1}^n \partial_{2}^m(f)\big) = \nabla_1^n\nabla_2^m(\pi(f)) ,\qquad f\in \s{D}^{(n,m)}_B .
 \]
 In particular
 \[
 \bb{S}_B \subset \bigcap_{(n,m)\in\N^2}\bb{D}^{(n,m)}_B .
\]
 Let $\bb{C}^{N}_B:=\bigcap_{n+m\leqslant N}\bb{D}^{(n,m)}_B$. These are Banach spaces obtained as the completion
 \[
 \bb{C}^{N}_B = \overline{\bb{S}_B}^{ \lVert \ \rVert_{N}}
 \]
 with respect to the norm
\begin{equation}\label{eq:diff_norm}
 \lVert A \rVert_N := \sum_{n+m\leqslant N}\|\nabla_1^n\nabla_2^m(A)\| .
 \end{equation}
Let $\bb{C}^{\infty}_B:=\bigcap_{N\in\N}\bb{C}^{N}_B$ be the subalgebra of \emph{smooth elements} of $\bb{C}_B$.
\begin{Proposition}
$\bb{C}_{B}^\infty$ is a non-unital pre-$C^*$-algebra of $\bb{C}_{B}$.
\end{Proposition}
\begin{proof}
By construction $\bb{C}_{B}^\infty$ is the subalgebra of {smooth elements} of $\bb{C}_B$
with respect to the strongly continuous action of the Lie group~$\R^2$. Then, the result follows from~\cite[Proposition~3.45]{gracia-varilly-figueroa-01}.
\end{proof}

Consider an element $A\in\bb{C}_B$ such that
 $A[\s{D}(x_j)]\subseteq \s{D}(x_j)$ where $\s{D}(x_j)$ is the domain of the position operator $x_j$. In this case the difference
$
x_jA-Ax_j
$ is initially well defined on the dense set
$\s{D}(x_j)$ and is closable. Therefore, it uniquely defines a closed operator, the \emph{commutator}, that will be denoted with $[x_j,A]$. Let
\[
\bb{C}^{1,0}_{B} := \{A\in\bb{C}_B \,|\, A[\s{D}(x_j)]\subseteq \s{D}(x_j) , [x_j,A]\in\bb{C}_B , j=1,2 \} .
\]
\begin{Proposition}\label{props:commut_deriv}
The space $\bb{C}^{1,0}_{B}\subset \bb{D}^1_{B}$ is a core for $\nabla_1$ and $\nabla_2$ and, for every $A\in \bb{C}^{1,0}_{B}$, the equalities
\begin{equation*}
\nabla_j(A) = -\ii [x_j,A] ,\qquad j=1,2
\end{equation*}
hold true. Moreover, one has that
\begin{equation*}
\nabla_1(A) = -\ii\ell_B [K_2,A] ,\qquad \nabla_2(A) = \ii\ell_B [K_1,A] ,
\end{equation*}
where $K_j$ are the {magnetic momenta}~\eqref{eq:momenta}.
\end{Proposition}
\begin{proof} The first part of the claim which states that $\bb{C}^{1,0}_{B}$ is a
core and the equality $\nabla_j(A)=-\ii [x_j,A]$ can be proved as in~\cite[Theorem~7.3]{pagter-sukochev-07}. The second equality $\nabla_j(A)=-\ii [K_j,A]$ follows from
Lemma~\ref{lemma:comm_rel_bpm} and the equalities $x_1=\ell_B(K_2-G_1)$ and
$x_2=\ell_B(G_2-K_1)$. \end{proof}

It is useful to have criteria to establish when an operator sits inside $\bb{C}^{1,0}_{B}$.
\begin{Proposition}\label{prop:criter_commut}
Let $\langle x \rangle$ denotes the Japanese bracket~\eqref{eq:jap_brack}. The following facts hold true:
\begin{itemize}\itemsep=0pt
\item[$(1)$] if $f\in L^1\big(\R^2\big)$ and $\langle x \rangle f\in L^1\big(\R^2\big)$ then $\pi(f)\in \bb{C}^{1,0}_{B}$;
\item[$(2)$] $\bb{S}_B\subset \bb{C}^{1,0}_{B}$;
\item[$(3)$] if $f\in L^2\big(\R^2\big)$ and $\langle x \rangle f\in L^2\big(\R^2\big)$ then $L_f\in \bb{C}^{1,0}_{B}$.
\end{itemize}
\end{Proposition}
\begin{proof}(1) Let $\phi\in \s{D}(x_j)$. Then a straightforward computation shows that
\[
 x_j(\pi(f)\phi) = \pi(f)\widetilde{\phi} -
 \pi\big(\widetilde{f}\big){\phi} ,
\]
where $\widetilde{\phi}(x):=x_j\phi(x)$ is a vector in $L^2\big(\R^2\big)$ and $\widetilde{f}(x):=x_jf(x)$ is in $L^1\big(\R^2\big)$ as a consequence of the respective assumptions. This shows that $x_j(\pi(f)\phi)$ is well defined, namely $\pi(f)\phi\in \s{D}(x_j)$. Moreover, the same computation shows that $[x_j,\pi(f)]=\pi(\widetilde{f})\in \bb{C}_B$. Then the conditions for $f\in \bb{C}^{1,0}_{B}$
are satisfied. (2) is a direct consequence of (1) since the space of Schwartz
functions is stable under the multiplication by $\langle x \rangle$. (3) follows from the
same argument used to prove~(1) by replacing $\pi(f)$ and $\pi(\widetilde{f})$ with $L_f$ and $L_{\widetilde{f}}$, respectively. In particular one gets that $[x_j,L_f]=L_{\widetilde{f}}\in \bb{L}^2_B$.
\end{proof}

In view of Proposition~\ref{prop:criter_commut} one gets that $\Upsilon_{j\mapsto k}\in \bb{C}^{1,0}_{B}$. Therefore, the derivatives of
$\Upsilon_{j\mapsto k}$ can be computed by the commutator as in Proposition~\ref{props:commut_deriv}. By observing that
$K_1=2^{-\frac{1}{2}}(\rr{a}^++\rr{a}^-)$ and
$K_2=-\ii2^{-\frac{1}{2}}(\rr{a}^+-\rr{a}^-)$
and using the relations in Proposition~\ref{prop:struct} (and in particular the commutation relations~\eqref{eq:commut_ups}) one gets
\begin{gather}\label{eq:comm_ups_01}
\nabla_1(\Upsilon_{j\mapsto k}) = \frac{\ell_B}{\sqrt{2}}\big( \sqrt{k}\Upsilon_{j\mapsto k-1}+\sqrt{j}\Upsilon_{j-1\mapsto k}
 -\sqrt{k+1}\Upsilon_{j\mapsto k+1}-\sqrt{j+1}\Upsilon_{j+1\mapsto k}\big) ,
\end{gather}
and
\begin{gather}\label{eq:comm_ups_02}
 \nabla_2(\Upsilon_{j\mapsto k}) = \ii\frac{\ell_B}{\sqrt{2}}\big(\sqrt{k}\Upsilon_{j\mapsto k-1}-\sqrt{j}\Upsilon_{j-1\mapsto k}
 +\sqrt{k+1}\Upsilon_{j\mapsto k+1}-\sqrt{j+1}\Upsilon_{j+1\mapsto k}\big) .
\end{gather}
In the special case $\Upsilon_{k\mapsto k}=\Pi_k$, the equations~\eqref{eq:comm_ups_01}
and~\eqref{eq:comm_ups_02} show how to compute the derivations of the Landau projections.

The $\R^2$-action~\eqref{eq:R2act} extends to the von Neumann algebra $\bb{M}_B$.
\begin{Proposition}\label{prop:R-flow_W^*}
The family of automorphisms~\eqref{eq:R2act} defines an ultra-weakly continuous action
$\R^2\ni k\mapsto {\alpha}_k\in{\rm Aut}(\bb{M}_B)$ of $\R^2$ on the von Neumann algebra $\bb{M}_B$. Moreover
\[
\fint_B({\alpha}_k(T)) = \fint_B(T) ,\qquad \forall\, T\in\big(\bb{L}^2_B\big)^2 .
\]
\end{Proposition}
\begin{proof}
The density of $\bb{C}_B$ in $\bb{M}_B$ with respect to the strong (weak) topology implies that ${\alpha}_k(T)\in\bb{M}_B$ whenever $T\in \bb{M}_B$.
This proves that every ${\alpha}_k$ is an automorphism of the von Neumann algebra $\bb{M}_B$. Since the maps $k\mapsto \expo{\pm\ii k\cdot x}$ are strongly continuous, and since multiplication on norm-bounded subsets of $\bb{B}\big(L^2\big(\R^2\big)\big)$ is strongly continuous, it follows that ${\alpha}_k(T)\to0$ when $k\to0$ with respect to the strong topology for all $T\in \bb{M}_B$.
 Moreover, on a norm-bounded subset of $\bb{B}\big(L^2\big(\R^2\big)\big)$ the strong and ultra-strong topology coincide, and so ${\alpha}_k(T)\to0$ ultra-strongly, and hence ultra-weakly, as $t\to 0$ for all $T\in \bb{M}_B$.
This shows that the $\R^2$-action $k\mapsto {\alpha}_k$ is ultra-weakly continuous on $\bb{M}_B$. Let $T=L_f^*L_g\in \big(\bb{L}^2_B\big)^2$. A direct computation shows that
${\alpha}_k(T)=L_{{f'}}^*L_{{g'}}$ where ${f}'(x):=\expo{\ii k\cdot x}f(x)$ and ${g}'(x):=\expo{\ii k\cdot x}g(x)$. The chain of equalities
\[
\fint_B({\alpha}_k(T)) = \langle{f}',{g}'\rangle_B = \langle{f},{g}\rangle_B = \fint_B(T)
\]
concludes the proof.
\end{proof}

Using the jargon of~\cite{pagter-sukochev-07},
Proposition~\ref{prop:R-flow_W^*} says that
$k\mapsto {\alpha}_k$ is an \emph{$\R^2$-flow} on the pair $\big(\bb{M}_B,\fint_B\big)$.
Thus,~\cite[Proposition 4.1]{pagter-sukochev-07} guarantees that every ${\alpha}_k$ extends to a $\ast$-isometry in every Banach space
 $\rr{L}_B^p$ and the map $\R^2\ni k\mapsto {\alpha}_k\in{\rm Iso}\big(\rr{L}_B^p\big)$ is strongly continuous, i.e.,
 \[
 \lim_{k\to 0} \nnorm{{\alpha}_k(T)-T}_{B,p} = 0\qquad \forall\, T\in \rr{L}_B^p .
 \]
The $\R^2$-flow can be differentiated in every space $\rr{L}_B^p$ (with the convention $\rr{L}_B^\infty=\bb{M}_B$) with respect to the relative topology. We will denote the respective generators simply by $\nabla_1$ and $\nabla_2$ without further reference to the particular space $\rr{L}_B^p$. It results that
$\nabla_1$ and $\nabla_2$ have a dense domain in every space $\rr{L}_B^p$~\cite[Proposition~3.1.6]{bratteli-robinson-87} with a dense core in $\rr{L}_B^p\cap\bb{M}_B$ where they can be represented as the commutators with the position operators~\cite[Theorem~7.3]{pagter-sukochev-07}.
From the invariance of the trace $\fint_B$ under the $\R^2$-flow ${\alpha}_k$ it immediately follows that
\begin{equation*}
\fint_B \circ \nabla_j = 0 ,\qquad j=1,2 .
\end{equation*}
This property, along with the Leibniz rule for derivations, provides that
\begin{equation*}
\fint_B(T\nabla_j(S)) = -\fint_B(\nabla_j(T)S) ,\qquad j=1,2,
\end{equation*}
whenever $T$, $S$ and $TS$ are differentiable and
$T\nabla_j(S)$ and $\nabla_j(T)S$ are trace-class with respect to $\fint_B$.
It is also possible to define the \emph{noncommutative Sobolev spaces} $\rr{W}_B^{N,p}\subset \rr{L}_B^p$ as the closure of the set of smooth elements with respect to the Sobolev norm
\begin{equation}\label{eq:sob_norm}
 \nnorm{T}_{B,N,p} := \left(\sum_{n+m\leqslant N}\nnorm{\nabla_1^n\nabla_2^m(T)}_{B,p}^p\right)^{\frac{1}{p}} .
\end{equation}
In the rest of this work we will often use the notation
\[
\nabla(T) := (\nabla_1(T),\nabla_2(T))
\]
for the (noncommutative) \emph{gradient} of $T\in\rr{W}_B^{1,p}$.

\section{The magnetic spectral triple}\label{sec:magn_spec_trip}
In this section, we will introduce a compact spectral triple suitable for the study of the geometry of the magnetic $C^*$-algebra $\bb{C}_B$. The key element is a Dirac operator $D_B$ which is, in a sense that will be made more precise below, the square root of the
harmonic oscillator $Q_B$ defined in Section~\ref{sec:landau_ham}. We will prove
that our spectral triple provides a good example of a noncommutative manifold
of spectral dimension 2. Moreover, we will prove the continuous version of the
first Connes' formula appearing in~\cite{bellissard-elst-schulz-baldes-94}.

\subsection{Construction of the spectral triple}\label{sec:spectral-triple}
For the standard definitions concerning the theory of spectral triples we will refer to the classical textbooks~\cite{connes-94,gracia-varilly-figueroa-01}. In particular we will take~\cite[Definition~9.16]{gracia-varilly-figueroa-01} as the recipe for the construction of our spectral triple.

We will start with the $C^*$-algebra of magnetic operators $\bb{C}_B$ endowed with the pre-$C^*$-algebra $\bb{S}_B\subset\bb{C}_B$. We will represent $\bb{C}_B$ on the separable Hilbert space
\[
\s{H}_4 := L^2\big(\R^2\big) \otimes \C^4
\]
through the \emph{diagonal} representation
$\rho\colon \bb{C}_B\to \bb{B}(\s{H}_4)$ defined by
\[
 \rho(A) := A \otimes {\bf 1}_4 = \left(\begin{matrix}A & 0 & 0 & 0 \\0 & A & 0 & 0 \\0 & 0 & A & 0 \\0 & 0 & 0 & A\end{matrix}\right) ,\qquad A\in\bb{C}_B,
\]
where ${\bf 1}_4\in {\rm Mat}_4(\C)$ is the identity of the algebra of $4\times 4$
complex matrices. The Laguerre basis of $L^2\big(\R^2\big)$ can be lifted to $\s{H}_4$ as $\psi_{n,m}^{(r)}:=\psi_{n,m}\otimes e_r$ where $\psi_{n,m}$ are the Laguerre functions and $e_1,\ldots,e_4$ is the canonical basis of $\C^4$.
It follows that $\big\{\psi_{n,m}^{(r)} \,|\, (n,m)\in\N^2_0, r=1,\ldots,4\big\}$ is a~complete
orthonormal system in $\s{H}_4$ that will be still called the (magnetic) Laguerre basis.
To complete the spectral triple we need a \emph{Dirac operator} $D_B$. The definition of $D_B$ is subordinated to the election of an irreducible representation of the Clifford algebra $C\ell_4(\C)$ on $\C^4$. More precisely, we need four Hermitian $4\times 4$ matrices $\{\gamma_1,\ldots,\gamma_4\}$ which satisfies the fundamental anti-commutation relations
 \[
 \{\gamma_i,\gamma_j\} := \gamma_i \gamma_j + \gamma_j \gamma_1 = 2\delta_{i,j} {\bf 1}_4 ,\qquad i,j=1,\ldots,4 .
 \]
A possible convenient choice is the following:
\begin{alignat*}{3}
& \gamma_{1}:= \left(
 \begin{matrix}
 0 & 0 & 0 & 1\\
 0 & 0 & 1 & 0\\
 0 & 1 & 0 & 0\\
 1 & 0 & 0 & 0
 \end{matrix}
\right) ,\qquad
 &&\gamma_{2} : = \left(\begin{matrix}
 0 & 0 & 0 & -\ii\\
 0 & 0 & \ii & 0\\
 0 & -\ii & 0 & 0\\
 \ii & 0 & 0 & 0
 \end{matrix}
 \right) ,&
\\
& \gamma_{3}: = \left(
 \begin{matrix}
 0 & 0 & 1 & 0\\
 0 & 0 & 0 & -1\\
 1 & 0 & 0 & 0\\
 0 & -1 & 0 & 0
 \end{matrix}
\right) ,\qquad
 &&\gamma_{4} : = \left(\begin{matrix}
 0 & 0 & \ii & 0\\
 0 & 0 & 0 & \ii\\
 -\ii & 0 & 0 & 0\\
 0 & -\ii & 0 & 0
 \end{matrix}
\right) .&
\end{alignat*}
The \emph{$($magnetic$)$} Dirac operator is defined by
\begin{equation*}
D_B := \frac{1}{\sqrt{2}}\big(K_1 \otimes \gamma_1 + K_2 \otimes \gamma_2 + G_1 \otimes \gamma_3 + G_2 \otimes \gamma_4\big) ,
\end{equation*}
where $K_1$, $K_2$, $G_1$, $G_2$ are the magnetic momenta and the dual magnetic momenta defined by~\eqref{eq:momenta} and~\eqref{eq:dual-momenta}, respectively. In terms of the ladder operators $\rr{a}^\pm$ and $\rr{b}^\pm$ the operator $D_B$ reads
\[
D_B = \left(\begin{matrix}
 0 & 0 & -\rr{b}^+ & \rr{a}^-\\
 0 & 0 & \rr{a}^+ & \rr{b}^-\\
 -\rr{b}^- & \rr{a}^- & 0 & 0\\
 \rr{a}^+ & \rr{b}^+ & 0 & 0
 \end{matrix}
\right) .
\]
The operator $D_B$ is a sum of four operators which are essentially self-adjoint on the dense invariant core $S\big(\R^2\big)\otimes\C^4\subset\s{H}_4$. As a consequence~$D_B$ is well defined and symmetric on
$S\big(\R^2\big)\otimes\C^4$.
Consider the involution
\[
\chi := {\bf 1} \otimes \gamma_1\gamma_2\gamma_3\gamma_4 = \left(\begin{matrix}
 +{\bf 1} & 0 & 0&0\\
 0 & +{\bf 1} & 0 & 0\\
 0 &0 & -{\bf 1}& 0\\
 0&0 & 0 & -{\bf 1}
 \end{matrix}
\right) .
\]
The algebraic relations among the $\gamma$-matrices imply
\begin{equation}\label{eq:parit_D}
\chi D_B \chi = -D_B ,\qquad \chi^* = \chi^{-1} = \chi .
\end{equation}
The involution $\chi$ will be called \emph{chiral} or \emph{grading} operator.
A direct computation shows that the square $D_B^2$ acts on $S\big(\R^2\big)\otimes\C^4$ according to
\[
D_B^2 := \left(\begin{matrix}
 Q_B & 0 & 0 & 0\\
 0 & Q_B & 0 & 0\\
 0 & 0 & Q_B & 0\\
 0 & 0 & 0 &Q_B
 \end{matrix}
\right) + \left(\begin{matrix}
 0 & 0 & 0 & 0\\
 0 & 0 & 0 & 0\\
 0 & 0 & +{\bf 1} & 0\\
 0 & 0 & 0 &-{\bf 1}
 \end{matrix}
\right),
\]
where $Q_B$ is the {harmonic oscillator}
defined by~\eqref{eq:harm_osc}. Since $Q_B$ is essentially self-adjoint on $S\big(\R^2\big)$ and
$D_B^2$ is a bounded perturbation of $Q_B\otimes{\bf 1}_4$ it follows from the
Kato--Rellich theorem that $D_B^2$ is essentially self-adjoint on $S\big(\R^2\big)\otimes\C^4$ and define a self-adjoint operator with domain $\s{D}(Q_B)\otimes\C^4$. Moreover $D_B^2$ has spectrum
\[
\sigma\big(D^2_B\big) = \{\theta_j:=j\, |\, j\in\N_0 \} ,
\]
and the eigenvalues $\theta_j$ have finite {multiplicity} $\operatorname{Mult}[\theta_0]=1$,
and $\operatorname{Mult}[\theta_j]=4j$ for every $j\geqslant 1$. As a consequence
\[
\big(D^2_B - z {\bf 1}\big)^{-s} \in \bb{K}(\s{H}_4)
\]
is a compact operator for every $s>0$ and $z\in\C\setminus \N_0$.
\begin{Proposition}\label{prop:comp_res_D}
The operator $D_B$ is essentially self-adjoint on the dense domain $S\big(\R^2\big)\otimes\C^4\subset\s{H}_4$, and its spectrum is given by
\[
\sigma(D_B) = \big\{\vartheta_{\pm j}:=\pm\sqrt{j}\, |\, j\in\N_0 \big\}
\]
with $\operatorname{Mult}[\vartheta_0]=1$,
and $\operatorname{Mult}[\vartheta_{\pm j}]=2j$ for every $j\geqslant 1$.
Moreover $(D_B - z {\bf 1})^{-1}\in \bb{K}(\s{H}_4)$ is compact for every $z\in\C\setminus \sigma(D_B)$.
\end{Proposition}
\begin{proof}
We already know that $D_B$ is symmetric
and that $D_B^2$ is essentially self adjoint on $S\big(\R^2\big)\otimes\C^4$. Moreover, the existence of the resolvent $\big(D_B^2+{\bf 1}\big)^{-1}$ implies that $D_B^2+{\bf 1}$ maps the domain $S\big(\R^2\big)\otimes\C^4$ onto a dense subspace of $\s{H}_4$. The factorization
\[
D_B^2+{\bf 1} = (D_B+\ii{\bf 1})(D_B-\ii{\bf 1}) = (D_B-\ii{\bf 1})(D_B+\ii{\bf 1})
\]
implies that the two operators $D_B\pm\ii{\bf 1}$ map $S\big(\R^2\big)\otimes\C^4$ into a dense
subspace of $\s{H}_4$. This proves that $D_B$ is essentially self-adjoint on
$S\big(\R^2\big)\otimes\C^4$. The spectral mapping theorem and $\sigma\big(D_B^2\big)=\N_0$ imply that for
all $j\in\N_0$, at least one value between $+\sqrt{j}$ and $-\sqrt{j}$ must be in
$\sigma(D_B)$. However, the symmetry~\eqref{eq:parit_D} implies that
$+\sqrt{j}\in \sigma(D_B)$, if and only if, $-\sqrt{j}\in\sigma(D_B)$ and this completes the
description of the spectrum of $D_B$. Let $P^{(2)}_j$ be the spectral projection
of $D_B^2$ related to the eigenvalue $\theta_j$ and $P_{\pm j}$ be the spectral
projections of $D_B$ related to $\vartheta_{\pm j}$, respectively. The spectral
theorem implies that, for $j\neq 0$, $P^{(2)}_j=P_{+ j}+P_{- j}$ and the symmetry~$\chi$ provides $\chi P_{\pm j} \chi= P_{\mp j}$. Therefore, $P_{+ j}$~and~$P_{- j}$ have the
same dimension which is half of the dimension of $P^{(2)}_j$. For the zero
eigenvalue one has that $P^{(2)}_0=P_{0}$. This completes the description of the
multiplicity of the eigenvalues of $D_B$. Finally, the spectral analysis of
$D_B$ guarantees the compactness of the resolvents of $D_B$, which can be
obtained as norm-limit of finite rank operators.
\end{proof}

The chiral operator $\chi$ provides a grading of the Hilbert space
\begin{equation}\label{eq:grading}
\s{H}_4 = \s{H}_4^+ \oplus \s{H}_4^- ,
\end{equation}
where $\s{H}_4^\pm\simeq L^2\big(\R^2\big)\otimes\C^2$ are the ranges of the projections $\chi_\pm:=\frac{1}{2}(\chi\pm{\bf 1})$. Equation~\eqref{eq:parit_D} says that $D_B$ is \emph{odd} with respect to this grading. On the other hand the representation $\rho$ is \emph{even} in the sense that
\begin{equation*}
\chi \rho(A) \chi = \rho(A) ,\qquad \forall\, A\in\bb{C}_B .
\end{equation*}
The next step is to control the commutators of $D_B$ with sufficiently regular elements of $\bb{C}_B$.
\begin{Proposition}\label{prop:commut_D}
For every $A\in \bb{S}_B$ the commutator
\[
[D_B,\rho(A)] := D_B \rho(A) - \rho(A) D_B
\]
is a well defined bounded operator on $\s{H}_4$.
\end{Proposition}
\begin{proof} Every $A\in \bb{S}_B$ has a kernel $f\in \s{S}_B$
such that $A=\pi(f)$. Then, for every $\psi\in S\big(\R^2\big)$ one gets $A\psi=f^-\ast\psi\in S\big(\R^2\big)$
where $f^-(x):=f(-x)$. Therefore, $\rho(A)\big[S\big(\R^2\big)\otimes\C^4\big]\subseteq S\big(\R^2\big)\otimes\C^4$ and, as a consequence, the difference $D_B\rho(A)-\rho(A)D_B$ is well defined on the dense core $S\big(\R^2\big)\otimes\C^4$ of
$D_B$, and it is closable. The closure will be denote with $[D_B,\rho(A)]$. In view of the definition of~$D_B$, Lemma~\ref{lemma:comm_rel_bpm} and Proposition~\ref{props:commut_deriv} one gets
\begin{align*}
[D_B,\rho(A)] &= [K_1,A] \otimes \frac{\gamma_1}{\sqrt{2}} + [K_2,A] \otimes \frac{\gamma_2}{\sqrt{2}}
 = \nabla_1(A) \otimes \frac{\ii\gamma_2}{\sqrt{2}\ell_B} - \nabla_2(A) \otimes \frac{\ii\gamma_1}{\sqrt{2}\ell_B}
\\
&= \pi(\partial_1 (f)) \otimes \frac{\ii\gamma_2}{\sqrt{2}\ell_B} - \pi(\partial_2 (f)) \otimes \frac{\ii\gamma_1}{\sqrt{2}\ell_B} .
\end{align*}
Since $\pi(\partial_j (f))\in\bb{S}_B$ for $j=1,2$, it follows that $[D_B,\rho(A)]$ is a bounded operator.
\end{proof}

Let us introduce the unbounded derivation ${\delta}_B$ initially defined as
\[
{\delta}_B(A) := -\ii[D_B,\rho(A)] ,\qquad A\in \bb{S}_B .
\]
The action of ${\delta}_B$ can be extended continuously to the closure of $\bb{S}_B$ with respect to the graph-norm
\[
 \lVert A \rVert_{D_B} := \lVert A \rVert +\lVert{\delta}_B(A)\rVert,
\]
which is dominated by the norm~\eqref{eq:diff_norm}. This immediately implies the following result:
\begin{Corollary}\label{cor:reg}
The derivation ${\delta}_B$ is well defined on $\bb{C}^1_B$ and the following formula holds true:
\[
{\delta}_B(A) = \nabla_1(A) \otimes \frac{\gamma_2}{\sqrt{2}\ell_B} - \nabla_2(A) \otimes \frac{\gamma_1}{\sqrt{2}\ell_B} ,\qquad A\in \bb{C}^1_B .
\]
Moreover ${\delta}_B(A)=-\ii[D_B,\rho(A)]$ for all $A\in \bb{C}^{1,0}_B$.
\end{Corollary}

Let us denoted with $\bb{A}_B\subset\bb{B}(\s{H}_4)$ the $\ast$-algebra generated by the elements $\rho( \bb{S}_B)$ and the commutators $[D_B,\rho( \bb{S}_B)]$. From Proposition~\ref{prop:commut_D} it follows that $\bb{A}_B\subset\bb{S}_B\otimes {\rm Mat}_4(\C)$.
For every $\varepsilon\geqslant 0$ let
\[
\rr{d}_{B,\varepsilon}(T) := \Big[\sqrt{D_B^2+\varepsilon{\bf 1}},T\Big] ,\qquad T\in \s{D}(\rr{d}_{B,\varepsilon})\subset \bb{B}(\s{H}_4) .
\]
This is an unbounded derivation.
Let $\s{D}^k(\rr{d}_{B,\varepsilon})$ be the domain of the power ${\rr{d}_{B,\varepsilon}}^k:=\rr{d}_{B,\varepsilon}\circ\cdots\circ \rr{d}_{B,\varepsilon}$ and $\s{D}^\infty(\rr{d}_{B,\varepsilon}):=\cap_{k\in\N_0}\s{D}^k(\rr{d}_{B,\varepsilon})$ the related smooth domain. The boundedness of $\rr{d}_{B,\varepsilon}(\rho(A))$ for elements $A\in\bb{S}_B$
is not guaranteed by Proposition~\ref{prop:commut_D}. For instance, when $\varepsilon=0$ the polar decomposition $|D_B|=V^*D_B$ provides
\[\rr{d}_{B,\varepsilon=0}(\rho(A)) = V^{*}{\delta}_B(A)+[V^{*},\rho(A)]D_B
\]
and the second summand is not bounded a priori. For this reason the following result is not at all obvious.
\begin{Proposition}\label{prop:commut_D-II}
With the notation introduced above it holds true that $\bb{A}_B\subset \s{D}^\infty(\rr{d}_{B,\varepsilon})$ for every $\varepsilon\geqslant 0$.
\end{Proposition}
\begin{proof}We will prove something a little more general, namely that $\bb{S}_B\otimes {\rm Mat}_4(\C)\subset \s{D}^\infty(\rr{d}_{B,\varepsilon})$. Let $\tau_{i,k}$ be the $4\times4$
matrix which has a single $1$ in the entry at the position $(i,k)$ and 0 in all other entries.
Then, every element in $\bb{S}_B\otimes {\rm Mat}_4(\C)$ is a finite linear combination of elements of the form $A\otimes \tau_{i,k}$ with $A\in \bb{S}_B$. A direct computation shows that
\[
\rr{d}_{B,\varepsilon}(A\otimes \tau_{i,k}) = \Big[\sqrt{D_B^2+\varepsilon{\bf 1}},A\otimes \tau_{i,k}\Big] = \wp_{i,k}(A)\otimes \tau_{i,k},
\]
where
\[
\wp_{i,k}(A) := \sqrt{Q_B+\alpha_i{\bf 1}} A - A \sqrt{Q_B+\alpha_k{\bf 1}}
\]
and $\alpha_i,\alpha_k\in\{\varepsilon,\varepsilon\pm 1\}$. Since every $A\in \bb{S}_B$ is a linear combination with fast decaying coefficients of the operators $\Upsilon_{r\mapsto s}\in\bb{S}_B$ and the map $A\mapsto \wp_{i,k}(A)$
is linear it is enough to study the generic element $\wp_{i,k}(\Upsilon_{r\mapsto s})$.
Observing that
\[
\sqrt{Q_B+\alpha{\bf 1}} = \sum_{(n,m)\in\N^2_0}\sqrt{n+m+1+\alpha} \Pi_n P_m ,
\]
where $P_m$ are the {dual} Landau projections~\eqref{eq:Lan_proj-dual}, and using the relations
\[\Pi_n P_m\Upsilon_{r\mapsto s}=\delta_{n,s}P_m\Upsilon_{r\mapsto s} \qquad \text{and} \qquad \Upsilon_{r\mapsto s}\Pi_n P_m=\delta_{n,r}P_m\Upsilon_{r\mapsto s},\]
one obtains after a direct computation
\begin{equation}\label{eq:P_ups}
\wp_{i,k}(\Upsilon_{r\mapsto s}) = C_{i,k}^{s,r} \Upsilon_{r\mapsto s},
\end{equation}
where
\[
 C_{i,k}^{s,r} := \sum_{m\in\N_0}\Big(\sqrt{(s+1+\alpha_i)+m}-\sqrt{(r+1+\alpha_k)+m}\Big)P_m .
\]
Since
$\big\lVert \wp_{i,k}(\Upsilon_{r\mapsto s})\big\rVert \leqslant\lVert C_{i,k}^{s,r}\rVert=\big\lvert\sqrt{(s+1+\alpha_i)}-\sqrt{(r+1+\alpha_k)}\big\rvert$
has a growth of order $\frac{1}{2}$ in~$r$ and~$s$, it follows that
$\|\wp_{i,k}(A)\|<+\infty$ in view of the rapid decay of the coefficients of $A$. This
proves that $A\otimes \tau_{i,k}\in \s{D}(\rr{d}_{B,\varepsilon})$ for every $i,k=1,\ldots,4$, and
consequently $\bb{S}_B\otimes {\rm Mat}_4(\C)\subset \s{D}(\rr{d}_{B,\varepsilon})$. The same argument
can be used to prove that the elements of $\bb{S}_B\otimes {\rm Mat}_4(\C)$ are
smooth. For that, it is enough to prove the validity of the equations
${\rr{d}_{B,\varepsilon}}^{p}(A\otimes \tau_{i,k})={\wp_{i,k}}^p(A)\otimes \tau_{i,k}$ and
${\wp_{i,k}}^p(\Upsilon_{r\mapsto s})=\big(C_{i,k}^{s,r}\big)^p\Upsilon_{j\mapsto k}$ for every $p\in\N$.
\end{proof}

After putting the pieces together, and comparing with~\cite[Definitions 9.16 and~10.10]{gracia-varilly-figueroa-01}, we obtain that $(\bb{S}_B,\s{H}_4, D_B)$ is a~\emph{regular even compact spectral triple} as stated in Theorem~\ref{theo:main_01}.

\subsection{Measurability properties}\label{sec:mes_prop}
Let us introduce the family of operators
\[
{|D_{B,\varepsilon}|^{-s}} := \big(D_B^2+\varepsilon{\bf 1}\big)^{-\frac{s}{2}} ,\qquad \varepsilon>0 ,\qquad s\geqslant 1 .
\]
\begin{Proposition}\label{props:old_resul_D_dix}
Let ${|D_{B,\varepsilon}|^{-s}}$ be defined as above. Then:
\begin{itemize}\itemsep=0pt
\item[$(1)$] ${|D_{B,\varepsilon}|^{-s}}\in\rr{S}^{1}$ for every $s>4$, and $\varepsilon>0$;
\item[$(2)$] $|D_{B,\varepsilon}|^{-4}\in \rr{S}^{1^+}_{\rm m}$ for every $\varepsilon>0$,
and
\begin{equation*}
{\Tr}_{\rm Dix}\big(|D_{B,\varepsilon}|^{-4}\big) = 2 ,
\end{equation*}
independently of $\varepsilon>0$.
\end{itemize}
\end{Proposition}
\begin{proof}
An explicit computation provides
\[
{|D_{B,\varepsilon}|^{-s}} := \left(\begin{matrix}
 Q_{B,\varepsilon}^{-\frac{s}{2}} & 0 & 0 & 0\\
 0 & Q_{B,\varepsilon}^{-\frac{s}{2}} & 0 & 0\\
 0 & 0 & Q_{B,\varepsilon+1}^{-\frac{s}{2}} & 0\\
 0 & 0 & 0 &Q_{B,\varepsilon-1}^{-\frac{s}{2}}
 \end{matrix}
\right),
\]
where $Q_{B,\varepsilon}:= Q_B + \varepsilon{\bf 1}$. Then, the result follows from Proposition~\ref{props:old_resul} along with ${\Tr}_{\s{H}_4}={\Tr}_{L^2(\R^2)}\otimes{\Tr}_{\C^4}$ and ${\Tr}_{\rm Dix}|_{\s{H}_4}={\Tr}_{\rm Dix}|_{L^2(\R^2)}\otimes{\Tr}_{\C^4}$. In particular, the second relation is proved in~\cite[Lemma~B.3]{denittis-gomi-moscolari-19}.
\end{proof}

According to~\cite[Definitions 10.8 and~10.12]{gracia-varilly-figueroa-01} we can infer from Proposition~\ref{props:old_resul_D_dix} that the magnetic spectral triple $(\bb{S}_B,\s{H}_4,D_B)$ is \emph{$p$-summable}
for every $p>4$ and \emph{$4^+$-summable}. In particular this would imply that
its \emph{classical dimension} is~4.
However, it is worth recalling that~$\bb{C}_B$ (and hence $\bb{S}_B$) is not unital and therefore, a more appropriate definition of dimension is given by~\cite[Definition~2.1]{gayral-gracia-Bondia-iochum-all-04}.
Let us also observe that the operator $|D_{B,\varepsilon}|$ and the Dixmier trace can be combined to compute the trace $\fint_B$ of the von Neumann algebra $\bb{M}_B$. Indeed, as a~direct consequence of {Proposition}~\ref{prop:trace_dix_main res}
 one gets for all $T\in \bb{L}^1_B$
 \begin{align}
\fint_B(T) &= \frac{1}{4}{\Tr}_{\rm Dix}\big(|D_{B,\varepsilon}|^{-2}\rho(T)\big)
 = \frac{1}{4}{\Tr}_{\rm Dix}\big(\rho(T)|D_{B,\varepsilon}|^{-2}\big)\nonumber\\
&= \frac{1}{4}{\Tr}_{\rm Dix}\big(|D_{B,\varepsilon}|^{-1}\rho(T)|D_{B,\varepsilon'}|^{-1}\big) ,\label{eq:connes_integr}
 \end{align}
independently of $\varepsilon,\varepsilon'>0$.
Equation~\eqref{eq:connes_integr} can be compared with~\cite[equation~(7.83)]{gracia-varilly-figueroa-01}
which describes the noncommutative \emph{Connes' integral}.
In particular, in view of~\cite[Corollary~7.21]{gracia-varilly-figueroa-01}, one can interpret formula~\eqref{eq:connes_integr} as the \virg{volume-integral} of the \virg{noncommutative smooth manifold}~$\bb{S}_B$. It turns out that
the spectral triple $(\bb{S}_B,\s{H}_4,D_B)$
behaves like a noncommutative Riemannian manifold of dimension~2. Borrowing the parlance of~\cite[Theorem~3.2]{gayral-gracia-Bondia-iochum-all-04}
 we can state the following result:
\begin{Theorem}[spectral dimension]\label{theo:spec_dim}
The spectral dimension of the magnetic spectral triple $(\bb{S}_B,\s{H}_4,D_B)$ is~$2$.
\end{Theorem}

The content of Theorem~\ref{theo:spec_dim} describes an expected property of non-unital spectral triples investigated in full detail, and in a quite similar framework, in the excellent works \cite{gayral-wulkenhaar-13,grosse-wulkenhaar-12}.

\begin{Remark}[relation with the trace per unit of volume III]\label{rk:trac_unit_volIII}
From Remark~\ref{rk:trac_unit_volII} we get the formula
\[
\begin{aligned}
 \s{T}_B(T) &= \frac{1}{8\Lambda_B} {\Tr}_{\rm Dix}\big(|D_{B,\varepsilon}|^{-2}\rho(T)\big) ,
 \end{aligned}
 \]
for the trace per unit of volume of every operator $T\in\bb{G}^{1\cdot2}_B$, independently of $\varepsilon>0$.
\end{Remark}

\subsection{Absence of real structures}\label{sec:no real}
One interesting property of the magnetic spectral triple is that, while
it satisfies most of the axioms required for a noncommutative
spin geometry, it does not admit a real structure as it is
evidenced in the next proposition.

\begin{Proposition}[real structure] The magnetic spectral triple does not admit a real structure.
\end{Proposition}

\begin{proof} Suppose that $(\mathscr{S}_{B}, \mathcal{H}_{B}, D_{B})$ admits a real
 structure $\mathscr{J}$. According to~\cite[Definition~9.18]{gracia-varilly-figueroa-01}, such structure $\mathscr{J}$
 would be an
 anti-linear isometry that satisfies $\mathscr{J}D_{B}
 =D_{B}\mathscr{J}$, and
 $\mathscr{J}^{2} = -{\bf 1}$. Let $\psi\in{\rm Ker}(D_{B})$
 a non-zero vector. Then $\mathscr{J}\psi = \alpha \psi$ for some $\alpha \not= 0$ since the kernel of
 $D_{B}$ is one-dimensional and $D_{B}$ commutes with $\mathscr{J}$. However, this would imply
 $- \psi = \mathscr{J}^{2} \psi = \lvert \alpha \rvert^{2} \psi$
 which is evidently a contradiction.
\end{proof}

From a physical point of view the absence of a real structure can be interpreted as a manifestation of the presence of a non-trivial magnetic field which breaks the time reversal symmetry.

\subsection[Magnetic Fredholm module and KK-homology]{Magnetic Fredholm module and $\boldsymbol{K\!K}$-homology}\label{sec:magn_fre-mod}
A $\ast$-algebra can be endowed with a quantized differential calculus by means of a \emph{Fredholm module}. In turn, a Fredholm module can be derived from a spectral triple. We will refer to~\cite[Chapter~4]{connes-94}
and~\cite[Chapters 8 and~10]{gracia-varilly-figueroa-01}
for the general theory concerning Fredholm modules, the
cyclic cohomology and their relation with spectral triples.

Let us start by introducing the following operator
\begin{equation}\label{eq:fred_phas}
F_{B,\varepsilon} := \frac{D_B}{|D_{B,\varepsilon}|} ,\qquad \varepsilon>0 .
\end{equation}
In view of the definition of $|D_{B,\varepsilon}|$, equation~\eqref{eq:fred_phas} defines a bounded operator which is self-adjoint, i.e.,
\[
F_{B,\varepsilon}^* = F_{B,\varepsilon} .
\]
Moreover, $F_{B,\varepsilon}$ is a quasi-involution in the sense that
\[
F_{B,\varepsilon}^2 = {\bf 1} - \varepsilon |D_{B,\varepsilon}|^{-2}
\]
and the difference $F_{B,\varepsilon}^2-{\bf 1}$ is a compact operator. Moreover, a direct check shows that
\begin{equation*}
\chi F_{B,\varepsilon} \chi = -F_{B,\varepsilon} .
\end{equation*}
The next result characterizes the commutator of $F_{B,\varepsilon}$ with the elements of the magnetic $C^*$-algebra $\bb{S}_B$.
\begin{Lemma}\label{lemma:sub_fred_01}
For every
$A\in \bb{C}_B$ the commutator $[F_{B,\varepsilon},\rho(A)]$ is compact.
\end{Lemma}
\begin{proof}Let us start with $A\in \bb{S}_B$. In view of Proposition~\ref{prop:commut_D} the commutator $[D_B,\rho(A)]$ is bounded. In a similar way, one can prove that
\[
\big[D_B^2,\rho(A)\big] = [Q_B,A]\otimes{\bf 1}_4 = [\rr{a}^+\rr{a}^-,A]\otimes{\bf 1}_4 \in \rho(\bb{S}_B) ,
\]
by using the properties proved in Proposition~\ref{prop:struct}, the commutators~\eqref{eq:commut_ups} and the fact that $A$ is a linear combination of $\Upsilon_{j\mapsto k}$ with rapidly decaying coefficients.
Following the strategy of~\cite[Lemma~10.18]{gracia-varilly-figueroa-01} we can rewrite
\begin{equation}\label{eq:fact_comm_F}
[F_{B,\varepsilon},\rho(A)] = [D_B,\rho(A)] |D_{B,\varepsilon}|^{-1} + D_B \big[|D_{B,\varepsilon}|^{-1},\rho(A)\big] .
\end{equation}
by introducing the spectral formula
\begin{equation*}
|D_{B,\varepsilon}|^{-1} = \frac{1}{\pi}\int_\varepsilon^{+\infty}\frac{\dd \lambda}{\sqrt{\lambda-\varepsilon}} |D_{B,\lambda}|^{-2} .
\end{equation*}
After some manipulation one gets
\begin{equation*}
[F_{B,\varepsilon},\rho(A)] = \frac{\ii}{\pi}\int_\varepsilon^{+\infty}\frac{\dd \lambda}{\sqrt{\lambda-\varepsilon}} \big(\delta_B(A) + \ii \eta_{B,\lambda}(A)\big) |D_{B,\lambda}|^{-2} ,
\end{equation*}
with $\delta_B(A):=-\ii[D_B,\rho(A)]$, and
\[
 \eta_{B,\lambda}(A) := D_B|D_{B,\lambda}|^{-2}\big[D_B^2,\rho(A)\big] .
\]
The compactness of $|D_{B,\lambda}|^{-2}$ for every $\lambda\geqslant \varepsilon$ and the boundedness
of $\delta_B(A)+\ii \eta_{B,\lambda}(A)$ implies the compactness of $[F_{B,\varepsilon},\rho(A)]$.
Now, let $\{A_n\}\subset \bb{S}_B$ be a sequence that converges in norm to $A\in \bb{C}_B$. The inequality
\[
\|[F_{B,\varepsilon},\rho(A_n)]-[F_{B,\varepsilon},\rho(A)]\| \leqslant 2 \|F_{B,\varepsilon}\lVert \, \rVert A_n-A\| ,
\]
along with the closure of the space of compact operators in norm topology, implies that \linebreak $[F_{B,\varepsilon},\rho(A)]$ is also compact.
\end{proof}

The next result concerns with the summability property of the commutators $[F_{B,\varepsilon},\rho(A)]$.

We need to introduce the \emph{second} Dixmier ideal
$\rr{S}^{2^+}$ (also known as the Ma\u{c}aev ideal of order~$2^+$) which consists of the
compact operators $T$ such that their \emph{$($Calder\'on$)$ norm}
\begin{equation}\label{eq:clad_norm_2}
\lVert T\rVert_{2^+} := \sup_{N>1}\ \frac{\sigma^1_N(T)}{\sqrt{N}}
\end{equation}
is finite. The ideal $\rr{S}^{2^+}$ is closed with respect to the norm \ref{eq:clad_norm_2}. Moreover, every element $T \in \rr{S}^{2^+}$ satisfies
 $|T|^2=T^*T \in \rr{S}^{1^+}$~\cite[Lemma~7.37]{gracia-varilly-figueroa-01}. As a consequence, if
$S_1,S_2\in\rr{S}^{2^+}$ then their product $S_1S_2\in \rr{S}^{1^+}$ is in the Dixmier
ideal. The latter fact can be initially justified for positive operators by following the same argument in the proof of~\cite[Proposition~7.16]{gracia-varilly-figueroa-01}.
Then, the polar decomposition and the fact that~$\rr{S}^{2^{+}}$ is an ideal allow to extend the result to arbitrary pairs of elements of~$\rr{S}^{2^{+}}$.

\begin{Lemma}\label{lemma:sub_fred_02}
For every $A\in \bb{S}_B$, the commutator $[F_{B,\varepsilon},\rho(A)]$ lies in the second Dixmier ideal~$\rr{S}^{2^+}$.
\end{Lemma}
\begin{proof}
First of all, let us prove that $T|D_{B,\varepsilon}|^{-1}\in \rr{S}^{2^+}$ for every
$T\in \bb{S}_B\otimes {\rm Mat}_4(\C)$.

A direct inspection to the matrix-valued operator
\begin{equation*}
\big\lvert T|D_{B,\varepsilon}|^{-1}\big\rvert^2 = |D_{B,\varepsilon}|^{-1} |T|^2 |D_{B,\varepsilon}|^{-1}
\end{equation*}
shows that its entries are of the form $Q_{B,\varepsilon}^{-\frac{1}{2}}A Q_{B,\varepsilon'}^{-\frac{1}{2}}$ with $A\in \bb{S}_B$ and $\varepsilon,\varepsilon'>0$. Consider first the case $A = \Upsilon_{j \mapsto k}$. A comparison with \eqref{eq:aux_1ref} shows that $\mu_{m}\big(Q_{B,\varepsilon}^{-\frac{1}{2}}A Q_{B,\varepsilon'}^{-\frac{1}{2}}\big)(1 + m)<1$ for all $m\in\N_0$. Therefore, one gets
\[
\big\|Q_{B,\varepsilon}^{-\frac{1}{2}}\Upsilon_{j\mapsto k}Q_{B,\varepsilon'}^{-\frac{1}{2}}\big\|_{2^+} < \sup_{N>1}\frac{1}{\sqrt{N}}\sum_{m=0}^{N-1}\frac{1}{m+1} < 2 ,
\]
independently of $j$ and $k$. Therefore, by using that $\rr{S}^{2^+}$ is closed with respect to the norm~\eqref{eq:clad_norm_2}, and the expansion in terms of the $\Upsilon_{j\mapsto k}$, one gets that $\big\|Q_{B,\varepsilon}^{-\frac{1}{2}}AQ_{B,\varepsilon'}^{-\frac{1}{2}}\big\|_{2^+}<+\infty$ for every \mbox{$A \in \bb{S}_{B}$}, and in turn $T|D_{B,\varepsilon}|^{-1}\in \rr{S}^{2^+}$. The factorization~\eqref{eq:fact_comm_F} implies that the claim is proved if we can show that the two summands on the left-hand side of~\eqref{eq:fact_comm_F} are in the ideal $\rr{S}^{2^+}$.
Proposition~\ref{prop:commut_D} suggests that $[D_B,\rho(A)]\in \bb{S}_B\otimes {\rm Mat}_4(\C)$ and the discussion above implies that $[D_B,\rho(A)] |D_{B,\varepsilon}|^{-1}\in \rr{S}^{2^+}$. For the second summand we can use the (trivial) identity $[{\bf 1},\rho(A)]=0$ to rewrite
\begin{align*}
D_B \big[|D_{B,\varepsilon}|^{-1},\rho(A)\big] = F_{B,\varepsilon}|D_{B,\varepsilon}|\big[|D_{B,\varepsilon}|^{-1},\rho(A)\big]
 = -F_{B,\varepsilon} \rr{d}_{B,\varepsilon}(\rho(A)) |D_{B,\varepsilon}|^{-1} ,
\end{align*}
where the derivation $\rr{d}_{B,\varepsilon}$
is the commutator with $|D_{B,\varepsilon}|$.
Since $F_{B,\varepsilon}$ is bounded, and $\rr{S}^{2^+}$ is an ideal, it is enough to prove that
$\rr{d}_{B,\varepsilon}(\rho(A))|D_{B,\varepsilon}|^{-1}$ is in the second Dixmier ideal.
From the proof of Proposition~\ref{prop:commut_D-II} one gets
\[
\rr{d}_{B,\varepsilon}(\rho(A)) |D_{B,\varepsilon}|^{-1}
=
\left(\begin{matrix}
\wp_{1,1}(A)Q_{B,\varepsilon}^{-\frac{1}{2}} & 0 & 0 & 0 \\0 & \wp_{2,2}(A)Q_{B,\varepsilon}^{-\frac{1}{2}} & 0 & 0 \\0 & 0 & \wp_{3,3}(A)Q_{B,\varepsilon+1}^{-\frac{1}{2}} & 0 \\0 & 0 & 0 & \wp_{4,4}(A)Q_{B,\varepsilon-1}^{-\frac{1}{2}}\end{matrix}\right) .
\]
Thus, one needs only to prove that each term in the diagonal is in
$\rr{S}^{2^+}$. From~\eqref{eq:P_ups} one deduces that
$\wp_{i,i}(\Upsilon_{r\mapsto s})Q_{B,\varepsilon}^{-\frac{1}{2}}=C_{i,i}^{s,r}\Upsilon_{r\mapsto s}Q_{B,\varepsilon}^{-\frac{1}{2}}\in \rr{S}^{2^+}$ since $C_{i,i}^{s,r}$ is a bounded operator (commuting with~$\Upsilon_{r\mapsto s}$ and~$Q_{B,\varepsilon}$). Clearly this result extends to finite linear combinations of operators~$\Upsilon_{r\mapsto s}$. The proof of the general case can be obtained by observing that the map
\begin{equation}\label{eq:ineq_mesur_mod_fred_00}
\bb{S}_B \ni A \longmapsto \wp_{i,i}(A)Q_{B,\varepsilon}^{-\frac{1}{2}} \in \rr{S}^{2^+} ,
\end{equation}
initially defined on finite linear combinations, is continuous when $\bb{S}_B$ is endowed with its
Fr\'echet topology, and $\rr{S}^{2^+}$ with the topology induced by the norm $\lVert \ \rVert_{2^+}$.
The latter fact can be proved by using the same strategy as in the proof of Proposition~\ref{prop:trace_dix_main res}. First of all one needs to compute
\[
\big\lVert\Upsilon_{r\mapsto s}Q_{B,\varepsilon}^{-\frac{1}{2}}\big\rVert_{2^+} = \sup_{N>1}\frac{1}{\sqrt{N}}\sum_{m=0}^{N-1}\frac{1}{\sqrt{m+r+1+\varepsilon}}
 \leqslant \sup_{N>1}\frac{1}{\sqrt{N}}\sum_{m=1}^{N}\frac{1}{\sqrt{m}} = 2 .
\]
Then, let $A=\sum_{(s,r)\in\N_0^2}c_{s,r}\Upsilon_{r\mapsto s}$ be a generic element of $\bb{S}_B$. A straightforward computation shows
\begin{equation}\label{eq:ineq_mesur_mod_fred_01}
\big\lVert \wp_{i,i}(A)Q_{B,\varepsilon}^{-\frac{1}{2}}\big\rVert_{2^+} \leqslant 2\sum_{(s,r)\in\N_0^2}|c_{s,r}| \big\lVert C_{i,i}^{s,r}\big\rVert \leqslant 2 \alpha_i r_4(\{c_{n,m}\}) ,
\end{equation}
where $r_4(\{c_{n,m}\})$ is the Schwartz semi-norm defined by~\eqref{eq:fre_top} and the constant $\alpha_i$ is defined by
\[
 \alpha_i^2 := \sum_{(s,r)\in\N_0^2}\frac{\|C_{i,i}^{s,r}\|^2}{{ (2s+1) }^4{ (2r+1) }^4} ,
\]
and the value of $\|C_{i,i}^{s,r}\|$ computed in the proof of Proposition~\ref{prop:commut_D-II}.
Inequality~\eqref{eq:ineq_mesur_mod_fred_01} proves the continuity of the map~\eqref{eq:ineq_mesur_mod_fred_00}.
\end{proof}

\begin{Remark}\label{Rk:mes_dob_prod}
A closer inspection of the proof of Lemma~\ref{lemma:sub_fred_02} allows to deduce a slightly stronger result, namely that
\[
[F_{B,\varepsilon},\rho(A_1)] [F_{B,\varepsilon},\rho(A_2)] \in \rr{S}^{1^+}_{\rm m} ,
\]
for all $A_1,A_2\in \bb{S}_B$.
\end{Remark}

Comparing all the results discussed above with~\cite[Definition 8.4]{gracia-varilly-figueroa-01} and~\cite[Chapter~4, Section~1.$\gamma$, Definition~8]{connes-94} we can claim that:
\begin{Theorem}\label{theo:Fred_module}
The triplet $(\rho,\s{H}_4,F_{B,\varepsilon})$ is a $($pre-$)$Fredholm module for the magnetic algebra~$\bb{C}_B$ with grading
operator $\chi$. It is densely $2^+$-summable $($on the pre-$C^*$-algebra $\bb{S}_B)$.
\end{Theorem}

\begin{Remark}[pre-Fredholm modules vs.\ Fredholm modules]\label{Rk:F_0}
 According to~\cite[Definition~8.4]{gracia-varilly-figueroa-01}, a
 \emph{genuine} Fredholm module is defined by a~self-adjoint involution
 $F=F^*=F^{-1}$. In our case the involution property is violated since the
 difference $F^2_{B,\varepsilon}-{\bf 1}$ is $\varepsilon$ times a compact operator. On the other
 hand the presence of $\varepsilon$ is needed to make the operator $D_B$ invertible. The
 violation of the involutive property is not a big issue. Indeed, as discussed
 in~\cite[p.~327]{gracia-varilly-figueroa-01} there is a~canonical procedure to
 associate a {genuine} Fredholm module to a given pre-Fredholm module. Another
 possibility is to define the \emph{$\varepsilon=0$ regularization} of the resolvent of
 $D_B$ as described in~\cite[Section~4.2.$\gamma$]{connes-94} or~\cite[p.~446]{gracia-varilly-figueroa-01} Let $V_0$ be the one-dimensional kernel of~$D_B$ and~$P_{V_0}$ the orthogonal projection on~$V_0$. Then, one can define
 $|D_{B,0}|^{-1}:=|D_B|^{-1}({\bf 1}-P_{V_0})$ on the orthogonal complement
 $V_0^\bot$ (where the inverse of $|D_B|$ is well defined). One introduces the
 extended Hilbert space $\widehat{\s{H}}_4:=V_0^\bot\oplus(V_0\oplus V_0)$ and with the
 identification $V_0\oplus V_0\simeq\C^2$, one defines
 \[
 F_{B,0} :=D_B|D_{B,0}|^{-1} \oplus \left(\begin{matrix}0 & 1 \\1 & 0\end{matrix}\right) .
 \]
In this way $F_{B,0}$ coincides with the {sign function} $\operatorname{sgn}(D_B)$ on
$V_0^\bot$ and provides a true self-adjoint involution on $\widehat{\s{H}}_4$.
Since $\widehat{\s{H}}_4= {\s{H}}_4\oplus V_0$ one can extend the representation $\rho$
of the elements $A\in\bb{C}_B$ in the following form $\widehat{\rho}(A):={\rho}(A)\oplus 0$
and the grading is restored by $\widehat{\chi}=\chi\oplus (-1)$.
\end{Remark}

The (pre-)Fredholm modules are the building blocks of the \emph{Kasparov
 $K\!K$-theory} (see~\cite[Chapter~4, Appendix~A]{connes-94} or~\cite[Chapter~VII]{blackadar-98} for more details). In a nutshell, we can say that the
Kasparov $K\!K$-theory is a homology theory in which the cycles are
(pre-)Fredholm modules. Elements of the $K\!K$-homology group are equivalence
classes of (pre-)Fredholm modules modulo homotopy. With this in mind, we can
consider the triplets $(\rho,\s{H}_4,F_{B,\varepsilon})$ as cycles of the $K\!K$-homology
$K\!K(\bb{C}_B,\C)$. Since the map $(0,+\infty)\ni\varepsilon\mapsto F_{B,\varepsilon}$ is {norm-continuous} it
follows that all these cycles define a unique class of $K\!K(\bb{C}_B,\C)$ that
will be denoted with $[F_B]$. Moreover, since $\lim\limits_{\varepsilon \to 0^+}F_{B,\varepsilon}=F_{B,0}$
when restricted to $V_0^\bot$, and the remainder in the orthogonal complement is a
\emph{compact perturbation} (in the sense of~\cite[Proposition
17.2.5]{blackadar-98}) it follows that the class $[F_B]$ contains also the $\varepsilon=0$
regularization described in Remark~\ref{Rk:F_0}.

\subsection{Graded structure of the magnetic Fredholm module}\label{sec:grad_struct_magn_fre-mod}
It is instructive to represent $F_{B,\varepsilon}$ with respect to the grading~\eqref{eq:grading} induced by $\chi$. A simple computation shows that
\[
F_{B,\varepsilon} = \left(\begin{array}{@{}c|c@{}}0 & U_{+,-} \\
\hline
U_{-,+} & 0\end{array}\right) ,
\]
where $U_{+,-}\colon \s{H}_-\to\s{H}_+$ is defined by
\[
U_{+,-} := \chi_+ F_{B,\varepsilon} \chi_-\big|_{\s{H}_-} = \big(Q_{B,\varepsilon}^{-\frac{1}{2}}\otimes{\bf 1}_2\big)\left(\begin{matrix}-\rr{b}^+ & \rr{a}^- \\\rr{a}^+ & \rr{b}^-\end{matrix}\right) ,
\]
and $U_{-,+}:=(U_{+,-})^*$. From these definitions, one gets the following relations
\[
U_{+,-}U_{-,+} = {\bf 1}_{\s{H}_-} - \varepsilon Q_{B,\varepsilon}^{-1}\otimes{\bf 1}_2 ,
\]
and
\[
U_{-,+}U_{+,-} = {\bf 1}_{\s{H}_+} - \varepsilon \left(\begin{matrix}Q_{B,\varepsilon+1}^{-1} & 0 \\0 & Q_{B,\varepsilon-1}^{-1}\end{matrix}\right) .
\]

An element $T\in \bb{B}(\s{H}_4)$ will be called of \emph{degree $j\in\{0,1\}$} with respect to the grading induced by $\chi$ if $\chi T\chi=(-1)^j T$.
Every $T$ can be naturally split as $T=T_0+T_1$
where $T_0:=\chi_+T\chi_++\chi_-T\chi_-$ has degree $0$ and $T_1:=\chi_+T\chi_-+\chi_-T\chi_+$ has degree~1.
In the matricial form one has that
\begin{equation}\label{eq:nat_split}
T = T_0+T_1 := \left(\begin{array}{@{}c|c@{}}T_{+,+} & 0 \\
\hline
0 & T_{-,-}\end{array}\right) + \left(\begin{array}{@{}c|c@{}}0 & T_{+,-} \\
\hline
T_{-,+} & 0\end{array}\right) .
\end{equation}
Let us consider the commutator induced by the operator $F_{B,\varepsilon}$ on the algebra $\bb{B}(\s{H}_4)$, i.e., $T\mapsto [F_{B,\varepsilon}, T]$. An explicit computation shows that
\begin{gather*}
{[F_{B,\varepsilon}, T]} = {\left(\begin{array}{@{}c|c@{}}U_{+,-}T_{-,+}- T_{+,-}U_{-,+}& 0 \\
\hline
0 & U_{-,+}T_{+,-}-T_{-,+}U_{+,-}\end{array}\right)}\\
\hphantom{{[F_{B,\varepsilon}, T]} =}{} + {\left(\begin{array}{@{}c|c@{}}0 & U_{+,-}T_{-,-}-T_{+,+}U_{+,-} \\
\hline
U_{-,+}T_{+,+}-T_{-,-}U_{-,+} & 0\end{array}\right)} .
\end{gather*}
One gets that $[F_{B,\varepsilon}, T_0]$ is odd and $[F_{B,\varepsilon}, T_1]$ is even, namely the commutator shifts by $1\!\mod 2$ the degree of an element.

The latter observation suggests to introduce the \emph{graded commutator} on $\bb{B}(\s{H}_4)$, cf.\ \cite[Section~5.A]{gracia-varilly-figueroa-01}. Let $S,T\in\bb{B}(\s{H}_4)$ be two elements with definite degree, say $\deg(S),\deg(T)\in\{0,1\}$. The graded commutator between $S$ and $T$ is initially defined as
\[
[S,T]_\chi := ST -(-1)^{\deg(S)\deg(T)} TS,
\]
and then is extended to the whole algebra $\bb{B}(\s{H}_4)$ by linearity using the natural split~\eqref{eq:nat_split}. Since $F_{B,\varepsilon}$ is of degree 1 one obtains that
\[
[F_{B,\varepsilon}, ST]_\chi =[F_{B,\varepsilon}, S]_\chi T+(-1)^{\deg(S)}S[F_{B,\varepsilon},T]_\chi
\]
for every pair of graded elements $S,T\in\bb{B}(\s{H}_4)$. Similarly, it can be verified that
\[
\big[F_{B,\varepsilon}, [F_{B,\varepsilon}, T]_\chi\big]_\chi = \big[F_{B,\varepsilon}^2, T\big] = -\varepsilon\big[|D_{B,\varepsilon}|^{-2},T\big]
\]
for every element $T\in\bb{B}(\s{H}_4)$ independently of the degree.

\subsection{The first Connes' formula}\label{sec:1stconn_form}
The graded structure of $\bb{B}(\s{H}_4)$ provided by the magnetic Fredholm module $(\rho,\s{H}_4,F_{B,\varepsilon})$
and the involution~$\chi$ is the starting point for the construction of a quantized differential calculus~\cite[Chapter~4]{connes-94}. In this calculus the role of the exterior derivative is played by the commutator with~$F_{B,\varepsilon}$. This justifies the introduction of the following notation:
\begin{equation}\label{eq:def_diff_d}
\dd(T) := [F_{B,\varepsilon}, T]_\chi ,\qquad T\in \bb{B}(\s{H}_4) .
\end{equation}
For reasons that will be discussed in a future work we will refer to the map $\mathrm{d}$ as the \emph{graded quasi-differential}.
The behavior under the involution is given by
\[
 \dd(T^*) = (-1)^{1-{\deg}(T)}\dd(T)^* ,\qquad T\in \bb{B}(\s{H}_4) .
\]

Even though in this work we will not discuss the construction of the quantized calculus for
the magnetic Fredholm module, we will use the notation~\eqref{eq:def_diff_d} to prove the so called
\emph{first Connes' formula} according to the name introduced by Bellissard et al.\ for~\cite[Theorem~9]{bellissard-elst-schulz-baldes-94}. It is worth pointing out that the first Connes' formula as proved in~\cite[Theorem~9]{bellissard-elst-schulz-baldes-94} works only for the discrete magnetic algebra. In contrast, the proof that we will provide here concerns with the continuous case.

Let us focus now on the product $[F_{B,\varepsilon},T_1]^*[F_{B,\varepsilon},T_2]$ for elements of the type $T_1,T_2\in \bb{S}_B\otimes{\rm Mat}_4(\C)$. After exploiting the expansion~\eqref{eq:fact_comm_F} one gets
\[
[F_{B,\varepsilon},T_1]^*[F_{B,\varepsilon},T_2] = \sum_{i=0}^3I_i(T_1,T_2) ,
\]
with
\begin{gather}
I_0(T_1,T_2) := |D_{B,\varepsilon}|^{-1}[D_B,T_1]^*[D_B,T_2]|D_{B,\varepsilon}|^{-1},\nonumber\\
I_1(T_1,T_2) := |D_{B,\varepsilon}|^{-1}[D_B,T_1]^*D_B\big[|D_{B,\varepsilon}|^{-1},T_2\big],\nonumber\\
I_2(T_1,T_2) = \big[|D_{B,\varepsilon}|^{-1},T_1\big]^*D_B[D_B,T_2]|D_{B,\varepsilon}|^{-1}, \nonumber\\
I_3(T_1,T_2) = \big[|D_{B,\varepsilon}|^{-1},T_1\big]^*D_B^2\big[|D_{B,\varepsilon}|^{-1},T_2\big] . \label{eq:fact_comm_F_terms_Is}
\end{gather}

\begin{Lemma}\label{lemma:pre_connes01}
Assume that $T_1,T_2\in\rho(\bb{S}_B)=\bb{S}_B\otimes{\bf 1}_4$ and that
$Y\in\bb{B}(\s{H}_4)$ is a bounded operator. Then, for every election of the Dixmier trace, it holds true that
\[
{\Tr}_{{\rm Dix},\omega}\big(Y [F_{B,\varepsilon},T_1 ]^* [F_{B,\varepsilon},T_2 ]\big) = {\Tr}_{{\rm Dix},\omega}(YI_0)
\]
where $I_0\equiv I_0(T_1,T_2)$ is defined in~\eqref{eq:fact_comm_F_terms_Is}.
\end{Lemma}
\begin{proof}
By inspecting the proof of Lemma~\ref{lemma:sub_fred_02} one infers that the operators $I_i(T_1,T_2)$, $i=0,\ldots,3$ are individually in $\rr{S}^{1^+}$. Then, to prove the claim it is enough to show that
$I_i(T_1,T_2)\in \rr{S}^{1^+}_0$ for $i=1,2,3$.
In view of the fact that the operators $T_1$ and $
T_2$ are linear combinations of the fundamental operators $\rho(\Upsilon_{r\mapsto s})$, the maps $I_i(T_1,T_2)$ are bilinear and the trace is linear, it is sufficient to prove that
$I_i(\rho(\Upsilon_{r\mapsto s}),\rho(\Upsilon_{r'\mapsto s'}))\in \rr{S}^{1^+}_0$
for every $r,s,r',s'\in\N_0$ and $i=1,2,3$. For that, let us observe
\[
\big[|D_{B,\varepsilon}|^{-1},\rho(\Upsilon_{r\mapsto s})\big] = \sum_{i=1}^4\big[Q_{B,\varepsilon_i}^{-\frac{1}{2}},\Upsilon_{r\mapsto s}\big]\otimes\tau_{i,i} ,
\]
where the matrices $\tau_{i,j}$ have been defined in the proof of Proposition~\ref{prop:commut_D-II}, $\varepsilon_i\in\{\varepsilon,\varepsilon\pm 1\}$,
and
\[
\big[Q_{B,\varepsilon_i}^{-\frac{1}{2}},\Upsilon_{r\mapsto s}\big] = \left(\sum_{m\in\N_0}\frac{\alpha_m^{(s,r)}}{(m+1+\varepsilon_i)^{\frac{3}{2}}}P_m\right)\Upsilon_{r\mapsto s} ,
\]
with
\[
\alpha_m^{(s,r)} := \frac{r-s}{\sqrt{\left(1+\frac{s}{m+1+\varepsilon_i}\right)\left(1+\frac{r}{m+1+\varepsilon_i}\right)}\left(\sqrt{1+\frac{r}{m+1+\varepsilon_i}}+\sqrt{1+\frac{s}{m+1+\varepsilon_i}}\right)} .
\]
The explicit spectral resolution shows that the operator
$\big[Q_{B,\varepsilon}^{-\frac{1}{2}},\Upsilon_{r\mapsto s}\big]$
is trace-class (i.e., it lies in the first Schatten ideal $\rr{S}^1$)
and $Q_{B,\varepsilon}\big[Q_{B,\varepsilon}^{-\frac{1}{2}},\Upsilon_{r\mapsto s}\big]$ is a bounded operator. Then, it follows that
\begin{gather}
\big[|D_{B,\varepsilon}|^{-1},\rho(\Upsilon_{r\mapsto s})\big]\in \rr{S}^1 \subset \rr{S}^{1^+}_0 ,\nonumber\\
|D_{B,\varepsilon}|^{2}\big[|D_{B,\varepsilon}|^{-1},\rho(\Upsilon_{r\mapsto s})\big] \in \bb{B}(\s{H}_4) .\label{eq:XXYXX}
\end{gather}
These conditions, along with the fact that $\rr{S}^{1^+}_0$ is a closed ideal and the identity $D_B^2=|D_{B,\varepsilon}|^2 - \varepsilon{\bf 1}$, immediately imply
\[
I_3(\rho(\Upsilon_{r\mapsto s}),\rho(\Upsilon_{r'\mapsto s'})) \in \rr{S}^{1^+}_0 .
\]
Let us focus now on $I_2(\rho(\Upsilon_{r\mapsto s}),\rho(\Upsilon_{r'\mapsto s'}))=B^*Z$
where
\begin{gather*}
Z := |D_{B,\varepsilon}|^{-1}[D_B,\rho(\Upsilon_{r'\mapsto s'})]|D_{B,\varepsilon}|^{-1} ,\\
B := F_{B,\varepsilon}|D_{B,\varepsilon}|^{2}\big[|D_{B,\varepsilon}|^{-1},\rho(\Upsilon_{r\mapsto s})\big]^* .
\end{gather*}
In view of~\eqref{eq:XXYXX} one has that $B\in \bb{B}(\s{H}_4)$.
On the other hand
\begin{gather*}
Z = \sum_{i,j=1}^4\big(Q_{B,\varepsilon_i}^{-\frac{1}{2}}\nabla_1(\Upsilon_{r\mapsto s})Q_{B,\varepsilon_j}^{-\frac{1}{2}}\big)\otimes\frac{\ii\tau_{i,i}\gamma_2\tau_{j,j}}{\sqrt{2}\ell_B}
 -\sum_{i,j=1}^4\big(Q_{B,\varepsilon_i}^{-\frac{1}{2}}\nabla_2(\Upsilon_{r\mapsto s})Q_{B,\varepsilon_j}^{-\frac{1}{2}}\big)\otimes\frac{\ii\tau_{i,i}\gamma_1\tau_{j,j}}{\sqrt{2}\ell_B} .
\end{gather*}
The operators $Q_{B,\varepsilon_i}^{-\frac{1}{2}}\nabla_k(\Upsilon_{r\mapsto s})Q_{B,\varepsilon_j}^{-\frac{1}{2}}$,
with $k=1,2$, are in the Dixmier ideal in view of
Proposition~\ref{prop:trace_dix_main res}. Moreover, an application
of~\cite[Lemma B.3]{denittis-gomi-moscolari-19}, along with the fact that the
matrices $\tau_{i,i}\gamma_k\tau_{j,j}$, with $k=1,2$, have vanishing trace, imply that
${\Tr}_{{\rm Dix}}(Z)\in \rr{S}^{1^+}_0$, and in turn
\[
I_2(\rho(\Upsilon_{r\mapsto s}),\rho(\Upsilon_{r'\mapsto s'})) \in \rr{S}^{1^+}_0 .
\]
The remaining case can be treated by observing that
\[
I_1(\rho(\Upsilon_{r\mapsto s}),\rho(\Upsilon_{r'\mapsto s'})) = I_2(\rho(\Upsilon_{r'\mapsto s'}),\rho(\Upsilon_{r\mapsto s}))^* \in \rr{S}^{1^+}_0
\]
in view of the fact that $\rr{S}^{1^+}_0$ is a self-adjoint ideal.
\end{proof}

We are now in a position to prove our main result.
For that we need to recall the definition of the trace $\fint_B$ introduced in Section~\ref{sec:dix_tr}. We also need the
short notation
\[
\nabla(A_1)\cdot\nabla(A_2) := \sum_{j=1}^2\nabla_j(A_1)\nabla_j(A_2)
\]
for elements $A_1$, $A_2$ in the domain of the noncommutative gradient $\nabla$.
\begin{Theorem}[first Connes' formula]\label{theo:1st_Connes}
For every pair $A_1,A_2\in\bb{S}_B$ the following formulas hold true:
\begin{gather*}
{\Tr}_{{\rm Dix}}\big(\dd(\rho(A_1))^*\dd(\rho(A_2))\big) = \frac{2}{\ell_B^2} \fint_B\big(\nabla(A_1)^*\cdot\nabla(A_2)\big) ,\\
{\Tr}_{{\rm Dix}}\big(\chi\dd(\rho(A_1))^*\dd(\rho(A_2))\big) = 0 .
\end{gather*}
\end{Theorem}
\begin{proof}
Since $\rho(A_j)$ are operators of degree $0$ one has that
\[
\dd (\rho(A_1) )^*\dd (\rho(A_2) ) = [F_{B,\varepsilon},\rho(A_1) ]^* [F_{B,\varepsilon},\rho(A_2) ] .
\]
Therefore, in view of Lemma~\ref{lemma:pre_connes01} we only need to compute the Dixmier trace of the operators $I_0(\rho(A_1),\rho(A_2))$ and $\chi I_0(\rho(A_1),\rho(A_2))$. An explicit computation based on Corollary~\ref{cor:reg} provides
\begin{gather*}
 I_0(\rho(A_1),\rho(A_2)) = \sum_{i=1}^4\big(Q_{B,\varepsilon_i}^{-\frac{1}{2}}\nabla_1(A_1)^*\nabla_1(A_2) Q_{B,\varepsilon_i}^{-\frac{1}{2}}\big)\otimes\frac{\tau_{i,i}}{2\ell_B^2}\\
\hphantom{I_0(\rho(A_1),\rho(A_2)) =}{} +\sum_{i,j=1}^4\big(Q_{B,\varepsilon_i}^{-\frac{1}{2}}\nabla_1(A_1)^*\nabla_2(A_2) Q_{B,\varepsilon_j}^{-\frac{1}{2}}\big)\otimes\frac{\tau_{i,i}\gamma_1\gamma_2\tau_{j,j}}{2\ell_B^2}\\
\hphantom{I_0(\rho(A_1),\rho(A_2)) =}{} -\sum_{i,j=1}^4\big(Q_{B,\varepsilon_i}^{-\frac{1}{2}}\nabla_2(A_1)^*\nabla_1(A_2)Q_{B,\varepsilon_j}^{-\frac{1}{2}}\big) \otimes\frac{\tau_{i,i}\gamma_1\gamma_2\tau_{j,j}}{2\ell_B^2}\\
\hphantom{I_0(\rho(A_1),\rho(A_2)) =}{} +\sum_{i=1}^4\big(Q_{B,\varepsilon_i}^{-\frac{1}{2}}\nabla_2(A_1)^*\nabla_2(A_2)Q_{B,\varepsilon_i}^{-\frac{1}{2}}\big) \otimes\frac{\tau_{i,i}}{2\ell_B^2} ,
\end{gather*}
where the sign in the second line is due to the identity $\gamma_2\gamma_1=-\gamma_1\gamma_2$. In view of Theorem~\ref{prop:trace_dix_main res_bis} all the operators in the round brackets are measurable elements of the Dixmier ideal with trace which does not depend on the indices~$i$,~$j$. More precisely, one has that
\[
 {\Tr}_{{\rm Dix}}\big(Q_{B,\varepsilon_i}^{-\frac{1}{2}}\nabla_a(A_b)^*\nabla_c(A_d)Q_{B,\varepsilon_j}^{-\frac{1}{2}}\big) = \fint_B\big(\nabla_a(A_b)^*\nabla_c(A_d)\big)
\]
for all $a,b,c,d=1,2$ and independently of $i,j=1,\ldots,4$.
The consequence of this observation and of~\cite[Lemma~B.3]{denittis-gomi-moscolari-19} is that the computation of the Dixmier trace of
$I_0(\rho(A_1),\rho(A_2))$ only requires the computation of the trace of the matrices $\sum_{i=1}^4\tau_{i,i}={\bf 1}_4$ and
\[
 \sum_{i,j=1}^4\tau_{i,i}\gamma_1\gamma_2\tau_{j,j} = \gamma_1\gamma_2 = \ii\left(\begin{matrix}+1 & 0 & 0 & 0 \\0 & -1 & 0 & 0 \\0 & 0 & +1 & 0 \\0 & 0 & 0 & -1\end{matrix}\right) .
\]
This completes the proof of the first formula. The second formula is proved in
the same way with the only difference that the matrices $\chi$ and $\chi \gamma_1\gamma_2$ have a vanishing trace.
\end{proof}

It is worth ending this section with a couple of observations. First of all, it holds true that
\[
{\Tr}_{{\rm Dix}}\big(\big|\dd\big(\rho(A)\big)\big|^2\big) = \frac{2}{\ell_B^2} \big|\big|\big| |\nabla(A)|^2 \big|\big|\big|_{B,2}^2 \leqslant \frac{2}{\ell_B^2} \nnorm{A}_{B,1,2}^2 ,\qquad A\in\bb{S}_B,
\]
where on the right-hand side appears the Sobolev norm~\eqref{eq:sob_norm}. This is the starting point to extend the first Connes' formula to elements of the noncommutative Sobolev spaces $\rr{W}_B^{1,2}$. The second observation concerns the role of the trace per unit volume. From Remark~\ref{rk:trac_unit_volI} one gets
\[
 \frac{1}{2\pi}{\Tr}_{\text{Dix}}\big(\big|\dd\big(\rho(A)\big)\big|^2\big) = \s{T}_B\big(|\nabla(A)|^2\big) ,\qquad A\in\bb{S}_B .
\]
This is the (rigorous) \emph{continuous} version of the formula~\cite[equation~(47)]{bellissard-elst-schulz-baldes-94}.

\appendix
\section{Magnetic von Neumann algebra and distributional kernels}\label{appA}
The magnetic twisted convolution defined by~\eqref{eq:magn_conv_01} provides a product $\ast\colon S\big(\R^2\big)\times S\big(\R^2\big)\to S\big(\R^2\big)$ on the space of the Schwartz
functions. This product can be extended by continuity to the space of tempered
distributions $S'\big(\R^2\big)$. More precisely, if $( | )\colon S'\big(\R^2\big)\times S\big(\R^2\big)\to\C$
denotes the standard bilinear pairing between distributions and functions,\footnote{Let us
 recall that in the case of a distribution $\Psi\in S'\big(\R^2\big)\cap L^2\big(\R^2\big)$ the
 relation between pairing and scalar product is given by
 $(\Psi|f)=\langle\overline{\Psi},f\rangle_{L^2}$ for all $f\in S\big(\R^2\big)$.} one then defines
\[
 (\Psi\ast f|g) := (\Psi|f^-\ast g) ,\qquad \forall\, \Psi\in S'\big(\R^2\big) , \qquad \forall\, f,g\in S\big(\R^2\big) ,
\]
where $f^-(x):=f(-x)$. In this way the twisted convolution product extends to a
bilinear map $\ast\colon S'\big(\R^2\big)\times S\big(\R^2\big)\to S'\big(\R^2\big)$ (cf.\ \cite[Definition~2]{gracia-varilly-88}). This result can be made a bit more precise by writing
\begin{equation*}
\ast \colon \ S'\big(\R^2\big) \times S\big(\R^2\big) \longrightarrow \mathcal{O}_T\big(\R^2\big) \subset S'\big(\R^2\big),
\end{equation*}
where $\mathcal{O}_T\big(\R^2\big)$ is the set of smooth functions for which each derivative is polynomially bounded and the degree of the polynomial bound increases linearly with the order of the derivative~\cite[Theorem~3]{gracia-varilly-88}.

Let us introduce the space (cf.\ \cite[Section IV]{kammerer-86})
\[
 M\big(\R^2\big) := \left\{\Psi\in S'\big(\R^2\big) \left\vert
 \begin{array}{@{}l@{}}
 \Psi\ast f\in L^2\big(\R^2\big) \\
 \vspace{1mm}
 \|\Psi\ast f\|_{L^2}\leqslant C \lVert f \rVert_{L^2}\end{array}, \forall\, f\in S\big(\R^2\big)\right.\right\} ,
\]
which is, by definition, the set of distributions that define bounded operators on $L^2\big(\R^2\big)$. This space characterizes the magnetic von Neumann algebra $\bb{M}_B$.
\begin{Theorem}\label{theo:distr_kern}
$A\in \bb{M}_B$, if and only if, there is a $\Psi_A\in M\big(\R^2\big)$ such that
\[
 Af = \Psi_A\ast f ,\qquad \forall\, f\in S\big(\R^2\big) .
\]
\end{Theorem}
\begin{proof} If $\Psi_A\in M\big(\R^2\big)$, then the linear map $f\mapsto \Psi_A\ast f$, initially defined on
the dense domain $S\big(\R^2\big)\subset L^2\big(\R^2\big)$, extends continuously to a bounded
operator $A\in\bb{B}\big(L^2\big(\R^2\big)\big)$. One can check that
$\Psi_A\ast (V_B(a)f)=V_B(a)(\Psi_A\ast f)$ for every $a\in\R^2$ (cf.\ \cite[Theorem~2]{gracia-varilly-88}). This implies that by construction, $A$ commutes with the dual magnetic translations, i.e., $A\in \bb{V}_B'=\bb{M}_B$. On the other hand, let
$A\in \bb{M}_B$ and define $\Psi_A$ through the pairing
\[
 (\Psi_A|f) := 2\pi\ell^2_B \fint_B(A\pi(f)) ,\qquad f\in S\big(\R^2\big),
\]
where $\fint_B$ is the normal trace defined in Proposition~\ref{prop:FSN-trace}.
The definition of the pairing is well posed since $\pi(f)\in\bb{S}_B\subset\big(\bb{L}^2_B\big)^2$
which is the ideal of definition of the trace. The pairing is linear in view of
the linearity of the trace. Let $f=\sum_{n,m=0}^{\infty}f_{n,m}\psi_{n,m}$, where
$\{f_{n,m}\}\in{S}\big(\N_0^2\big)$ is the rapidly decreasing sequence associated to $f$
according to Proposition~\ref{prop:discret_sch}. Then
\begin{align*}
|(\Psi_A|f)| \leqslant 2\pi\ell^2_B \sum_{(n,m)\in\N^2_0}|f_{n,m}| \left|\fint_B(A\pi(\psi_{n,m}))\right|
 \leqslant \sqrt{2\pi}\ell_B \lVert A \rVert \sum_{(n,m)\in\N^2_0}|f_{n,m}| \fint_B(|\Upsilon_{m\mapsto n}|) .
\end{align*}
Observing that $|\Upsilon_{m\mapsto n}|^2=\Upsilon_{m\mapsto n}^*\Upsilon_{m\mapsto n}=\Pi_m=\Pi_m^2$, one gets
\[
 \fint_B(|\Upsilon_{m\mapsto n}|) = \fint_B(\Pi_m) = 1 ,
\]
in view of equation~\eqref{eq:trac_proj_olb} and in turn
\[
|(\Psi_A|f)| \leqslant \sqrt{2\pi}\ell_B \lVert A \rVert \sum_{(n,m)\in\N^2_0}|f_{n,m}| .
\]
Since
\begin{gather*}
\sum_{(n,m)\in\N^2_0}|f_{n,m}| = \sum_{(n,m)\in\N_0^2}\frac{(2n+1)(2m+1)}{(2n+1)(2m+1)} |f_{n,m}|
 \leqslant {\left(\sum_{n\in\N_0}\frac{1}{(2n+1)^2}\right)}^{2} r_2(\{f_{n,m}\}) ,
\end{gather*}
where $r_2$ is the semi-norm of $S\big(\R^2\big)$ defined by~\eqref{eq:fre_top}. It follows that
\[
|(\Psi_A|f)| \leqslant \left(\frac{\pi}{2}\right)^{\frac{5}{2}} \ell_B \lVert A \rVert r_2(\{f_{n,m}\}) .
\]
The latter inequality proves the continuity of the pairing with respect to the Fr\'echet topo\-lo\-gy~$S\big(\R^2\big)$, namely that $\Psi_A\in S'\big(\R^2\big)$. Consider now a pair $f,g\in S\big(\R^2\big)$. Then
\[
(\Psi_A\ast f\vert g) = (\Psi_A|f^-\ast g) = 2\pi\ell^2_B \fint_B(A\pi(f^-)\pi(g)) .
\]
Since $\pi(f^-)\in\bb{S}_B\subset \bb{L}^2_B$ and $\bb{L}^2_B$ is an ideal in $\bb{M}_B$
(cf.\ Proposition~\ref{prop:von_propert}), there exists a~$h_{A,f}\in L^2\big(\R^2\big)$
such that $L_{h_{A,f}}= A\pi(f^-)$. As a consequence
\[
(\Psi_A\ast f|g) = \langle J(h_{A,f}), g\rangle_{L^2} ,\qquad \forall\, g\in S\big(\R^2\big),
\]
which implies $\Psi_A\ast f= \overline{J(h_{A,f})}={h_{A,f}}^-\in L^2\big(\R^2\big)$. Moreover,
\[
\begin{aligned}
\|\Psi_A\ast f\|_{L^2}^2 &= \langle h_{A,f}, h_{A,f}\rangle_{L^2} = 2\pi\ell^2_B \fint_B(\pi(f^-)^*A^*A\pi(f^-))\\
&\leqslant \lVert A \rVert^2 2\pi\ell^2_B \fint_B(\pi(f^-)\pi(f^-)^*) = \lVert A \rVert^2 \lVert f \rVert_{L^2}^2 ,
\end{aligned}
\]
which shows that $\Psi_A\in M\big(\R^2\big)$. To conclude the proof let us observe that
\[
A\pi(f^-) g = A(f\ast g) = (Af)\ast g =: L_{(Af)^-}g
\]
for every $f, g\in S\big(\R^2\big)$. The second equality is a consequence of the
distributive property of the $\ast$-product on $S\big(\R^2\big)$ along with the fact that
$A$ can be approximated in the strong topology by a~sequence $\pi(a_n)$ with
$a_n\in S\big(\R^2\big)$. The last {equality} is justified by the fact that
$(Af)^-\in L^2\big(\R^2\big)$ is in the kernel of an element $L_{(Af)^-}\in \bb{L}^2_B$ Then,
it follows that $(\Psi_A\ast f)^-= h_{A,f}=(Af)^-$ and this concludes the proof.
\end{proof}

Let $A\in \bb{M}_B$ and $\Psi_A\in M\big(\R^2\big)$ the associated distribution according to
Theorem~\ref{theo:distr_kern}. Then we can represent $A$ as the
integral
\[
 A = \frac{1}{2\pi\ell_B^2}\int_{\R^2}\dd y \, \Psi_A(y) U_B(y) ,
\]
where $U_B(x)$ are the magnetic translations. Evidently the latter equation has a precise meaning only when evaluated on elements of the dense domain $S\big(\R^2\big)$ and out of this domain has to be considered as a formal expression.

Theorem~\ref{theo:distr_kern} provides a powerful tool to investigate the properties of the von Neumann algebra $\bb{M}_B$.
An interesting consequence is contained in the following result
\begin{Lemma}\label{lemma:append_A}
Let $S_1,S_2\in\bb{S}_B$ and $A\in\bb{M}_B$. Then $
S_1AS_2\in\bb{S}_B$.
\end{Lemma}
\begin{proof} Let $S_j=\pi(f_j)$, $j=1,2$, were $f_1,f_2\in S\big(\R^2\big)$. Then, in view of
 Theorem~\ref{theo:distr_kern}, The operator $S_1AS_2$ acts on the dense domain
 $S\big(\R^2\big)\subset L^2\big(\R^2\big)$ as a convolution operator with kernel
 $h:=f_1^-\ast\Psi_A\ast f_2^-$. To complete the proof one needs to show that
 $h\in S\big(\R^2\big)$. This can be done by exploiting the Hilbert spaces $\bb{G}_{s,t}$
 defined in~\cite[Definition~6]{gracia-varilly-88}. Let us recall that
 $\bb{G}_{s,t}\subseteq \bb{G}_{q,r}$ if and only if $s\geqslant q$, and $t\geqslant r$. We also have the
 equalities $S\big(\R^2\big)=\bigcap_{(s,t)\in\R^2}\bb{G}_{s,t}$, and
 $S'\big(\R^2\big)=\bigcup_{(s,t)\in\R^2}\bb{G}_{s,t}$. Moreover, one has the product relation
 $\bb{G}_{s,t}\ast \bb{G}_{q,r}\subset \bb{G}_{s,r}$ when $t+q\geqslant 0$~\cite[Theorem~8]{gracia-varilly-88} and the inclusion $M\big(\R^2\big)\subset \bigcup_{r>1}\bb{G}_{-r,0}$~\cite[Remarks on p.~886]{gracia-varilly-88-II} Using the associativity of the
 $\ast$-product one has that $h=f_1^-\ast t$ with
 $t:= \Psi_A\ast f_2^-\in \bb{G}_{-r,0}\ast\bb{G}_{q,t}\subset \bb{G}_{-r,t}$ for all $r>1$ and
 $t\in\R$ (it is enough to choose a $q>0$). Then
 $h\in \bb{G}_{s,r+\varepsilon}\ast\bb{G}_{-r,t} \subset \bb{G}_{s,t}$ for every $s,t\in\R$ (it is
 enough to choose a $\varepsilon>0$). Then $h\in \bigcap_{(s,t)\in\R^2}\bb{G}_{s,t}$ and this
 concludes the proof.
\end{proof}

\section{Technical results concerning the Dixmier trace}\label{appB}
Let us start with an improvement of Proposition~\ref{prop:very_inclus}.
\begin{Lemma}\label{lemm:very_inclus_impr}
The following chain of inclusions holds true
\begin{equation*}
\bb{S}_B \subset \big(\bb{L}^1_B\big)^2 \subset \bb{L}^1_B .
\end{equation*}
\end{Lemma}
\begin{proof}The first inclusion follows from $\bb{S}_B=(\bb{S}_B)^2$, which follows from Proposition~\ref{prop:struct_C_infty}, and the inclusion
$\bb{S}_B\subset \bb{L}^1_B$ is proved in Proposition~\ref{prop:very_inclus}. For the second inclusion consider a~pair of operators $L_{g_1},L_{g_2}\in \bb{L}^1_B$ with kernels $g_j:=\sum_{(n,m)\in\N_0^2}g_{n,m}^{(j)} \psi_{n,m}$, $j=1,2$.
The product $L_h:=L_{g_1}L_{g_2}$ has kernel $h:=g_1\ast g_2=\sum_{(n,m)\in\N_0^2}h_{n,m}
 \psi_{n,m}$ with coefficients given by
\[
 h_{n,m} := \frac{1}{\sqrt{2\pi}\ell_B}\sum_{r\in\N_0} g_{n,r}^{(1)}g_{r,m}^{(2)} .
\]
Since
\[
 \sum_{(n,m)\in\N_0^2}|h_{n,m}| \leqslant \frac{1}{\sqrt{2\pi}\ell_B}\left(\sum_{(n,r)\in\N_0^2}\left|g_{n,r}^{(1)}\right|\right)\left(\sum_{(r,m)\in\N_0^2}\left|g_{r,m}^{(2)}\right|\right)
\]
one gets that $\{h_{n,m}\}\in\ell^1\big(\N_0^2\big)$, and in turn
$L_h\in \bb{L}^1_B$.
\end{proof}

We are now ready to provide a proof of the result stated in Section~\ref{sec:dix_tr_2}.

\begin{proof}[Proof of Theorem~\ref{prop:trace_dix_main res_bis}]
Let $\{g_n\}\subset S\big(\R^2\big)$ be a sequence of Schwartz functions
such that $\|g_n-g\|_{L^2} \to 0$ as $n \to \infty$
and consider the sequence of operators $T_n:=L_f^*\pi(g_n)\in\bb{L}^1_B$
where the inclusion follows from Lemma~\ref{lemm:very_inclus_impr}.
Since $Q_{B,\varepsilon}^{-1}L_f^*$ is in the Dixmier ideal
in view of Proposition~\ref{prop:trace_dix_main res},
 also
$Q_{B,\varepsilon}^{-1}T$ and $Q_{B,\varepsilon}^{-1}T_n$ are in the Dixmier ideal.
Moreover, the symmetry property of the Calder\'on norm proved in~\cite[Proposition 7.16]{gracia-varilly-figueroa-01}
provides
\begin{gather*}
\|Q_{B,\varepsilon}^{-1}L_f^* (\pi(g_n)-L_{g} )\|_{1^+} \leqslant \|Q_{B,\varepsilon}^{-1}L_f^*\|_{1+} \|\pi(g_n)-L_g\|
\leqslant \frac{1}{\sqrt{2\pi}\ell_B}\|Q_{B,\varepsilon}^{-1}L_f^*\|_{1+} \|g_n-g\|_{L^2} ,
\end{gather*}
implying that $\big\|Q_{B,\varepsilon}^{-1} (T_n-T )\big\|_{1^+}\to0$ as $n \to \infty$. In
particular this shows that $Q_{B,\varepsilon}^{-1}T\in \rr{S}^{1^+}_{\rm m}$ since the
subspace $\rr{S}^{1^+}_{\rm m}$ is closed in the Calder\'on norm and
$Q_{B,\varepsilon}^{-1}T_n\in \rr{S}^{1^+}_{\rm m}$ in view of
Proposition~\ref{prop:trace_dix_main res}. Finally, since the Dixmier trace is
continuous with respect to the Calder\'on norm~\cite[p.~288]{gracia-varilly-figueroa-01}, it follows that
\begin{gather*}
{\Tr}_{\rm Dix}\big(Q_{B,\varepsilon}^{-1}T\big) = \lim_{n \to \infty} {\Tr}_{\rm Dix}\big(Q_{B,\varepsilon}^{-1}T_n\big) = \lim_{n \to \infty} (f^**g_n)(0) = \lim_{n \to \infty} \langle f,g_n\rangle_B = \langle f,g\rangle_B .
\end{gather*}
The equality with the trace $\fint_B$ is a consequence of
{Proposition}~\ref{prop:FSN-trace}.
\end{proof}

To push forward our analysis we need a hypothesis that, unfortunately, we are unable to prove in general.
\begin{Proposition}\label{prop_app_B03}
Let $L_g\in \bb{L}^2_B$ and assume that
\begin{equation}\label{eq:main_assum}
\big[Q_{B,\varepsilon}^{-1},L_g\big] \in \rr{S}^{1^+}_{0} .
\end{equation}
Let $T=L_{g}^*L_f$ with $L_f\in \bb{L}^1_B$. Then
$Q_{B,\varepsilon}^{-1}T\in\rr{S}^{1^+}_{\rm m}$ and the equality
\[
{\Tr}_{\rm Dix}\big(Q_{B,\varepsilon}^{-1}T\big) = \langle g,f \rangle_B = \fint_B(T)
\]
holds true, independently of $\varepsilon>-1$.
\end{Proposition}
\begin{proof} First of all let us observe that $L_g$ meets the condition~\eqref{eq:main_assum}, if and only if, $L_g^+$ meets the same condition. Then, the identity
\[Q_{B,\varepsilon}^{-1}T = \big[Q_{B,\varepsilon}^{-1},L_g^*\big]L_f + L_g^*Q_{B,\varepsilon}^{-1}L_f,\]
and Proposition~\ref{prop:trace_dix_main res} imply $Q_{B,\varepsilon}^{-1}T\in \rr{S}^{1^+}$. The linearity of the Dixmier trace along with the fact that $\rr{S}^{1^+}_{0}$ is an ideal
provides
\[
{\Tr}_{\rm Dix}\big(Q_{B,\varepsilon}^{-1}T\big) = {\Tr}_{\rm Dix}\big(L_g^*Q_{B,\varepsilon}^{-1}L_f\big) = {\Tr}_{\rm Dix}\big(Q_{B,\varepsilon}^{-1}L_fL_g^*\big) .
\]
The last equality follows from~\cite[Section~4.2.$\beta$, Proposition~3(b)]{connes-94}. Finally, using again Theorem~\ref{prop:trace_dix_main res_bis}, one gets
\[
{\Tr}_{\rm Dix}\big(Q_{B,\varepsilon}^{-1}L_fL_g^*\big) =
\langle \overline{f},\overline{g} \rangle_B = \langle {g},{f} \rangle_B =\fint_B(T),
\]
and this concludes the proof.
\end{proof}

Let us point out that the set of operators that meet the
condition~\eqref{eq:main_assum} is not empty since it is satisfied at least on
the subset $\bb{L}^1_B\subset \bb{L}^2_B$ in view of~\eqref{eq:null_dix_comm_II}. It
is also verified by the resolvent $R_\lambda\in \bb{L}^2_B\setminus \bb{L}^1_B$ of the Landau
Hamiltonian described in Section~\ref{sec:C-magn_alg}, and more in general, by
all the functions $f(H_B)\in \bb{L}^2_B$. On the other hand, we believe that
condition~\eqref{eq:main_assum} should be verified by all the elements of
$\bb{L}^2_B$. However, up to this point, we have been unable to
prove such a result. It is worth remarking that if one could prove the validity of
Proposition~\ref{prop_app_B03} for all the elements of $\bb{L}^2_B$, then one
could prove the validity of Theorem~\ref{prop:trace_dix_main res_bis} to the
domain $\bb{L}^1_B\cdot \bb{L}^2_B\cup \bb{L}^2_B\cdot \bb{L}^1_B$ which is
closed under taking the adjoint.

\subsection*{Acknowledgements}
GD's research is supported by the grant \texttt{Fondecyt Regular - 1190204}. MS's
research is supported by the grant \texttt{CONICYT-PFCHA Doctorado Nacional
 2018--21181868}. GD is indebted to Jean Bellissard, who is the real
inspirer of this work. GD would like to cordially thank Chris Bourne,
Massimo Moscolari, and Hermann Schulz-Baldes for several inspiring discussions.
We would like to thank the anonymous referees for providing very useful suggestions which significantly improved the quality of this work.

\pdfbookmark{References}{ref}
\LastPageEnding

\end{document}